\documentclass{vldb}

\usepackage[ruled]{algorithm2e}
\usepackage{amsmath}
\usepackage{amsfonts}
\usepackage{epsfig}
\usepackage{graphics}
\usepackage{amssymb}
\usepackage{color}
\usepackage{comment}
\usepackage{paralist}
\usepackage{mathtools}
\usepackage{multirow}
\usepackage{epstopdf}
\usepackage{subfigure}
\usepackage{eepic}
\usepackage{stmaryrd}
\usepackage{textcomp}
\usepackage{balance}
\usepackage[inline]{enumitem}
\usepackage{ulem}

\begin{document}

%%%%%%%%%%%%%%%%%%%%%%%%%%%%%%%%%%%%%%%%%%%%%%%%%%%%%%%%%%%%%%%%%%%%%%%%%%%%%%
%%% Time-stamp: "2001-03-15 23:44:12 calvanes"
%%%%%%%%%%%%%%%%%%%%%%%%%%%%%%%%%%%%%%%%%%%%%%%%%%%%%%%%%%%%%%%%%%%%%%%%%%%%%%

%%%%%%%%%%%%%%%%%%%%%%%%%% General Math

\newcommand{\A}{\mathcal{A}} \newcommand{\B}{\mathcal{B}}
\newcommand{\C}{\mathcal{C}} \newcommand{\D}{\mathcal{D}}
\newcommand{\E}{\mathcal{E}} \newcommand{\F}{\mathcal{F}}
\newcommand{\G}{\mathcal{G}} \renewcommand{\H}{\mathcal{H}}
\newcommand{\I}{\mathcal{I}} \newcommand{\J}{\mathcal{J}}
\newcommand{\K}{\mathcal{K}} \renewcommand{\L}{\mathcal{L}}
\newcommand{\M}{\mathcal{M}} \newcommand{\N}{\mathcal{N}}
\renewcommand{\O}{\mathcal{O}} \renewcommand{\P}{\mathcal{P}}
\newcommand{\Q}{\mathcal{Q}} \newcommand{\R}{\mathcal{R}}
\renewcommand{\S}{\mathcal{S}} \newcommand{\T}{\mathcal{T}}
\newcommand{\U}{\mathcal{U}} \newcommand{\V}{\mathcal{V}}
\newcommand{\W}{\mathcal{W}} \newcommand{\X}{\mathcal{X}}
\newcommand{\Y}{\mathcal{Y}} \newcommand{\Z}{\mathcal{Z}}

%%%%%%%%%%%%%%%%%%%%%%%%%% Abbreviations

\newcommand{\setone}[2][1]{\set{#1\cld #2}}
\newcommand{\eset}{\emptyset}
\newcommand{\ol}[1]{\overline{#1}}                % overline
\newcommand{\ul}[1]{\underline{#1}}               % underline
\newcommand{\uls}[1]{\underline{\raisebox{0pt}[0pt][0.45ex]{}#1}}
%% ul with space between text and line

\newcommand{\ra}{\rightarrow}
\newcommand{\Ra}{\Rightarrow}
\newcommand{\la}{\leftarrow}
\newcommand{\La}{\Leftarrow}
\newcommand{\lra}{\leftrightarrow}
\newcommand{\Lra}{\Leftrightarrow}
\newcommand{\lora}{\longrightarrow}
\newcommand{\Lora}{\Longrightarrow}
\newcommand{\lola}{\longleftarrow}
\newcommand{\Lola}{\Longleftarrow}
\newcommand{\lolra}{\longleftrightarrow}
\newcommand{\Lolra}{\Longleftrightarrow}
\newcommand{\ua}{\uparrow}
\newcommand{\Ua}{\Uparrow}
\newcommand{\da}{\downarrow}
\newcommand{\Da}{\Downarrow}
\newcommand{\uda}{\updownarrow}
\newcommand{\Uda}{\Updownarrow}

%%%%%%%%%%%%%%%%%%%%%%%%%% Relations

\newcommand{\incl}{\subseteq}
\newcommand{\imp}{\rightarrow}
\newcommand{\deq}{\doteq}
\newcommand{\dleq}{\dot{\leq}}                   % dotted less equal

%%%%%%%%%%%%%%%%%%%%%%%%%% Spaces

\newcommand{\per}{\mbox{\bf .}}                  % period

\newcommand{\cld}{,\ldots,}                      % ,...,
\newcommand{\ld}[1]{#1 \ldots #1}                 % #1 ... #1
\newcommand{\cd}[1]{#1 \cdots #1}                 % #1 ... #1
\newcommand{\lds}[1]{\, #1 \; \ldots \; #1 \,}    % _#1_..._#1_
\newcommand{\cds}[1]{\, #1 \; \cdots \; #1 \,}    % _#1_..._#1_

\newcommand{\dd}[2]{#1_1,\ldots,#1_{#2}}             % x1,...,xn (da da)
\newcommand{\ddd}[3]{#1_{#2_1},\ldots,#1_{#2_{#3}}}  % xi1,...,xin (da da down)
\newcommand{\dddd}[3]{#1_{11}\cld #1_{1#3_{1}}\cld #1_{#21}\cld #1_{#2#3_{#2}}}
%%                                x_11,...,x_1n,...,x_m1,...,x_mn_m

\newcommand{\ldop}[3]{#1_1 \ld{#3} #1_{#2}}   % x1 #3...#3 xn
\newcommand{\cdop}[3]{#1_1 \cd{#3} #1_{#2}}   % x1 #3...#3 xn
\newcommand{\ldsop}[3]{#1_1 \lds{#3} #1_{#2}} % x1 _#3_..._#3_ xn
\newcommand{\cdsop}[3]{#1_1 \cds{#3} #1_{#2}} % x1 _#3_..._#3_ xn

%%%%%%%%%%%%%%%%%%%%%%%%%% Delimiters

\newcommand{\quotes}[1]{{\lq\lq #1\rq\rq}}
\newcommand{\set}[1]{\{#1\}}                      % set
\newcommand{\Set}[1]{\left\{#1\right\}}
\newcommand{\bigset}[1]{\Bigl\{#1\Bigr\}}
\newcommand{\bigmid}{\Big|}
\newcommand{\card}[1]{|{#1}|}                     % cardinality of a set
\newcommand{\Card}[1]{\left| #1\right|}
\newcommand{\cards}[1]{\sharp #1}
\newcommand{\sub}[1]{[#1]}
\newcommand{\tup}[1]{\langle #1\rangle}            % tuple
\newcommand{\Tup}[1]{\left\langle #1\right\rangle}

%%%%%%%%%%%%%%%%%%%%%%%%% Environment delimiters

\def\qed{\hfill{\qedboxempty}      % qed with empty box
  \ifdim\lastskip<\medskipamount \removelastskip\penalty55\medskip\fi}

\def\qedboxempty{\vbox{\hrule\hbox{\vrule\kern3pt
                 \vbox{\kern3pt\kern3pt}\kern3pt\vrule}\hrule}}

\def\qedfull{\hfill{\qedboxfull}   % qed with full box
  \ifdim\lastskip<\medskipamount \removelastskip\penalty55\medskip\fi}

\def\qedboxfull{\vrule height 4pt width 4pt depth 0pt}

\newcommand{\markfull}{\qedboxfull}
\newcommand{\markempty}{\qed}

%%% Local Variables: 
%%% mode: latex
%%% TeX-master: "main"
%%% save-place: t
%%% End: 

%%%%%%%%%% LOCAL MACROS %%%%%%%%%%

%\newcounter{banana}

%\newcommand{\impl}{\la}
%\newcommand{\limp}{\ra}

\newcommand{\assign}{:=}
\newcommand{\cert}[3]{\mathit{cert}(#1,#2,#3)}

\newcommand{\varpos}[2]{[#1]_{#2}}

%%%%%%%%%% GENERAL
\newcommand{\rt}[1]{\mathit{root}(#1)}
\newcommand{\diam}[2]{\mathit{diam}_{#1}(#2)}

\newcommand{\blank}{\sqcup}

\newcommand{\child}{\mathit{child}}
\newcommand{\pchild}{\mathit{parentorchild}}
\newcommand{\descend}{\mathit{descendant}}
\newcommand{\nroot}{\mathit{root}}
\newcommand{\nct}{\mathit{CT}}
\newcommand{\suc}{\mathit{succ}}
\newcommand{\same}{\mathit{same}}
\newcommand{\samecell}{\mathit{sameCell}}
\newcommand{\sameleaf}{\mathit{sameLeaf}}
\newcommand{\nskel}{s}
\newcommand{\ncell}{c}

\newcommand{\mi}[1]{\mathit{#1}}

\newcommand{\con}{\mathit{conf}}
\newcommand{\beg}{\mathit{begin}}

\newcommand{\enc}[1]{\mathit{enc}(#1)}

\newcommand{\troot}[1]{\mathit{root}(#1)}
\newcommand{\tnodes}[1]{\mathit{nodes}(#1)}
\newcommand{\critical}[1]{\mathit{critical}(#1)}

\newcommand{\fr}[1]{\mathit{frontier}(#1)}
\newcommand{\nonfr}[1]{\mathit{nonfrontier}(#1)}

\newcommand{\tw}[1]{\mathit{tw}(#1)}

\newcommand{\norm}[1]{\mathsf{N}(#1)}

\newcommand{\ovtabentrya}[5]{
  \multirow{3}{*}{#1}&
  \footnotesize {#2} &
  \footnotesize {#3} &
  \footnotesize {#4} &
  \footnotesize {#5}
}

\newcommand{\ovtabentryb}[4]{
  &
  {#1}&
  {#2}&
  {#3}&
  {#4}
}

\newcommand{\ovtabentryc}[4]{
  &
  \footnotesize {#1} &
  \footnotesize {#2} &
  \footnotesize {#3} &
  \footnotesize {#4}
}

\newcommand{\cmpitem}{\noindent\ensuremath{-}\xspace}

\def\naf{\ensuremath{\raise.17ex\hbox{\ensuremath{\scriptstyle\mathtt{\sim}}}}\xspace}
\def\nafd{\ensuremath{\raise.17ex\hbox{\ensuremath{\scriptstyle\mathtt{\sim}.}}}\xspace}

\newcommand{\bas}[3]{\mathsf{B}(#1,#2,#3)}
\newcommand{\base}[2]{\mathsf{B}(#1,#2)}
\newcommand{\skolem}[2]{\mathsf{S}(#1,#2)}
\newcommand{\unify}[2]{\mathsf{U}(#1,#2)}
\newcommand{\match}[2]{\mathit{heads}(#1,#2)}
\newcommand{\cunify}[3]{\mathsf{C}_{#1}(#2,#3)}
\newcommand{\cunifyy}[2]{\mathsf{C}_{#1}(#2)}
\newcommand{\lin}[2]{\mathsf{b}(#1,#2)}
\newcommand{\mgu}[2]{\mathsf{MGU}(#1,#2)}
\newcommand{\f}[1]{\mathsf{F}(#1)}
\newcommand{\ff}[1]{\textsf{ff}(#1)}

\newcommand{\nav}[1]{\mathit{nav}(#1)}

\newcommand{\DB}{\mathit{DB}} \newcommand{\wrt}[0]{w.r.t.\ }

\newcommand{\ifdirection}{``$\Leftarrow$'' {}}
\newcommand{\onlyifdirection}{``$\Rightarrow$'' {}}

\renewcommand{\emptyset}{\varnothing}

\newcommand{\qans}[3]{\mathsf{QAns}(#1,#2,#3)} % QAns(D,Sigma,q)

\newcommand{\relevent}[2]{\substack{#1,#2 \\ \twoheadrightarrow}} % QAns(D,Sigma,q)

%%%%%%%%%% DATABASES

\newcommand{\rel}[1]{\mathsf{#1}}
\newcommand{\attr}[1]{\mathit{#1}}
\newcommand{\const}[1]{\mathit{#1}}
\newcommand{\vett}[1]{\vec{#1}}

% extension of a relation
\newcommand{\ext}[2]{#1^{#2}}

%\newcommand{\tup}[1]{\langle #1 \rangle}

% Operator that returns the predicate to which a tuple belong
\newcommand{\pred}[1]{\mathit{pred}(#1)}

\newcommand{\atom}[1]{\underline{#1}}
\newcommand{\tuple}[1]{\mathbf{#1}}

\newcommand{\parent}[1]{\mathit{parent}(#1)}
\newcommand{\dept}[1]{\mathit{depth}(#1)}

\newcommand{\mar}[1]{\hat{#1}}

%%%%%%%%%% DOMAINS

\newcommand{\dom}{\mathbf{C}}
\newcommand{\freshdom}{\mathbf{N}}
\newcommand{\variables}{\mathbf{V}}
\newcommand{\adom}[1]{\mathit{dom}(#1)}
\newcommand{\var}[1]{\mathit{var}(#1)}
\newcommand{\cons}[1]{\mathit{const}(#1)}
\newcommand{\term}[1]{\mathit{terms}(#1)}
\newcommand{\vari}[2]{\mathit{var}_{#1}(#2)}

\newcommand{\aff}[1]{\mathit{affected}(#1)}
\newcommand{\nonaff}[1]{\mathit{nonaffected}(#1)}
\newcommand{\sch}[1]{\mathit{sch}(#1)}
\newcommand{\edb}[1]{\mathit{edb}(#1)}

%%%%%%%%%% CONCEPTUAL MODEL

\newcommand{\entities}{\mathit{Ent}}
\newcommand{\relationships}{\mathit{Rel}}
\newcommand{\attributes}{\mathit{Att}}
\newcommand{\symbols}{\mathit{Sym}}

%%%%%%%%%% DEPENDENCIES

\newcommand{\dep}{\Sigma}
\newcommand{\tdep}{\Sigma_T}
\newcommand{\edep}{\Sigma_E}
\newcommand{\fdep}{\Sigma_F}
\newcommand{\kdep}{\Sigma_K}
\newcommand{\ndep}{\Sigma_\bot}

\newcommand{\key}[1]{\mathit{key}(#1)}
\newcommand{\isa}[1]{\mathit{ISA}}
\newcommand{\no}{\textrm{not}}

%%%%%%%%%% CONJUNCTIVE QUERIES

\newcommand{\Var}[1]{\mathit{Var}(#1)}
\newcommand{\head}[1]{\mathit{head}(#1)}
\newcommand{\heads}[1]{\mathsf{heads}(#1)}
\newcommand{\body}[1]{\mathit{body}(#1)}
\newcommand{\pbody}[1]{\mathit{B}^{+}(#1)}
\newcommand{\nbody}[1]{\mathit{B}^{-}(#1)}
\newcommand{\arity}[1]{\mathit{ar}(#1)}
\newcommand{\conj}[1]{\mathit{conj}(#1)}

\newcommand{\dvar}[1]{\mathit{dvar}(#1)}
\newcommand{\rset}[2]{\mathit{rset}(#1,#2)}

\newcommand{\cover}[1]{\mathit{cover}(#1)}

\newcommand{\ans}[3]{\mathit{ans}(#1,#2,#3)}
\newcommand{\pos}[2]{\mathit{pos}(#1,#2)}

%%%%%%%%%% SETS OF VARIABLES

\newcommand{\ins}[1]{\mathbf{#1}}
\newcommand{\insA}{\ins{A}}
\newcommand{\insB}{\ins{B}}
\newcommand{\insC}{\ins{C}}
\newcommand{\insD}{\ins{D}}
\newcommand{\insX}{\ins{X}}
\newcommand{\insY}{\ins{Y}}
\newcommand{\insZ}{\ins{Z}}
\newcommand{\insK}{\ins{K}}
\newcommand{\insU}{\ins{U}}
\newcommand{\insT}{\ins{T}}
\newcommand{\insO}{\ins{O}}
\newcommand{\insW}{\ins{W}}
\newcommand{\insN}{\ins{N}}
\newcommand{\insV}{\ins{V}}
\newcommand{\insS}{\ins{S}}

%%%%%%%%%% CHASE STUFF

\newcommand{\chase}[2]{\mathit{chase}(#1,#2)}
\newcommand{\ochase}[2]{\mathit{Ochase}(#1,#2)}
\newcommand{\rchase}[2]{\mathit{Rchase}(#1,#2)}
\newcommand{\instant}{\mathit{inst}}
%\newcommand{\base}{\mathit{base}}

% chase up to derivation depth #1:
\newcommand{\pchase}[3]{\mathit{chase}^{#1}(#2,#3)}
% chase up #1 application of the chase step:
\newcommand{\apchase}[3]{\mathit{chase}^{[#1]}(#2,#3)}
\newcommand{\aprchase}[3]{\mathit{Rchase}^{[#1]}(#2,#3)}
\newcommand{\level}[1]{\mathit{level}(#1)}
\newcommand{\freeze}[1]{\mathit{fr}(#1)}
\newcommand{\mods}[2]{\mathit{mods}(#1,#2)}
\newcommand{\subs}[2]{\gamma_{#1,#2}}

%%%%%%%%%% PARAGRAPH

\renewcommand{\paragraph}[1]{\textbf{#1}}
\newenvironment{proofsk}{\textsc{Proof (sketch).}}{\hfill$\square$\newline}
\newenvironment{pf}{\textsc{Proof.}}{\hfill$\square$\newline}
\newenvironment{proofrsk}{\textsc{Proof (rough sketch).}}{$\square$\newline}
\newenvironment{proofidea}{\textsc{Proof idea.}}{$\square$\newline}
\newenvironment{proofcustom}[1]{\textsc{Proof #1.}}{\hfill$\square$\newline}

%%%%%%%%%% TGDS

\newcommand{\dg}[2]{\mathit{DG}(#1,#2)}
\newcommand{\rank}[1]{\mathit{rank}(#1)}

\newcommand{\EXP}{{\scshape ExpTime}}
\newcommand{\PTIME}{\textsc{P}}
\newcommand{\NP}{{\scshape NP}}
\newcommand{\co}{{co}}
\newcommand{\LOGSPACE}{\textsc{LogSpace}}
\newcommand{\AC}[1]{\textsc{$\mbox{AC}_{#1}$}}
\newcommand{\NC}[1]{\textsc{$\mbox{NC}_{#1}$}}
\newcommand{\LOGCFL}{\textsc{LogCFL}}
\newcommand{\SIGMA}[2]{$\Sigma_{\textrm{#2}}^{\textrm{#1}}$}
\newcommand{\PI}[2]{$\Pi_{\textrm{#2}}^{\textrm{#1}}$}
\newcommand{\PSPACE}{\textsc{PSpace}}
\newcommand{\LINSPACE}{\textsc{LinSpace}}
\newcommand{\EXPSPACE}{\textsc{ExpSpace}}
\newcommand{\NEXP}{\textsc{NExpTime}}

\newcommand{\datalogpm}{Datalog$^\pm$}

\newcommand{\dllite}[2]{\textit{DL-Lite$^{\cal #1}_{#2}$}}

% pierre
\newcommand{\dinclusion}{DIDs}
\newcommand{\certain}{certain satisfaction}
\newcommand{\twoexptime}{2\EXP}
\newcommand{\exptime}{\EXP}
\newcommand{\exspace}{\EXPSPACE}
\newcommand{\Or}{\mathrm{Or}}
\newcommand{\False}{\mathrm{False}}
\newcommand{\True}{\mathrm{True}}
\newcommand{\mis}{monadic disjunctive inclusion dependencies}
\newcommand{\dfd}{disjunctive full dependencies}

\newcommand{\Skolem}[1]{{\mathcal{S}(#1)}}

\newcommand{\TODO}[1]{\textbf{[TODO:} #1\textbf{]}}
\newcommand{\malvi}[2]{{\color{red}\sout{#1} #2}}

%%% Local Variables:
%%% mode: latex
%%% TeX-master: "main"
%%% End:

\newcommand{\SA}{{\sf SemAc}}
\newcommand{\FSA}{{\sf FinSemAc}}
\newcommand{\ESA}{{\sf SemAc}$^{=}$}
\newcommand{\FESA}{{\sf FinSemAc}$^{=}$}
\newcommand{\saeval}{{\sf SemAcEval}}
\newcommand{\eval}{{\sf Eval}}
\newcommand{\cont}{{\sf Cont}}
\newcommand{\coeval}{{\sf coEval}}
\newcommand{\cocont}{{\sf coCont}}
\newcommand{\RCont}{{\sf RestCont}}
\newcommand{\ABCont}{{\sf AcBoolCont}}
\newcommand{\ac}{\mathbb{AE}}
\newcommand{\full}{\mathbb{F}}
\newcommand{\guarded}{\mathbb{G}}
\newcommand{\linear}{\mathbb{L}}
\newcommand{\id}{\mathbb{ID}}
\newcommand{\nr}{\mathbb{NR}}
\newcommand{\sticky}{\mathbb{S}}
\newcommand{\class}[1]{\mathbb{#1}}
\newcommand{\rew}{{\sf UCQRew}} 
\newcommand{\dist}{{\sf Dist}}
\renewcommand{\tup}[1]{\ensuremath{{(#1)}}}
\renewcommand{\set}[1]{\ensuremath{\{#1\}}}
\newcommand{\ve}[1]{\ensuremath{\bar{#1}}}
\newcommand{\limpl}{\ensuremath{\rightarrow}}
\renewcommand{\land}{\ensuremath{\wedge}}
\newcommand{\liff}{\ensuremath{\leftrightarrow}}
\renewcommand{\lor}{\ensuremath{\vee}}
\newcommand{\ca}[1]{\ensuremath{\mathcal{#1}}}
\newcommand{\sche}[1]{\ensuremath{\mathbf{#1}}}
\newcommand{\mbb}[1]{\ensuremath{\mathbb{#1}}}
\newcommand{\names}[1]{\ensuremath{\mathrm{names}(#1)}}
\newcommand{\fk}[1]{\ensuremath{\mathfrak{#1}}}
\newcommand{\defequ}{\ensuremath{\Longleftrightarrow_{\mathit{df}}}}
\newcommand{\GTODO}[1]{\textcolor{red}{\textbf{[TODO: #1]}}}
\newcommand{\dec}[1]{\ensuremath{\llbracket #1 \rrbracket}}
\newcommand{\cqasinst}[1]{\ensuremath{[#1]}}
\newcommand{\nd}[1]{\ensuremath{\mathrm{nd}(#1)}}
\newcommand{\ndia}[1]{\ensuremath{\langle #1 \rangle}}
\newcommand{\nbox}[1]{\ensuremath{[#1]}}
\newcommand{\free}[1]{\ensuremath{\mathrm{free}(#1)}}
\newcommand{\ptrue}{\ensuremath{\mathsf{true}}}
\newcommand{\pfalse}{\ensuremath{\mathsf{false}}}

\newcommand\diamonddot{\ensuremath{\mathord{\tmp}}}
\newcommand\tmp{{%
    \setbox0\hbox{$\Diamond$}%
    \rlap{\hbox to \wd0{\hss\kern.012em\raisebox{.07\height}{\scalebox{1.2}{$\cdot$}}\hss}}\box0
}}

%
%%% Local Variables:
%%% fill-column: 72
%%% TeX-PDF-mode: t
%%% TeX-debug-bad-boxes: t
%%% TeX-master: "main.tex"
%%% TeX-parse-self: t
%%% TeX-auto-save: t
%%% reftex-plug-into-AUCTeX: t
%%% End:

\newcommand{\OMIT}[1]{}

\newtheorem{theorem}{Theorem}%[chapter]
\newtheorem{corollary}[theorem]{Corollary}
\newtheorem{proposition}[theorem]{Proposition}
\newtheorem{lemma}[theorem]{Lemma}
\newtheorem{claim}[theorem]{Claim}
\newtheorem{fact}[theorem]{Fact}
\newtheorem{apptheorem}{Theorem}[section]
\newtheorem{appcorollary}[apptheorem]{Corollary}
\newtheorem{appproposition}[apptheorem]{Proposition}%[section]
\newtheorem{applemma}[apptheorem]{Lemma}%[section]
\newtheorem{appclaim}[apptheorem]{Claim}%[section]
\newtheorem{appfact}[apptheorem]{Fact}%[section]

\newdef{definition}{Definition}
\newdef{example}{Example}
\newdef{appdefinition}{Definition}
\newdef{appexample}{Example}

\title{Containment for Rule-Based Ontology-Mediated Queries}

\numberofauthors{3}

\author{
% 1st. author
\alignauthor
Pablo Barcel\'{o}\\
       \affaddr{{Center for Semantic Web Research \&}}\\
       \affaddr{{DCC, University of Chile}}\\
       \email{{pbarcelo@dcc.uchile.cl}}
% 2nd. author
\alignauthor Gerald Berger\\
       \affaddr{{Institute of Information Systems}}\\
       \affaddr{{TU Wien}}\\
       \email{{gberger@dbai.tuwien.ac.at}}
% 3rd. author
\alignauthor Andreas Pieris\\
       \affaddr{{School of Informatics}}\\
       \affaddr{{University of Edinburgh}}\\
       \affaddr{~}\\
       \email{{apieris@inf.ed.ac.uk}}
}

%\author{
%% 1st. author
%\alignauthor
%Pablo Barcel\'{o}\\
%       \affaddr{{\small Center for Semantic Web Research \&}}\\
%       \affaddr{{\small DCC, University of Chile}}\\
%       \email{{\small pbarcelo@dcc.uchile.cl}}
%% 2nd. author
%\alignauthor Gerald Berger\\
%       \affaddr{{\small Inst. of Information Systems}}\\
%       \affaddr{{\small TU Wien}}\\
%       \email{{\small gberger@dbai.tuwien.ac.at}}
%% 3rd. author
%\alignauthor Andreas Pieris\\
%       \affaddr{{\small School of Informatics}}\\
%       \affaddr{{\small University of Edinburgh}}\\
%       \email{{\small apieris@inf.ed.ac.uk}}
%}

\maketitle

\sloppy

%\fontsize{10pt}{10.2pt}
%\selectfont

\begin{abstract}
Many efforts have been dedicated to identifying restrictions on ontologies expressed as tuple-generating dependencies (tgds), a.k.a.~existential rules, that lead to the decidability for the problem of answering ontology-mediated queries (OMQs). This has given rise to three families of formalisms: guarded, non-recursive, and sticky sets of tgds. In this work, we study the containment problem for OMQs expressed in such formalisms, which is a key ingredient for solving static analysis tasks associated with them. Our main contribution is the development of specially tailored techniques for OMQ containment under the classes of tgds stated above. This enables us to obtain sharp complexity bounds for the problems at hand, which in turn
allow us to delimitate its practical applicability.
We also apply our techniques to pinpoint the complexity of problems associated with two emerging applications of OMQ containment:
 distribution over components and UCQ rewritability of OMQs.
\end{abstract}

%A category including the fourth, optional field follows...
%\category{D.2.8}{Software Engineering}{Metrics}[complexity measures, performance measures]


%\keywords{ACM proceedings, \LaTeX, text tagging} % NOT required for Proceedings

\section{Introduction}\label{sec:introduction}

\noindent
{\bf Motivation and goals.}
The novel application of knowledge representation tools for handling incomplete and heterogeneous data is giving rise to a new field, recently coined as {\em knowledge-enriched data management} \cite{dagstuhl}. A crucial problem in this field is {\em ontology-based data access} (OBDA) \cite{PCGLR08}, which refers to the utilization of ontologies (i.e., sets of logical sentences) for providing a unified conceptual view of various data sources. Users can then pose their queries solely in the schema provided by the ontology, abstracting away from the specifics of the individual sources.
In OBDA, one interprets the ontology $\Sigma$ and
the user query $q$, which is typically a {\em union of conjunctive queries} (UCQ), or, equivalently, the expressions defined by the select-project-join-union operators of relational algebra, as two components of one composite query $Q = (\insS,\Sigma,q)$, known as {\em ontology-mediated query} (OMQ); $\mathbf{S}$ is called the {\em data schema}, indicating that $Q$ will be posed on databases over $\insS$~\cite{BCLW13}. Therefore, OBDA is often realized as the problem of answering OMQs.

Following recent work \cite{CaGK13,CaGP10a,CaGP12,GoOP14}, we focus on the case where the ontology is defined by a set of {\em tuple-generating dependencies} (tgds), a.k.a.~{\em existential rules} or {\em Datalog$^{\pm}$ rules}. Handling such OMQs implies new challenges for classical database tasks. Interestingly, some of these challenges are by now well-studied; most notably (a) {\em query evaluation}~\cite{BLMS11,CaGK13,CaGL12,CaGP12}: given an OMQ $Q = (\insS,\Sigma,q)$, a database $D$ over $\insS$, and a tuple of constants $\bar c$, does $\bar c$ belong to the evaluation of $q$ over every extension of $D$ that satisfies $\Sigma$, or, equivalently, is $\bar c$ a {\em certain answer} for $Q$ over $D$? and (b) {\em relative expressiveness}~\cite{BCLW13,GoPS16,GoRS14}: how does the expressiveness of OMQs compare to the one of other query languages?
Surprisingly, despite its prominence, no work to date has carried out an in-depth investigation of {\em containment} for OMQs based on tgds and UCQs.

Query containment is a fundamental static analysis task that amounts to
check if the evaluation of a query is always contained
in the evaluation of another query. Several database tasks crucially depend on the ability to check query containment; these include, e.g., query optimization, view-based query answering, querying incomplete databases, integrity checking, and implication of dependencies: cf. \cite{BDHS96,CKPS95,FV84,FFLS99,FLM99,IL84}.
%While in general checking containment between queries is difficult computationally,
%
%
%Checking containment between queries is difficult computationally.
%For instance, the containment problem is undecidable for
%queries expressible in first-order logic (FO), and, thus, for any relational
%language that contains FO (such as SQL). Decidability results
%can be obtained for syntactically restricted classes of FO formulas.
%
A particularly important instance of the containment problem
%In the context of databases, the most important such restriction
is the one defined by the class of CQs. It follows from the seminal work of
Chandra and Merlin~\cite{ChMe77} that CQ containment is polynomially equivalent to CQ evaluation, and thus \NP-complete. The \NP~upper bound is not affected if we consider UCQs~\cite{SY80}. This is seen as a positive result for practical applications that rely on UCQ containment, as the input (the two UCQs) is small. In addition, it shows a stark difference with more expressive relational query languages, e.g., relational algebra (or, equivalently, first-order logic), for which containment is undecidable.

%%work focuses on another crucial task for OMQs; namely, {\em containment}:
%This is formally defined as the problem of given two such OMQs $Q_1$ and $Q_2$ with data schema $\insS$,
%does $Q_1(D) \subseteq Q_2(D)$ hold for every
%%(finite)
%database $D$ over $\insS$ (where $Q(D)$ denotes the certain answers for $Q$ %over $D$)?

The main goal of this work is to understand up to which extend the good computational properties of UCQ containment discussed above can be leveraged to the containment problem for OMQs based on tgds and UCQs (simply called OMQs from now on). In particular, we want to understand which classes
of tgds guarantee the decidability of the problem, and, whenever this is the case, how can we obtain complexity bounds that are reasonable for practical purposes. We also want to understand what is the exact relationship between OMQ containment and evaluation for such classes.
Let us stress that, apart from the traditional applications of containment mentioned above, it has been recently shown that OMQ containment has applications on other important static analysis tasks for OMQs, namely, distribution over components~\cite{BP16}, and UCQ rewritability~\cite{BLW16}.

\medskip
\noindent
{\bf The context.}
As one might expect, when considered in its full generality, i.e., without any restrictions on the set of tgds, the OMQ containment problem is undecidable. To understand, on the other hand, which restrictions lead to decidability, we recall the two main reasons that render the general containment problem undecidable. These are:

\medskip

\noindent
{\em \smash{Undecidability of query evaluation:}} OMQ evaluation is, in general, undecidable~\cite{BeVa81}, and it can be reduced to OMQ containment. More precisely, OMQ containment is undecidable whenever query evaluation for at least one of the involved languages (i.e., the language of the left-hand or the right-hand side query) is undecidable.

\medskip

\noindent
{\em \smash{Undecidability of containment for Datalog:}} decidability of query evaluation does not ensure decidability of query containment. A prime example is Datalog, or, equivalently, the OMQ language based on {\em full} tgds. Datalog containment is undecidable~\cite{Shmu93}, and thus, OMQ containment is undecidable if the involved languages extend Datalog.

\begin{center}
\begin{table*}[t]
\centering
  \begin{tabular}{c||c|c}
    & \textbf{Arbitrary Arity} & \textbf{Bounded Arity}\\
    \hline\hline
    \rule{0pt}{4ex}
    \textbf{Linear} & \begin{tabular}{@{}c@{}} \PSPACE-c \\ \scriptsize{\PSPACE-c} \end{tabular} & \begin{tabular}{@{}c@{}} $\Pi_{2}^P$-c \\ \scriptsize{\NP-c} \end{tabular} \\

    \hline
    \rule{0pt}{4ex}
    \textbf{Sticky} & \begin{tabular}{@{}c@{}} {\rm co}\NEXP-c \\ \scriptsize{\EXP-c} \end{tabular} & \begin{tabular}{@{}c@{}} $\Pi_{2}^P$-c \\ \scriptsize{\NP-c} \end{tabular} \\
    \hline
    \rule{0pt}{4ex}
    \textbf{Non-recursive} & \begin{tabular}{@{}c@{}} in \EXPSPACE~and \text{\rm P}$^{\textsc{NEXP}}$-hard \\ \scriptsize{\NEXP-c} \end{tabular} & \begin{tabular}{@{}c@{}} in \EXPSPACE~and \text{\rm P}$^{\textsc{NEXP}}$-hard \\ \scriptsize{\NEXP-c} \end{tabular} \\

    \hline
    \rule{0pt}{4ex}
    \textbf{Guarded} & \begin{tabular}{@{}c@{}} 2\EXP-c \\ \scriptsize{2\EXP-c} \end{tabular} & \begin{tabular}{@{}c@{}} 2\EXP-c \\ \scriptsize{\EXP-c} \end{tabular} \\
  \end{tabular}
  \caption{Complexity of OMQ containment -- in small fonts, we recall the complexity of OMQ evaluation.}
  \label{tab:containment-complexity}
  \vspace{-2mm}
\end{table*}
\end{center}

\vspace{-7mm}

In view of the above observations, we focus on languages that (a) have a decidable query evaluation, and (b) do not extend Datalog. The main classes of tgds, which give rise to OMQ languages with the desirable properties, can be classified into three main families depending on the underlying syntactic restrictions: (i) {\em guarded} tgds~\cite{CaGK13}, which contain inclusion dependencies and linear tgds, (ii) {\em non-recursive} sets of tgds~\cite{FKMP05}, and (iii) {\em sticky} sets of tgds \cite{CaGP12}.

While the decidability of containment for the above OMQ languages can be established via translations into query languages with a decidable containment problem, such translations do not lead to optimal complexity upper bounds (details are given below). Therefore, the main goal of our paper is to develop specially tailored decision procedures for the containment problem under the OMQ languages in question, and ideally obtain precise complexity bounds. Our second goal is to exploit such techniques in the study of distribution over components and UCQ rewritability of OMQs.

\medskip
\noindent
{\bf Our contributions.} The complexity of OMQ containment for the languages in question is given in Table~\ref{tab:containment-complexity}. Using small fonts, we recall the complexity of OMQ evaluation in order to stress that containment is, in general, harder than evaluation. We divide our contributions as follows:

\medskip
\noindent
{\em \underline{\smash{Linear, non-recursive and sticky sets of tgds.}}}
The OMQ languages based on linear, non-recursive, and sticky sets of tgds share a useful property: they are {\em UCQ rewritable} (implicit in~\cite{GoOP14}), that is, an OMQ can be rewritten into a UCQ. This property immediately yields decidability for their associated containment problems, since UCQ containment is decidable~\cite{SY80}. However, the obtained complexity bounds are not optimal, since the UCQ rewritings are unavoidably very large \cite{GoOP14}. To obtain more precise bounds, we reduce containment to query evaluation, an idea that is often applied in query containment; see, e.g., \cite{ChMe77,CV97,SY80}.

Consider a UCQ rewritable OMQ language $\class{O}$. If $Q_1$ and $Q_2$ belong to $\class{O}$, both with data schema $\insS$, then we can establish a {\em small witness property}, which states that non-containment of $Q_1$ in $Q_2$ can be witnessed via a database over $\insS$ whose size is bounded by an integer $k \geq 0$, the maximal size of a disjunct in a UCQ rewriting of $Q_1$. For linear tgds, such an integer $k$ is polynomial, but for non-recursive and sticky sets of tgds it is exponential (implicit in~\cite{GoOP14}).
The above small witness property allows us to devise a simple non-deterministic algorithm, which makes use of query evaluation as a subroutine for checking non-containment of $Q_1$ in $Q_2$: guess a database $D$ over $\insS$ of size at most $k$, and then check if there is a certain answer for $Q_1$ over $D$ that is not a certain answer for $Q_2$ over $D$.
This algorithm allows us to obtain optimal upper bounds for OMQs based on linear and sticky sets of tgds; however, the exact complexity of OMQs based on non-recursive sets of tgds remains open:
\begin{itemize}%\itemsep-\parsep
\item For OMQs based on linear tgds, the problem is in \PSPACE, and in $\Pi_2^P$ if the arity is fixed. The \PSPACE-hardness is shown by reduction from query evaluation~\cite{JoKl84}, while the $\Pi_{2}^{P}$-hardness is inherited from~\cite{BiLW12}.

\item For OMQs based on sticky sets of tgds, the problem is in \text{\rm co}\NEXP, and in $\Pi_2^P$ if the arity of the schema is fixed. The \text{\rm co}\NEXP-hardness is shown by exploiting the standard tiling problem for the exponential grid, while the $\Pi_{2}^{P}$-hardness is inherited from~\cite{BiLW12}.

\item Finally, for OMQs based on non-recursive sets of tgds, containment is in \EXPSPACE~and hard for \text{\rm P}$^{\textsc{NEXP}}$, even for fixed arity. The lower bound is shown by exploiting a recently introduced tiling problem~\cite{EiLP16}.
\end{itemize}

We conclude that in all these cases OMQ containment is harder than evaluation, with one exception: the OMQs based on linear tgds over schemas of unbounded arity.
%, where both problems are \PSPACE-complete.

\medskip
\noindent
{\em \underline{\smash{Guarded tgds.}}} The OMQ language based on guarded tgds is not UCQ rewritable, which forces us to develop different tools to study its containment problem. Let us remark that guarded OMQs can be rewritten as guarded Datalog queries (by exploiting the translations devised in~\cite{BaBC13,GoRS14}), for which containment is decidable in \twoexptime~\cite{BKR15}. But, again, the known rewritings are very large \cite{GoRS14}, and hence the reduction of containment for guarded OMQs to containment for guarded Datalog does not yield optimal upper bounds.

To obtain optimal bounds for the problem in question, we exploit {\em two-way alternating parity automata on trees} (2WAPA)~\cite{CGKV88}. We first show that if $Q_1$ and $Q_2$ are guarded OMQs such that $Q_1$ is not contained in $Q_2$, then this is witnessed over a class of ``tree-like'' databases that can be represented as the set of trees accepted by a 2WAPA $\fk{A}$. We then build a 2WAPA $\fk{B}$ with exponentially many states that recognizes those trees accepted by $\fk{A}$ that represent witnesses to non-containment of $Q_1$ in $Q_2$. Hence, $Q_1$ is contained in $Q_2$ iff $\fk{B}$ accepts no tree. Since the emptiness problem for 2WAPA is feasible in exponential time in the number of states~\cite{CGKV88}, we obtain that containment for guarded OMQs is in 2\EXP. A matching lower bound, even for fixed arity schemas, follows from~\cite{BLW16}.

Similar ideas based on 2WAPA have been recently used to show that containment for OMQs based on expressive {\em description logics} (DLs)
is in \twoexptime~\cite{BLW16}. In the DL context, schemas consist only of unary and binary relations. Our automata construction, however, is different from the one in~\cite{BLW16} for two reasons: (a) we need to deal with higher arity relations, and (b) even for unary and binary relations, our OMQ language allows to express properties that are not expressible by the DL-based OMQ languages studied in~\cite{BLW16}.

\medskip
\noindent
{\em \underline{\smash{Combining languages.}}}
The above complexity results refer to the containment problem relative to a certain OMQ language $\class{O}$, i.e., both queries fall in $\class{O}$.
However, it is natural to consider the version of the problem where the involved OMQs fall in different languages.
Unsurprisingly, if the left-hand side query is expressed in a UCQ rewritable OMQ language (based on linear, non-recursive or sticky sets of tgds), we can use the algorithm that relies on the small witness property discussed above, which provides optimal upper bounds for almost all the considered cases (the only exception is the containment of sticky in non-recursive OMQs over schemas of unbounded arity).
Things are more interesting if the ontology of the left-hand side query is expressed using guarded tgds, while the ontology of the right-hand side query is not guarded. By exploiting automata techniques, we show that containment of guarded in non-recursive OMQs is in 3\EXP, while containment of guarded in sticky OMQs is in 2\EXP. We establish matching lower bounds, even over schemas of fixed arity, by refining techniques from~\cite{CV97}.

\medskip
\noindent
{\em \underline{\smash{Applications.}}}
Our techniques and results on containment for guarded OMQs can be applied to other important static analysis tasks, in particular, distribution over components and UCQ rewritability.

The notion of distribution over components has been introduced in~\cite{AKNZ14}, in the context of declarative networking, and it states that the answer to an OMQ $Q$ can be computed by parallelizing it over the (maximally connected) components of the database. If this is the case, then $Q$ can always be evaluated in a distributed and coordination-free manner.
The problem of deciding distribution over components for OMQs has been recently studied in~\cite{BP16}. However, the exact complexity of the problem for guarded OMQs has been left open. By exploiting our results on containment, we can show that it is 2\EXP-complete.

It is well-known that the OMQ language based on guarded tgds is not UCQ rewritable. In view of this fact, it is important to study when a given guarded OMQ $Q$ can be rewritten as a UCQ. This has been studied for OMQs based on central Horn DLs~\cite{BLW16,BiLW13}. Interestingly, our automata-based techniques for guarded OMQ containment can be adapted to decide in 2\EXP~whether an OMQ based on guarded tgds over unary and binary relations is UCQ rewritable; a matching lower bound is inherited from~\cite{BLW16}. Our result generalizes the result that deciding UCQ rewritability for OMQs based on $\ca{ELHI}$, one of the most expressive members of the $\ca{EL}$-family of DLs, is 2\EXP-complete~\cite{BLW16}.

\medskip
\noindent
{\bf Discussion on Applicability.}
As shown in Table \ref{tab:containment-complexity}, the containment problem for OMQs based on linear sets of tgds is \PSPACE-complete, and thus
can be solved in single-exponential time. This is not a big practical drawback since the containment problem corresponds to a static analysis task. In fact, the runtime is single exponential only in the size of the UCQs and the maximum arity of the underlying schema, which are typically very small.
For such tasks, a single-exponential time procedure is considered to be acceptable, and it is actually the norm in many cases including database and verification problems; see, e.g., \cite{AbHV95,MZ09,RV01}.
%Moreover, in both cases the problem becomes NP-complete
%if the arity of the schema is fixed; i.e., it has the same complexity than the containment problem for UCQs.

For OMQs based on sticky, non-recursive and guarded sets of tgds, the containment problem becomes {\rm co}\NEXP-complete, \text{\rm P}$^{\textsc{NEXP}}$-hard and 2\EXP-complete,
respectively. This means that we require double-exponential time to solve the problem, which is practically not acceptable.
Nevertheless, for sticky sets of tgds, the runtime is double-exponential only in the maximum arity of the schema,
%and single exponential only in the size of the UCQs.
%
while for guarded sets of tgds is double-exponential only in the size of the UCQs and the maximum arity of the schema.
This is good news since, as said above, the size of the UCQs and the arity are typically small, and usually UCQs in OMQs are much smaller than the ontologies.

For non-recursive sets of tgds, on the other hand, the runtime is double-exponential, not only in the maximum arity, but also in the number of predicates occurring in the ontology. It is unrealistic to assume that the number of predicates occurring in real-life ontologies is small. This fact, together with the fact that the precise complexity of OMQ containment for non-recursive sets of tgds is still open, suggests that a more careful complexity analysis is needed. This is left as an interesting open problem for future work.

%Notice, however, that for sticky sets of tgds the problem becomes $\Pi_2^P$-complete (in particular, solvable in single-exponential time)
%if the arity of the schema is fixed. This is a reasonable assumption in practice.
%For non-recursive and guarded sets of tgds, on the other hand, the previous lower bounds hold even for
%fixed arity schemas.
%It is worth noticing, however, that our algorithms are double-exponential only the size of the UCQs, but
%not in the size of the ontologies (i.e., the sets of tgds in both OMQs).
%This is good, since often UCQs in OMQs are much smaller than the ontologies.
%A single-exponential time procedure can be considered to be ``reasonable'' for solving static analysis tasks, in particular,
%containment. In fact, such a running time is the norm in many static analysis and
%verification questions

\medskip
\noindent
{\bf Organization.} Preliminaries are  in Section~\ref{sec:preliminaries}. In Section~\ref{sec:containment-basics} we introduce the OMQ containment problem.
Containment for UCQ rewritable OMQs is studied in Section~\ref{sec:ucq-rewritability}, and for guarded OMQs in Section~\ref{sec:guardedness}. In Section~\ref{sec:different-languages} we consider the case where the involved queries fall in different languages. In Section~\ref{sec:applications} we discuss the applications of our results on guarded OMQ containment and we conclude in Section~\ref{sec:conclusions}.
Proofs and additional details can be found in the appendix.

%%% Local Variables:
%%% fill-column: 72
%%% TeX-PDF-mode: t
%%% TeX-debug-bad-boxes: t
%%% TeX-master: "main.tex"
%%% TeX-parse-self: t
%%% TeX-auto-save: t
%%% reftex-plug-into-AUCTeX: t
%%% End:

\section{Preliminaries}\label{sec:preliminaries}

\noindent
\paragraph{Databases and conjunctive queries.}
Let $\insC$, $\insN$, and $\insV$ be disjoint countably infinite sets of {\em constants}, {\em (labeled) nulls} and (regular) {\em variables} (used in queries and dependencies), respectively. A {\em schema} $\insS$ is a finite set of relation symbols (or predicates) with associated arity. We write $R/n$ to denote that $R$ has arity $n$. A {\em term} is a either a constant, null or variable. An {\em atom} over $\insS$ is an expression of the form $R(\bar v)$, where $R \in \insS$ is of arity $n > 0$ and $\bar v$ is an $n$-tuple of terms. A {\em fact} is an atom whose arguments consist only of constants. An {\em instance} over $\insS$ is a (possibly infinite) set of atoms over $\insS$ that contain constants and nulls, while a {\em database} over $\insS$ is a finite set of facts over $\insS$. We may call an instance and a database over $\insS$ an $\insS$-instance and $\insS$-database, respectively. The {\em active domain} of an instance $I$, denoted $\adom{I}$, is the set of all terms occurring in $I$.

A {\em conjunctive query} (CQ) over $\insS$ is a formula of the form:
\begin{equation}
\label{eq:cq}
q(\bar x) \ :=
\
\exists \bar y \big(R_1(\bar v_1) \wedge \dots \wedge R_m(\bar
v_m)\big),
\end{equation}
where each $R_i(\bar v_i)$ ($1 \leq i \leq m$) is an atom without nulls
over $\insS$, each variable mentioned in the $\bar v_i$'s
appears either in $\bar x$ or $\bar y$, and $\bar x$ are the free variables of $q$. If $\bar x$ is empty, then $q$ is a \emph{Boolean CQ}.
As usual, the evaluation of CQs is defined in terms of homomorphisms. Let $I$ be an instance and $q(\bar x)$ a CQ of the form~\eqref{eq:cq}. A {\em homomorphism} from $q$ to $I$ is a mapping $h$, which is the identity on $\insC$, from the variables that appear in $q$ to the set of constants and nulls $\insC \cup \insN$ such that $R_i(h(\bar v_i)) \in I$, for each $1 \leq i \leq m$. The {\em evaluation of $q(\bar x)$ over $I$}, denoted $q(I)$, is
the set of all tuples $h(\bar x)$ of constants such that $h$ is a homomorphism from $q$ to $I$. We denote by $\class{CQ}$ the class of conjunctive queries.

A {\em union of conjunctive queries} (UCQ) over $\insS$ is a formula of the form $q({\bar x}) := q_1({\bar x}) \vee \cdots \vee q_n({\bar x})$, where each $q_i({\bar x})$ is a CQ of the form~\eqref{eq:cq}. The {\em evaluation of $q(\bar x)$ over $I$}, denoted $q(I)$, is the set of tuples $\bigcup_{1 \leq i \leq n} q_i(I)$. We denote by $\class{UCQ}$ the class of union of conjunctive queries.

\medskip

\noindent
\paragraph{Tgds and the chase procedure.} A {\em tuple-generating dependency} (tgd) is a first-order sentence of the form:
\begin{equation}
\label{eq:tgd}
\forall \bar x \forall \bar y \big(\phi(\bar x,\bar y)
\rightarrow \exists \bar z\, \psi(\bar x,\bar z)\big),
\end{equation}
where $\phi$ and $\psi$ are conjunctions of atoms without nulls. For brevity, we write this tgd as $\phi(\bar x,\bar y) \rightarrow \exists \bar z\, \psi(\bar x,\bar z)$ and use comma instead of $\wedge$ for conjoining atoms. Notice that $\phi$ can be empty, in which case the tgd is called {\em fact tgd} and is written as $\top \ra \exists \bar z \, \psi(\bar x,\bar z)$.
We assume that each variable in $\bar x$ is mentioned in some atom of $\psi$. We call $\phi$ and $\psi$ the {\em body} and {\em head} of the tgd, respectively.
The tgd in \eqref{eq:tgd} is logically equivalent to the expression
$\forall \bar x (q_\phi(\bar x) \rightarrow q_\psi(\bar x))$, where $q_\phi(\bar x)$ and $q_\psi(\bar x)$ are the CQs $\exists \bar y\, \phi(\bar x,\bar y)$ and $\exists \bar z\, \psi(\bar x,\bar z)$, respectively. Thus, an instance $I$ over $\insS$ \emph{satisfies} this tgd iff $q_\phi(I) \subseteq q_\psi(I)$. We say that an instance $I$ satisfies a set $\Sigma$ of tgds, denoted $I \models \Sigma$, if $I$ satisfies every tgd in $\Sigma$. We denote by $\class{TGD}$ the class of (finite) sets of tgds.

The {\em chase} is a useful algorithmic tool when reasoning with tgds \cite{CaGK13,FKMP05,JoKl84,MaMS79}. We start by defining a single chase step.
Let $I$ be an instance over a schema $\insS$ and $\tau = \phi(\bar x,\bar y)
\rightarrow \exists \bar z \, \psi(\bar x,\bar z)$ a tgd over $\insS$.
We say that $\tau$ is \emph{applicable} w.r.t.~$I$ if there exists a tuple $(\bar a,\bar b)$ of terms in $I$ such that $\phi(\bar a,\bar b)$ holds in $I$. In this case, {\em the result of applying $\tau$ over $I$ with $(\bar a,\bar b)$} is the instance $J$ that extends $I$ with every atom in $\psi(\bar a,\bar \bot)$, where $\bar \bot$ is the tuple obtained by simultaneously replacing each variable $z \in \bar z$ with a fresh distinct null not occurring in $I$. For such a single chase step we write $I \xrightarrow{\tau,(\bar a,\bar b)} J$.

Let us assume now that $I$ is an instance and
$\Sigma$ a finite set of tgds. A {\em chase sequence for $I$ under $\Sigma$}
is a sequence:
\[
I_0\ \xrightarrow{\tau_0,\bar c_0}\ I_1 \xrightarrow{\tau_1,\bar c_1}\ I_2\ \cdots
\]
of chase steps such that: (1) $I_0 = I$; (2) for each $i \geq 0$, $\tau_i$ is a tgd in $\Sigma$; and (3) $\bigcup_{i \geq 0} I_i \models \Sigma$.
We call $\bigcup_{i \geq 0} I_i$ the {\em result} of this
chase sequence, which always exists.
Although the result of a chase sequence is not necessarily unique (up to isomorphism), each such result is equally useful for our purposes, since it can be homomorphically embedded into every other result. Thus, from now on, we denote by $\chase{I}{\Sigma}$ the result of an arbitrary chase
sequence for $I$ under $\Sigma$.

\medskip
\noindent
\paragraph{Ontology-mediated queries.}
An {\em ontology-mediated query} (OMQ) is a triple $(\insS,\dep,q)$, where $\insS$ is a schema, $\dep$ is a set of tgds (the ontology), and $q$ is a (U)CQ over $\insS \cup \sch{\dep}$ (and possibly other predicates), with $\sch{\dep}$ the set of predicates occurring in $\dep$.\footnote{OMQs can be defined for arbitrary first-order theories, not only tgds, and first-order queries, not only UCQs~\cite{BCLW13}.} We call $\insS$ the {\em data schema}. Notice that the set of tgds can introduce predicates not in $\insS$, which allows us to enrich the schema of the UCQ $q$. Moreover, the tgds can modify the content of a predicate $R \in \insS$, or, in other words, $R$ can appear in the head of a tgd of $\dep$. We have explicitly included $\insS$ in the specification of the OMQ to emphasize that it will be evaluated over $\insS$-databases, even though $\dep$ and $q$ might use additional relational symbols.

The semantics of an OMQ is given in terms of certain answers. The {\em certain answers} to a
UCQ $q({\bar x})$ w.r.t.~a database $D$ and a set $\dep$ of tgds is the set of tuples:
\[
\cert{q}{D}{\dep}\,\,\,\,\, =\,\,\,\,\, \ \bigcap_{\mathclap{I \supseteq D, I \models \dep}}\ \{{\bar c} \in \adom{I}^{|{\bar x}|} \mid {\bar c} \in q(I)\}.
\]
Consider an OMQ $Q = (\insS,\dep,q)$. The {\em evaluation} of $Q$ over an $\insS$-database $D$, denoted $Q(D)$, is defined as $\cert{q}{D}{\dep}$. It is well-known that $\cert{q}{D}{\dep} = q(\chase{D}{\dep})$; see, e.g.,~\cite{CaGK13}. Thus, $Q(D) = q(\chase{D}{\dep})$.

\medskip
\noindent
\paragraph{Ontology-mediated query languages.}
We write $(\class{C},\class{Q})$ for the OMQ language that consists of all OMQs of the form $(\insS,\dep,q)$, where $\dep$ falls in the class $\class{C}$ of tgds, i.e., $\class{C} \subseteq \class{TGD}$ (concrete classes of tgds are discussed below), and the query $q$ falls in $\class{Q} \in \{\class{CQ},\class{UCQ}\}$.
A problem that is quite important for our work is {\em OMQ evaluation}, defined as follows:

\begin{center}
\fbox{\begin{tabular}{ll}
{\small PROBLEM} : & $\eval(\class{C},\class{Q})$
\\
{\small INPUT} : & An OMQ $Q = (\insS,\dep,q({\bar x})) \in (\class{C},\class{Q})$,\\
& an $\insS$-database $D$, and ${\bar c} \in \adom{D}^{|{\bar x}|}$.
\\
{\small QUESTION} : &  Does ${\bar c} \in Q(D)$?
\end{tabular}}
\end{center}
It is well-known that $\eval(\class{TGD},\class{CQ})$ is undecidable; implicit in~\cite{BeVa81}. This has led to a flurry of activity for identifying syntactic restrictions on sets of tgds that make the latter problem decidable. Such a restriction defines a subclass $\class{C}$ of tgds. The known decidable classes of tgds are classified into three main decidability paradigms, which, in turn, give rise to decidable OMQ languages:

\medskip

\noindent
\underline{{\em Guardedness:}} A tgd is {\em guarded} if its
body contains an atom, called {\em guard}, that contains all the body-variables. Although the chase under guarded tgds does not necessarily terminate, the problem of deciding whether a tuple of constants is a certain answer to a UCQ w.r.t.~a database and a set of guarded tgds is decidable. This follows from the fact that the result of the chase has {\em bounded treewidth} (see, e.g., \cite{CaGK13}). Let $\guarded$ be the class of (finite) sets of guarded tgds. Then:

\begin{proposition} \label{prop:eval-guarded}{\em \cite{CaGK13}}
$\eval(\class{G},\class{CQ})$ and $\eval(\class{G},\class{UCQ})$ are 2\EXP-complete, and \EXP-complete for fixed arity.
\end{proposition}

An important subclass of guarded tgds is the class of {\em linear} tgds whose body consists of a single atom. We write $\linear$ for the class of (finite) sets of linear tgds.

\begin{proposition} \label{prop:eval-linear}{\em \cite{CaGL12,JoKl84}}
$\eval(\class{L},\class{CQ})$ and $\eval(\class{L},\class{UCQ})$ are \PSPACE-complete, and \NP-complete for fixed arity.
\end{proposition}

\noindent \underline{{\em \smash{Non-recursiveness:}}} A set
$\dep$ of tgds is {\em non-recursive} (a.k.a. {\em acyclic}~\cite{FKMP05,LMPS15}), if its predicate graph, the directed graph that encodes how the predicates of $\sch{\dep}$ depend on each other, is acyclic. Non-recursiveness ensures the termination of the chase, and thus decidability of OMQ evaluation. Let $\nr$ be the class of non-recursive (finite) sets of tgds. Then:

\begin{proposition} \label{prop:eval-nr}{\em \cite{LMPS15}}
$\eval(\nr,\class{CQ})$ and $\eval(\nr,\class{UCQ})$ are \NEXP-complete, even for fixed arity.
\end{proposition}

\begin{figure}[t]
 \epsfclipon
  \centerline
  {\hbox{
  \leavevmode
  \epsffile{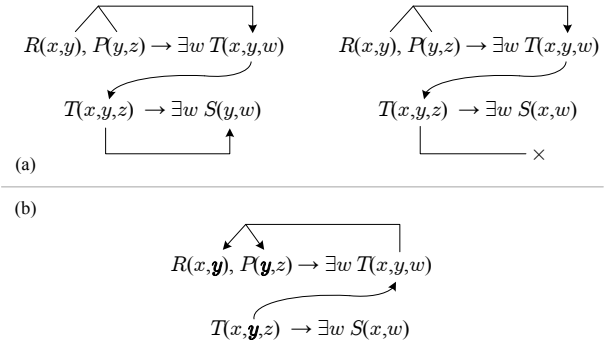}
  }} \epsfclipoff \caption{Stickiness and Marking.}
  \label{fig:stickiness}
  \vspace{-2mm}
\end{figure}

\noindent
\underline{\smash{\em Stickiness:}}
This condition ensures neither termination nor bounded treewidth of the chase. Instead, the decidability of OMQ evaluation is obtained by exploiting query rewriting techniques (more details on query rewriting are given in Section~\ref{sec:ucq-rewritability}).
The goal of stickiness is to capture joins among variables that are not expressible via guarded tgds, but without forcing the chase to terminate.
The key property underlying this condition can be described as follows: during the chase, terms that are associated (via a homomorphism) with variables that appear more than once in the body of a tgd (i.e., join variables) are always propagated (or ``stick'') to the inferred atoms. This is illustrated in Figure~\ref{fig:stickiness}(a); the left set of tgds is sticky, while the right set is not.
The formal definition is based on an inductive marking procedure that marks the variables that may violate the semantic property of the chase described above~\cite{CaGP12}. Roughly, during the base step of this procedure, a variable that appears in the body of a tgd $\tau$ but not in every head-atom of $\tau$ is marked. Then, the marking is inductively propagated from head to body as shown in Figure~\ref{fig:stickiness}(b).
Finally, a finite set of tgds $\Sigma$ is {\em sticky} if no tgd in $\Sigma$ contains two occurrences of a marked variable. Let $\class{S}$ be the class of sticky (finite) sets of tgds. Then:

\begin{proposition} \label{prop:eval-sticky}{\em \cite{CaGP12}}
$\eval(\sticky,\class{CQ})$ and $\eval(\sticky,\class{UCQ})$ are \EXP-complete, and \NP-complete for fixed arity.
\end{proposition}

%%% Local Variables:
%%% fill-column: 72
%%% TeX-PDF-mode: t
%%% TeX-debug-bad-boxes: t
%%% TeX-master: "main.tex"
%%% TeX-parse-self: t
%%% TeX-auto-save: t
%%% reftex-plug-into-AUCTeX: t
%%% End:

\section{OMQ Containment: The Basics}\label{sec:containment-basics}

The goal of this work is to study in depth the problem of checking whether an OMQ $Q_1$ is {\em contained} in an OMQ $Q_2$, both over the same data schema $\insS$, or, equivalently, whether $Q_1(D) \subseteq Q_2(D)$ over every (finite) $\insS$-database $D$. In this case we write $Q_1 \subseteq Q_2$; we write $Q_1 \equiv Q_2$ if $Q_1 \subseteq Q_2$ and $Q_2 \subseteq Q_1$. The {\em OMQ containment} problem in question is defined as follows; $\class{O}_1$ and $\class{O}_2$ are OMQ languages $(\class{C},\class{Q})$, where $\class{C}$ is a class of tgds (e.g., linear, non-recursive, sticky, etc.), and $\class{Q} \in \{\class{CQ}, \class{UCQ}\}$:

\begin{center}
\fbox{\begin{tabular}{ll}
{\small PROBLEM} : & $\cont(\class{O}_1,\class{O}_2)$
\\
{\small INPUT} : & Two OMQs $Q_1 \in \class{O}_1$ and $Q_2 \in \class{O}_2$.
\\
{\small QUESTION} : &  Does $Q_1 \subseteq Q_2$?
\end{tabular}}
\end{center}
Whenever $\class{O}_1 = \class{O}_2 = \class{O}$, we refer to the containment problem by simply writing $\cont(\class{O})$.

In what follows, we establish some simple but fundamental results, which help to better understand the nature of our problem.
We first investigate the relationship between evaluation and containment, which in turn allows us to obtain an initial boundary for the decidability of our problem, i.e., we can obtain a positive result only if the evaluation problem for the involved OMQ languages is decidable (e.g., those introduced in the previous section).
We then focus on the OMQ languages introduced in Section~\ref{sec:preliminaries} and observe that, once we fix the class of tgds, it does not make a difference whether we consider CQs or UCQs. In other words, we show that an OMQ in $(\class{C},\class{UCQ})$, where $\class{C} \in \{\class{G},\class{L},\class{NR},\class{S}\}$, can be rewritten as an OMQ in $(\class{C},\class{CQ})$. This fact simplifies our later complexity analysis since for establishing upper (resp., lower) bounds it suffices to focus on CQs (resp., UCQs).

\subsection{Evaluation vs. Containment}

As one might expect, OMQ evaluation and OMQ containment are strongly connected. In fact, as we explain below, the former can be easily reduced to the latter. But let us first introduce some auxiliary notation. Consider a database $D$ and a tuple ${\bar c} = (c_1,\ldots,c_n) \in \adom{D}^n$, where $n \geq 0$. We denote by $q_{D,{\bar c}}({\bar x})$, where ${\bar x} = (x_{c_1},\ldots,x_{c_n})$, the CQ obtained from the conjunction of atoms occurring in $D$ after replacing each constant $c$ with the variable $x_{c}$.
Consider now an OMQ $Q = (\insS,\dep,q({\bar x})) \in (\class{C},\class{CQ})$, where $\class{C}$ is some class of tgds, an $\insS$-database $D$, and a tuple ${\bar c} \in \adom{D}^{|{\bar x}|}$. It is not difficult to show that
\[
{\bar c} \in Q(D)
\iff
\underbrace{(\sch{\dep},\emptyset,q_{D,{\bar c}})}_{Q_1} \subseteq \underbrace{(\sch{\dep},\dep,q)}_{Q_2}\!.
\]
Let $\class{O}_{\emptyset}$ be the OMQ language that consists of all OMQs of the form $(\insS,\emptyset,q)$, i.e., the set of tgds is empty, where $q$ is a CQ. It is clear that $Q_1 \in \class{O}_{\emptyset}$ and $Q_2 \in (\class{C},\class{CQ})$. Therefore, for every OMQ language $\class{O} = (\class{C},\class{CQ})$, where $\class{C}$ is a class of tgds, we immediately get that:

\begin{proposition}\label{pro:eval-to-cont}
$\eval(\class{O})$ can be reduced in polynomial time into $\cont(\class{O}_{\emptyset},\class{O})$.
\end{proposition}

We now show that the problem of evaluation is reducible to the complement of containment.
Let us say that, for technical reasons which will be made clear in a while, we focus our attention on classes $\class{C}$ of tgds that are {\em closed under fact tgd extension}, i.e., for every set $\dep \in \class{C}$, a set obtained from $\dep$ by adding a (finite) set of fact tgds is still in $\class{C}$. This is not an unnatural assumption since every reasonable class of tgds, such as the ones introduced above, enjoy this property.
Consider now an OMQ $Q = (\insS,\dep,q({\bar x})) \in (\class{C},\class{CQ})$, where $\class{C}$ is some class of tgds, an $\insS$-database $D$, and a tuple ${\bar c} \in \adom{D}^{|{\bar x}|}$. It is easy to see that
\[
{\bar c} \in Q(D)
\iff
\underbrace{(\insS,\dep_{D}^{\star},q^{\star}_{\bar c})}_{Q_1} \not\subseteq \underbrace{(\insS,\emptyset,\exists x \, P(x))}_{Q_2},
\]
where $\dep_{D}^{\star}$ is obtained from $\dep$ by renaming each predicate $R$ in $\dep$ into $R^{\star} \not\in \insS$ and adding the set of fact tgds
\[
\{\top \ra R^{\star}(c_1,\ldots,c_k) \mid R(c_1,\ldots,c_k) \in D\},
\]
$q^{\star}_{\bar c}$ is obtained from $q(\bar c)$ by renaming each predicate $R$ into $R^{\star} \not\in \insS$, and the predicate $P$ does not occur in $\insS$. Indeed, the above equivalence holds since $P \not\in \insS$ implies that $Q_2(D) = \emptyset$, for every $\insS$-database $D$.
Since $\class{C}$ is closed under fact tgd extension, $Q_1 \in (\class{C},\class{CQ})$, while $Q_2 \in \class{O}_{\emptyset}$. We write $\cocont(\class{O}_1,\class{O}_2)$ for the complement of $\cont(\class{O}_1,\class{O}_2)$. Hence, for every OMQ language $\class{O} = (\class{C},\class{CQ})$, where $\class{C}$ is a class of tgds (closed under fact tgd extension), it holds that:

\begin{proposition}\label{pro:coeval-to-cocont}
$\eval(\class{O})$ can be reduced in polynomial time into $\cocont(\class{O},\class{O}_{\emptyset})$.
\end{proposition}

By definition, $\class{O}_{\emptyset}$ is contained in every OMQ language $(\class{C},\class{CQ})$, where $\class{C}$ is a class of tgds. Therefore, as a corollary of Propositions~\ref{pro:eval-to-cont} and~\ref{pro:coeval-to-cocont}, we obtain an initial boundary for the decidability of OMQ containment: we can obtain a positive result only if the evaluation problem for the involved OMQ languages is decidable. More formally:

\begin{corollary}\label{cor:undecidability}
$\cont(\class{O}_1,\class{O}_2)$ is undecidable if $\eval(\class{O}_1)$ is undecidable or $\eval(\class{O}_2)$ is undecidable.
\end{corollary}

Can we prove the converse of Corollary~\ref{cor:undecidability}: $\cont(\class{O}_1,\class{O}_2)$ is decidable if both $\eval(\class{O}_1)$ and $\eval(\class{O}_2)$ are decidable? The answer to this question is negative. This is due to the fact that containment of Datalog queries is undecidable~\cite{Shmu93}. Since Datalog queries can be directly encoded in the OMQ language based on the class $\class{F}$ of {\em full tgds}, i.e., those without existentially quantified variables, we obtain the following:

\begin{proposition}\label{prop:cont-full}{\em \cite{Shmu93}}
$\cont((\class{F},\class{CQ}))$ is undecidable.
\end{proposition}

This result, combined with the fact that $\eval(\class{F})$ is decidable (since the chase under full tgds always terminates), implies that the converse of Corollary~\ref{cor:undecidability} does not hold.
Proposition~\ref{prop:cont-full} also rules out the OMQ languages that are based on classes of tgds that extend $\class{F}$; e.g., the weak versions of the ones introduced in Section~\ref{sec:preliminaries}, called {\em weakly guarded}~\cite{CaGK13}, {\em weakly acyclic}~\cite{FKMP05}, and {\em weakly sticky}~\cite{CaGP12} that guarantee the decidability of OMQ evaluation.\footnote{The idea of those classes is the same: relax the conditions in the definition of the class, so that only those positions that receive null values during the chase are taken into account.}
The question that comes up concerns the decidability and complexity of containment for the OMQ languages that are based on the non-weak versions of the above classes, i.e., guarded, non-recursive, and sticky. This will be the subject of the next two sections.

\subsection{From UCQs to CQs}

Before we proceed with the complexity analysis of containment for the OMQ languages in question, let us state the following useful result:

\begin{proposition}\label{pro:from-ucq-to-cq}
Given an OMQ $Q \in (\class{C},\class{UCQ})$, where $\class{C} \in \{\class{G}, \class{L}, \class{NR}, \class{S}\}$, we can construct in polynomial time an OMQ $Q' \in (\class{C},\class{CQ})$ such that $Q \equiv Q'$.
\end{proposition}

The proof of Proposition~\ref{pro:from-ucq-to-cq} relies on the idea of encoding boolean operations (in our case the `or' operator) using a set of atoms; this idea has been used in several other works (see, e.g.,~\cite{BeGo10,BMMP16,GoPa03}).
Proposition~\ref{pro:from-ucq-to-cq} allows us to focus on OMQs that are based on CQs; in fact, $\cont((\class{C}_1,\class{CQ}),(\class{C}_2,\class{CQ}))$ is $\C$-complete, where $\class{C}_1,\class{C}_2 \in \{\class{G},\class{L},\class{NR},\class{S}\}$ and $\C$ is a complexity class that is closed under polynomial time reductions, iff $\cont((\class{C}_1,\class{UCQ}),(\class{C}_2,\class{UCQ}))$ is $\C$-complete.

\subsection{Plan of Attack}

We are now ready to proceed with the complexity analysis of containment for the OMQ languages in question. Our plan of attack can be summarized as follows:
\begin{itemize}\itemsep-\parsep
\item We consider, in Section~\ref{sec:ucq-rewritability}, $\cont((\class{C},\class{CQ}))$, for $\class{C} \in \{\class{L},\class{NR},\class{S}\}$. These languages enjoy a crucial property, called UCQ rewritability, which is very useful for our purposes. This property allows us to show the following result: if the containment does not hold, then this is witnessed via a ``small'' database, which in turn allows us to devise simple guess-and-check algorithms.
\item We then proceed, in Section~\ref{sec:guardedness}, with $\cont((\class{G},\class{CQ}))$. This OMQ language does not enjoy UCQ rewritability, and the task of establishing a small witness property as above turned out to be challenging. However, we show the following: if the containment does not hold, then this is witnessed via a ``tree-shaped'' database, which allows us to devise a decision procedure based on two-way alternating parity automata on finite trees.

\item In Section~\ref{sec:different-languages}, we study the case where the OMQ containment problem involves two different languages. If the left-hand side language is UCQ rewritable, then we can devise a guess-and-check algorithm by exploiting the above small witness property. The challenging case is when the left-hand side language is $(\class{G},\class{CQ})$, where again we employ techniques based on tree automata.
\end{itemize}

%%% Local Variables:
%%% fill-column: 72
%%% TeX-PDF-mode: t
%%% TeX-debug-bad-boxes: t
%%% TeX-master: "main.tex"
%%% TeX-parse-self: t
%%% TeX-auto-save: t
%%% reftex-plug-into-AUCTeX: t
%%% End:

\section{UCQ Rewritable Languages}\label{sec:ucq-rewritability}

We now focus on OMQ languages that enjoy the crucial property of UCQ rewritability.

\begin{definition}\label{def:ucq-rewritability}(\textbf{UCQ Rewritability})
An OMQ language $(\class{C},\class{CQ})$, where $\class{C} \subseteq \class{TGD}$, is {\em UCQ rewritable} if, for each OMQ $Q = (\insS,\dep,q({\bar x})) \in (\class{C},\class{CQ})$  we can construct a UCQ $q'({\bar x})$ such that $Q(D) = q'(D)$ for every $\insS$-database $D$. \hfill\markfull
\end{definition}

We proceed to establish our desired small witness property, based on UCQ rewritability. By the definition of UCQ rewritability, for each language $\class{O}$ that is UCQ rewritable, there exists a computable function $f_{\class{O}}$ from $\class{O}$ to the natural numbers such that the following holds: for every OMQ $Q = (\insS,\dep,q({\bar x})) \in \class{O}$, and UCQ rewriting $q_1({\bar x}) \vee \cdots \vee q_n({\bar x})$ of $Q$, it is the case that $\max_{1 \leq i \leq n} \{|q_i|\} \leq f_{\class{O}}(Q)$, where $|q_i|$ denotes the number of atoms occurring in $q_i$. Then:

\begin{proposition}\label{pro:small-witness-property}
Consider a UCQ rewritable language $\class{O}$, and two OMQs $Q \in \class{O}$ and $Q' \in (\class{TGD},\class{CQ})$, both with data schema $\insS$. If $Q \not\subseteq Q'$, then there exists an $\insS$-database $D$, where $|D| \leq f_{\class{O}}(Q)$, such that $Q(D) \not\subseteq Q'(D)$.
\end{proposition}

%\begin{proofsk}
%We assume that $q({\bar x}) = \bigvee_{i=1}^{n} q_i({\bar x})$ is a UCQ rewriting of $Q$. Since $Q \not\subseteq Q'$, we conclude that $q \not\subseteq Q'$, which in turn implies that there exists an $i \in \{1,\ldots,n\}$ such that $q_i \not\subseteq Q'$. It is easy to show that $c({\bar x}) \not\in Q'(D_{q_i})$, where $c({\bar x})$ is a tuple of constants obtained by replacing each variable $x$ in ${\bar x}$ with the constant $c(x)$, and $D_{q_i}$ is the $\insS$-database obtained from $q_i$ after replacing each variable $x$ in $q_i$ with the constant $c(x)$. Since $c({\bar x}) \in q(D_{q_i})$, we get that $c({\bar x}) \in Q(D_{q_i})$. Therefore, $Q(D_{q_i}) \not\subseteq Q'(D_{q_i})$, and the claim follows since $|D_{q_i}| \leq f_{\class{O}}(Q)$.
%\end{proofsk}

In Proposition~\ref{pro:small-witness-property} we assume that the left-hand side query falls in a UCQ rewritable language, be we do not pose any restriction on the language of the right-hand side query. Thus, we immediately get a decision procedure for $\cont(\class{O}_1,\class{O}_2)$ if $\class{O}_1$ is UCQ rewritable and $\eval(\class{O}_2)$ is decidable. Given $Q_1 = (\insS, \dep_1, q_1({\bar{x}})) \in \class{O}_1$ and $Q_2 = (\insS, \dep_2, q_2({\bar x})) \in \class{O}_2$:
\begin{enumerate}
\item Guess an $\insS$-database $D$ such that $|D| \leq f_{\class{O}_1}(Q_1)$, and a tuple ${\bar c} \in \adom{D}^{|{\bar x}|}$; and

\item Verify that ${\bar c} \in Q_1(D)$ and ${\bar c} \not\in Q_2(D)$.
\end{enumerate}
We immediately get that:

\begin{theorem}\label{the:ucq-rewritable-languages-decidability}
$\cont(\class{O}_1,\class{O}_2)$ is decidable if $\class{O}_1$ is UCQ rewritable and $\eval(\class{O}_2)$ is decidable.
\end{theorem}
This general result shows that $\cont((\class{C},\class{CQ}))$ is decidable for every $\class{C} \in \{\class{L},\class{NR},\class{S}\}$, but it says nothing about its complexity. This will be the subject of the rest of the section.

\subsection{Linearity}

The problem of computing UCQ rewritings for OMQs in $(\class{L},\class{CQ})$ has been studied in~\cite{GoOP14}, where a resolution-based procedure, called $\mathsf{XRewrite}$, has been proposed. This rewriting algorithm accepts a query $Q = (\insS,\dep,q({\bar x})) \in (\class{L},\class{CQ})$ and constructs a UCQ rewriting $q'({\bar x})$ over $\insS$ by starting from $q$ and exhaustively applying rewriting steps based on resolution. Let us illustrate this via a simple example:

\begin{example}
Assume that $\insS = \{P,T\}$. Consider the set $\dep$ consisting of the linear tgds
\[
P(x) \ra \exists y \, R(x,y), \qquad R(x,y) \ra P(y), \qquad T(x) \ra P(x),
\]
and the CQ $q({\bar x}) \coloneqq \exists y (R(x,y) \wedge P(y))$.
$\mathsf{XRewrite}$ will first resolve the atom $P(y)$ in $q$ using the second tgd, and produce the CQ $\exists y (R(x,y) \wedge R(x,z))$, which is equivalent to the CQ $\exists y\, R(x,y)$. Then, $\exists y\, R(x,y)$ will be resolved using the first tgd, and the CQ $P(x)$ will be obtained, which in turn will be resolved using the third tgd in order to produce the CQ $T(x)$. The UCQ rewriting $q'({\bar x})$ is $P(x) \vee T(x)$. \hfill\markfull
\end{example}

It is easy to see that, whenever the input OMQ consists of linear tgds, during the execution of $\mathsf{XRewrite}$ it is not possible to obtain a CQ that has more atoms than the original one. This is an immediate consequence of the fact that linear tgds have only one atom in their body. Then:

\begin{proposition}\label{pro:function-linear}
$f_{(\class{L},\class{CQ})}\big((\insS,\dep,q)\big)\ \leq\ |q|$.
\end{proposition}

Having the above result in place, it can be shown that the algorithm underlying Theorem~\ref{the:ucq-rewritable-languages-decidability} guesses a polynomially sized witness to non-containment, and then calls a $\C$-oracle for solving query evaluation under linear OMQs, where $\C$ is \PSPACE~in general, and \NP~if the arity is fixed; these complexity classes are obtained from Proposition~\ref{prop:eval-linear}. Therefore, $\cocont((\class{L},\class{CQ}))$ is in \PSPACE~in general, and in $\Sigma_{2}^{P}$ in case of fixed arity. Regarding the lower bounds, Proposition~\ref{pro:eval-to-cont} allows us to inherit the \PSPACE-hardness of $\eval(\class{L},\class{CQ})$; this holds even for constant-free tgds. Unfortunately, in the case of fixed arity, we can only obtain \NP-hardness, while Proposition~\ref{pro:coeval-to-cocont} allows to obtain co\NP-hardness. Nevertheless, it is implicit in~\cite{BiLW12} (see the proof of Theorem~9), where the containment problem for OMQ languages based on description logics is considered, that $\cont((\class{L},\class{CQ}))$ is $\Pi_{2}^{P}$-hard, even for tgds of the form $P(x) \ra R(x)$. Then:

\begin{theorem}\label{the:cont-linear}
$\cont((\class{L},\class{CQ}))$ is \PSPACE-complete, and $\Pi_{2}^{P}$-complete if the arity of the schema is fixed. The lower bounds hold even for tgds without constants.
\end{theorem}

\subsection{Non-Recursiveness}

Although the OMQ language $(\class{NR},\class{CQ})$ is not explicitly considered in~\cite{GoOP14}, where the algorithm $\mathsf{XRewrite}$ is defined, the same algorithm can deal with $(\class{NR},\class{CQ})$. By analyzing the UCQ rewritings constructed by $\mathsf{XRewrite}$, whenever the input query falls in $(\class{NR},\class{CQ})$, we can establish the following result; here, $\body{\tau}$ denotes the body of the tgd $\tau$:

\begin{proposition}\label{pro:function-nr}
It holds that
\[
f_{(\class{NR},\class{CQ})}\big((\insS,\dep,q)\big)\ \leq\ |q| \cdot \left(\max_{\tau \in \dep} \{|\body{\tau}|\}\right)^{|\sch{\dep}|}.
\]
\end{proposition}

Proposition~\ref{pro:function-nr} implies that non-containment for queries that fall in $(\class{NR},\class{CQ})$ is witnessed via a database of at most exponential size. We show next that this bound is optimal:

\begin{proposition}\label{pro:function-nr-lower-bound}
There are sets of $(\class{NR},\class{CQ})$ OMQs
\[
\{Q_{1}^{n} = (\insS,\dep_{1}^{n},q_1)\}_{n>0} \quad\textrm{and}\quad \{Q_{2}^{n} = (\insS,\dep_{2}^{n},q_2)\}_{n>0},
\]
where $|\sch{\dep_{1}^{n}}| = |\sch{\dep_{2}^{n}}| = n+2$, such that for every $\insS$-database $D$, if $Q_{1}^{n}(D) \not\subseteq Q_{2}^{n}(D)$ then
$|D| \geq 2^{n-1}$.
\end{proposition}

Let us now focus on the complexity of $\cont((\class{NR},\class{CQ}))$. The algorithm underlying Theorem~\ref{the:ucq-rewritable-languages-decidability}, together with the exponential bound provided by Proposition~\ref{pro:function-nr}, implies that $\cocont((\class{NR},\class{CQ}))$ is feasible in non-deterministic exponential time with access to a \NEXP~oracle, which immediately implies that $\cont((\class{NR},\class{CQ}))$ is in \EXPSPACE.
Unfortunately, the exact complexity of $\cont((\class{NR},\class{CQ}))$ is still an open problem, and we conjecture that is \text{\rm P}$^{\textsc{NEXP}}$-complete; recall that $\NEXP \subseteq \text{\rm P}^{\textsc{NEXP}} \subseteq \EXPSPACE$.
In what follows, we briefly explain how the \text{\rm P}$^{\textsc{NEXP}}$-hardness is obtained.
To this end, we exploit a tiling problem that has been recently introduced in~\cite{EiLP16}. Roughly speaking, an instance of this tiling problem is a triple $(m,T_1,T_2)$, where $m$ is an integer in unary representation, and $T_1,T_2$ are standard tiling problems for the exponential grid $2^n \times 2^n$. The question is whether, for every initial condition $w$ of length $m$, $T_1$ has no solution with $w$ or $T_2$ has some solution with $w$. The initial condition $w$ simply fixes the first $m$ tiles of the first row of the grid.
We construct in polynomial time two $(\class{NR},\class{CQ})$ queries $Q_1$ and $Q_2$ such that $(m,T_1,T_2)$ has a solution iff $Q_1 \subseteq Q_2$. The idea is to force every input database to store an initial condition $w$ of length $m$, and then encode the problem whether $T_i$ has a solution with $w$ into $Q_i$, for each $i \in \{1,2\}$. From the above discussion we get that:

\begin{theorem}\label{the:cont-nr}
$\cont((\class{NR},\class{CQ}))$ is in \EXPSPACE, and {\em P}$^{\textsc{NEXP}}$-hard. The lower bound holds even if the arity of the schema is fixed and the tgds are without constants.
\end{theorem}

\subsection{Stickiness}

We now focus on $(\class{S},\class{CQ})$. As shown in~\cite{GoOP14}, given a query $(\insS,\dep,q)$, there exists an execution of $\mathsf{XRewrite}$ that constructs a UCQ rewriting $q_1({\bar x}) \vee \cdots \vee q_n({\bar x})$ over $\insS$ with the following property: for each $i \in \{1,\ldots,n\}$, if a variable $v$ occurs in $q_i$ in more than one atom, then $v$ already occurs in $q$.
This property has been used in~\cite{GoOP14} to bound the number of atoms that can appear in a single CQ $q_i$. We write $T(q)$ for the set of terms (constants and variables) occurring in $q$, $C(\dep)$ for the set of constants occurring in $\dep$, and $\arity{\insS}$ for the maximum arity over all predicates of $\insS$.

\begin{proposition}\label{pro:function-sticky}
It holds that
\[
f_{(\class{S},\class{CQ})}((\insS,\dep,q))\ \leq\ |\insS| \cdot \left(|T(q)| + |C(\dep)| + 1\right)^{|\arity{\insS}|}.
\]
\end{proposition}

Proposition~\ref{pro:function-sticky} implies that non-containment for $(\class{S},\class{CQ})$ queries is witnessed via a database of at most exponential size. As for $(\class{NR},\class{CQ})$ queries,
we can show that this bound is optimal; here, for a set $\dep$ of tgds, we denote by $||\dep||$ the number of symbols occurring in $\dep$:

\begin{proposition}\label{pro:function-sticky-lower-bound}
There exists a set of $(\class{S},\class{CQ})$ OMQs
\[
\{Q^{n} = (\{S/n\},\dep^{n},q({\bar x}))\}_{n>0},~\text{\rm where}~||\dep^n|| \in O(n^2),
\]
such that for every $Q = (\{S\},\dep',q'({\bar x})) \in (\class{TGD},\class{CQ})$ and $\{S\}$-database $D$, if $Q^{n}(D) \not\subseteq Q(D)$ then $|D| \geq 2^{n-2}$.
\end{proposition}

We now study the complexity of $\cont((\class{S},\class{CQ}))$. We first focus on schemas of unbounded arity. Proposition~\ref{pro:function-sticky} implies that the algorithm underlying Theorem~\ref{the:ucq-rewritable-languages-decidability} runs in exponential time assuming access to a $\C$-oracle, where $\C$ is a complexity class powerful enough for solving $\eval(\class{S},\class{CQ})$ and its complement. But, since $\eval(\class{S},\class{CQ})$ is in \EXP~(see Proposition~\ref{prop:eval-sticky}), both $\eval(\class{S},\class{CQ})$ and its complement are in~\NEXP, and thus, the oracle call is not really needed. Consequently, $\cocont((\class{C},\class{CQ}))$ is in~\NEXP.

A matching lower bound is obtained by a reduction from the standard tiling problem for the exponential grid $2^n \times 2^n$. In fact, the same lower bound has been recently established in~\cite{BP16}; however, our result is stronger as it shows that the problem remains hard even if the right-hand side query is a linear OMQ of a simple form -- this is also discussed in Section~\ref{sec:different-languages}, where containment of queries that fall in different OMQ languages is studied.
Regarding schemas of fixed arity, Proposition~\ref{pro:function-sticky} provides a witness for non-containment of polynomial size, which implies that the algorithm underlying Theorem~\ref{the:ucq-rewritable-languages-decidability} runs in polynomial time with access to an \NP-oracle. Therefore, $\coeval(\class{S},\class{CQ})$ is in $\Sigma_{2}^{P}$, while a matching lower bound is implicit in~\cite{BiLW12}. Then:

\begin{theorem}\label{the:cont-sticky}
$\cont((\class{S},\class{CQ}))$ is {\em co}\NEXP-complete, even if the set of tgds uses only two constants. In the case of fixed arity, it is $\Pi_{2}^{P}$-complete, even for constant-free tgds.
\end{theorem}

Clearly, there exists a double-exponential time algorithm for solving $\cont((\class{S},\class{CQ}))$, which might sound discouraging. However, Proposition~\ref{pro:function-sticky} implies that the runtime is double-exponential only in the maximum arity of the data schema.

%%% Local Variables:
%%% fill-column: 72
%%% TeX-PDF-mode: t
%%% TeX-debug-bad-boxes: t
%%% TeX-master: "main.tex"
%%% TeX-parse-self: t
%%% TeX-auto-save: t
%%% reftex-plug-into-AUCTeX: t
%%% End:

\section{Guardedness}\label{sec:guardedness}

We proceed with the problem of containment for guarded OMQs, and we establish the following result:

\begin{theorem}\label{the:cont-guarded}
$\cont((\class{G},\class{CQ}))$ is {\em 2}\EXP-complete. The lower bound holds even if the arity of the schema is fixed, and the tgds are without constants.
\end{theorem}

The lower bound is immediately inherited from~\cite{BLW16}, where it is shown that containment for OMQs based on the description logic $\ca{ELI}$ is 2\EXP-hard. Recall that a set of $\ca{ELI}$ axioms can be equivalently rewritten as a constant-free set of guarded tgds using only unary and binary predicates, which implies the lower bound stated in Theorem~\ref{the:cont-guarded}.
However, we cannot immediately inherit the desired upper bound since the DL-based OMQ languages considered in~\cite{BLW16} are either weaker than or incomparable to $(\class{G},\class{CQ})$. Nevertheless, the technique developed in~\cite{BLW16} was extremely useful for our analysis. Actually, our automata-based procedure exploits a combination of ideas from~\cite{BLW16,GrWa99}. The rest of this section is devoted to providing a high-level explanation of this procedure.
%; more details can be found in~\cite{BaBP17}.

For the sake of technical clarity, we focus on constant-free tgds and CQs, but all the results can be extended to the general case at the price of more involved definitions and proofs.
Moreover, for simplicity, we focus on Boolean CQs. In other words, we study the problem for $(\class{G},\class{BCQ})$, where $\class{BCQ}$ denotes the class of Boolean CQs. This does not affect the generality of our proof since it is known that $\cont((\class{G},\class{CQ}))$ can be reduced in polynomial time to $\cont((\class{G},\class{BCQ}))$~\cite{BLW16}.

\medskip
\noindent \paragraph{A first glimpse.} As already said, $(\class{G},\class{CQ})$ is not UCQ rewritable and, therefore, we cannot employ Proposition~\ref{pro:small-witness-property} in order to establish a small witness property as for the languages considered in Section~\ref{sec:ucq-rewritability}. We have tried to establish a small witness property for $(\class{G},\class{CQ})$ by following a different route, but it turned out to be a difficult task. Nevertheless, we can show a tree witness property, which states that non-containment for $(\class{G},\class{CQ})$ is witnessed via a tree-like database. This allows us to devise a procedure based on alternating tree automata. Summing up, the proof for the 2\EXP~membership of $(\class{G},\class{CQ})$ proceeds in three steps:
\begin{enumerate}
\item Establish a tree witness property;

\item Encode the tree-like witnesses as trees that can be accepted by an alternating tree automaton; and

\item Construct an automaton that decides $\cont((\class{G},\class{CQ}))$; in fact, we reduce $\cont((\class{G},\class{CQ}))$ into emptiness for two-way alternating parity automata on finite trees.
\end{enumerate}
Each one of the above three steps is discussed in more details in the following three sections.

\subsection{Tree Witness Property}

From the above informal discussion, it is clear that tree-like databases are crucial for our analysis. Let us make this notion more precise using guarded tree decompositions.
A \emph{tree decomposition} of a database $D$ is a labeled rooted tree $T = (V,E,\lambda)$, where $\lambda : V \ra 2^{\adom{D}}$, such that:
(i) for each atom $R(t_1,\ldots,t_n) \in D$, there exists $v \in V$ such that
$\lambda(v) \supseteq \{t_1,\ldots,t_n\}$, and (ii) for every term $t \in \adom{D}$, the set $\{v \in V \mid t \in \lambda(v)\}$ induces a connected subtree of $T$.
The tree decomposition $T$ is called {\em $[U]$-guarded}, where $U \subseteq V$, if, for every node $v \in V \setminus U$, there exists an atom $R(t_1,\ldots,t_n) \in D$ such that $\lambda(v) \subseteq \{t_1,\ldots,t_n\}$.
We write $\rt{T}$ for the root node of $T$, and $D_T(v)$, where $v \in V$, for the subset of $D$ induced by $\lambda(v)$. We are now ready to formalize the notion of the tree-like database:

\begin{definition}\label{def:c-tree}
An $\insS$-database $D$ is a {\em $C$-tree}, where $C \subseteq D$, if there is a tree decomposition $T$ of $D$ such that:
\begin{enumerate}
\item $D_T(\rt{T}) = C$ and

\item $T$ is $[\{\rt{T}\}]$-guarded. \hfill\markfull
\end{enumerate}
\end{definition}

Roughly, whenever a database $D$ is a $C$-tree, $C$ is the cyclic part of $D$, while the rest of $D$ is tree-like. Interestingly, for deciding $\cont((\class{G},\class{BCQ}))$ it suffices to focus on databases that are $C$-trees and $|\adom{C}|$ depends only on the left-hand side OMQ. Recall that for a schema $\insS$ we write $\arity{\insS}$ for the maximum arity over all predicates of $\insS$. Then:

\begin{proposition}\label{pro:tree-witness-property}
Let $Q_i = (\insS,\dep_i,q_i) \in (\class{G},\class{BCQ})$, for $i \in \{1,2\}$. The following are equivalent:
\begin{enumerate}
\item $Q_1 \subseteq Q_2$.

\item $Q_1(D) \subseteq Q_2(D)$, for every $C$-tree $\insS$-database $D$ such that $|\adom{C}| \leq (\arity{\insS \cup \sch{\dep_1}} \cdot |q_1|)$.
\end{enumerate}
\end{proposition}

The fact that $(1) \Rightarrow (2)$ holds trivially, while $(2) \Rightarrow (1)$ is shown by using a variant of the notion of guarded unravelling and compactness.
Let us clarify that the above result does not provide a decision procedure for $\cont((\class{G}, \class{BCQ}))$, since we have to consider infinitely many databases that are $C$-trees with $|\adom{C}| \leq (\arity{\insS \cup \sch{\dep_1}} \cdot |q_1|)$.

\subsection{Encoding Tree-like Databases}

It is generally known that a database $D$ whose treewidth\footnote{Recall that the treewidth of a database $D$ is the minimum width among all possible tree decompositions $T =(V,E,\lambda)$ of $D$, while the width of $T$ is defined as $\max_{v \in V} \{|\lambda(v)|\}-1$.} is bounded by an integer $k$ can be encoded into a tree over a finite alphabet of double-exponential size in $k$ that can be accepted by an alternating tree automaton; see, e.g.,~\cite{BeBB16}.
Consider an alphabet $\Gamma$, and let $\mbb{N}^\ast$ be the set of finite sequences of natural numbers, including the empty sequence. A \emph{$\Gamma$-labeled tree} is a pair $L = \tup{T, \lambda}$, where $T \subseteq \mbb{N}^\ast$ is closed under prefixes, and $\lambda \colon T \ra \Gamma$ is the labeling function. The elements of $T$ identify the nodes of $L$.
It can be shown that $D$ and a tree decomposition $T$ of $D$ with width $k$ can be encoded as a $\Gamma$-labeled tree $L$, where $\Gamma$ is an alphabet of double-exponential size in $k$, such that each node of $T$ corresponds to exactly one node of $L$ and vice versa.

Consider now a $C$-tree $\insS$-database $D$, and let $T$ be the tree decomposition that witnesses that $D$ is a $C$-tree. The width of $T$ is at most $k = (|\adom{C}| + \arity{\insS} - 1)$, and thus, the treewidth of $D$ is bounded by $k$. Hence, from the above discussion, $D$ and $T$ can be encoded as a $\Gamma$-labeled tree, where $\Gamma$ is of double-exponential size in $k$.
In general, given an $\insS$-database $D$ that is a $C$-tree due to the tree decomposition $T$, we show that $D$ and $T$ can be encoded as a $\Gamma_{\insS,l}$-labeled tree, with $|\adom{C}| \leq l$ and $|\Gamma_{\insS,l}|$ being double-exponential in $\arity{\insS}$ and exponential in $|\insS|$ and $l$.

Although every $C$-tree $\insS$-database $D$ can be encoded as a $\Gamma_{\insS,l}$-labeled tree, the other direction does not hold. In other words, it is not true that every $\Gamma_{\insS,l}$-labeled tree encodes a $C$-tree $\insS$-database $D$ and its corresponding tree decomposition. In view of this fact, we need the additional notion of consistency. A $\Gamma_{\insS,l}$-labeled tree is called {\em consistent} if it satisfies certain syntactic properties -- we do not give these properties here since they are not vital in order to understand the high-level idea of the proof. Now, given a consistent $\Gamma_{\insS,l}$-labeled tree $L$, we can show that $L$ can be decoded into an $\insS$-database $\dec{L}$ that is a $C$-tree with $|\adom{C}| \leq l$.
From the above discussion and Proposition~\ref{pro:tree-witness-property}, we obtain:
%the following technical lemma:

\begin{lemma}\label{lem:consistent-labeled-trees}
Let $Q_i = (\insS,\dep_i,q_i) \in (\class{G},\class{BCQ})$, for $i \in \{1,2\}$. The following are equivalent:
\begin{enumerate}
\item $Q_1 \subseteq Q_2$.

\item $Q_1(\dec{L}) \subseteq Q_2(\dec{L})$, for every consistent $\Gamma_{\insS,l}$-labeled tree $L$, where $l = (\arity{\insS \cup \sch{\dep_1}} \cdot |q_1|)$.
\end{enumerate}
\end{lemma}

\subsection{Constructing Tree Automata}

Having the above result in place, we can now proceed with our automata-based procedure. We make use of two-way alternating parity automata (2WAPA) that run on finite labeled trees. Two-way alternating automata process the input tree while branching in an alternating fashion to successor states, and thereby moving either down or up the input tree; the detailed definition can be found in~\cite{BaBP17}.
Our goal is to reduce $\cont((\class{G},\class{BCQ}))$ to the emptiness problem for 2WAPA. As usual, given a 2WAPA $\fk{A}$, we denote by $\ca{L}(\fk{A})$ the \emph{language} of $\fk{A}$, i.e., the set of labeled trees it accepts. The emptiness problem is defined as follows: given a 2WAPA $\fk{A}$, does $\ca{L}(\fk{A}) = \emptyset$? Thus, given $Q_1,Q_2 \in (\class{G},\class{BCQ})$, we need to construct a 2WAPA $\fk{A}$ such that $Q_1 \subseteq Q_2$ iff $\ca{L}(\fk{A}) = \emptyset$.
It is well-known that deciding whether $\ca{L}(\fk{A}) = \emptyset$ is feasible in exponential time in the number of states, and in polynomial time in the size of the input alphabet \cite{CGKV88}. Therefore, 
%in order to obtain the desired 2\EXP~upper bound, 
we should construct $\fk{A}$ in double-exponential time, while the number of states must be at most exponential.

We first need a way to check consistency of labeled trees. It is not difficult to devise an automaton for this task.

\begin{lemma}\label{lem:automaton-1}
Consider a schema $\insS$ and an integer $l > 0$. There is a 2WAPA $\fk{C}_{\insS,l}$ that accepts a $\Gamma_{\insS,l}$-labeled tree $L$ iff $L$ is consistent.
The number of states of $\fk{C}_{\insS,l}$ is logarithmic in the size of $\Gamma_{\insS,l}$. Furthermore, $\fk{C}_{\insS,l}$ can be constructed in polynomial time in the size of $\Gamma_{\insS,l}$.
\end{lemma}

Now, the crucial task is, given an OMQ $Q \in (\class{G},\class{BCQ})$, to devise an automaton that accepts labeled trees which correspond to databases that make $Q$ true. %We write $||Q||$ for the number of symbols that appear in $Q$.

\begin{lemma}\label{lem:automaton-2}
Let $Q = (\insS,\dep,q) \in (\class{G},\class{BCQ})$. There is a 2WAPA $\fk{A}_{Q,l}$, where $l > 0$, that accepts a consistent $\Gamma_{\insS,l}$-labeled tree $L$ iff $Q(\dec{L}) \neq \emptyset$.
The number of states of $\fk{A}_{Q,l}$ is exponential in $||Q||$ and $l$. Furthermore, $\fk{A}_{Q,l}$ can be constructed in double-exponential time in $||Q||$ and $l$.
\end{lemma}

The intuition underlying $\fk{A}_{Q,l}$ can be described as follows. $\fk{A}_{Q,l}$ tries to identify all the possible ways the CQ $q$ can be mapped to $\chase{D}{\dep}$, for \emph{any} $C$-tree $\insS$-database $D$ such that $|\adom{C}| \leq l$.
It then arrives at possible ways how the input tree can satisfy $Q$. These ``possible ways'' correspond to \emph{squid decompositions}, a notion introduced in~\cite{CaGK13} that indicates which part of the query is mapped to the cyclic part $C$ of $D$, and which to the tree-like part of $D$. The automaton exhaustively checks all squid decompositions by traversing the input tree and, at the same time, explores possible ways how to match the single parts of the squid decomposition at hand. The automaton finally accepts if it finds a squid decomposition that can be mapped to $\chase{D}{\dep}$.

Having the above automata in place, we can proceed with our main technical result, which shows that $\cont(\class{G},\class{BCQ})$ can be reduced to the emptiness problem for 2WAPA. But let us first recall some key results about 2WAPA, which are essential for our final construction. It is well-known that languages accepted by 2WAPAs are closed under intersection and complement.  Given two 2WAPAs $\fk{A}_1$ and $\fk{A}_2$, we write $\fk{A}_1 \cap \fk{A}_2$ for a 2WAPA, which can be constructed in polynomial time, that accepts the language $\ca{L}(\fk{A}_1) \cap \ca{L}(\fk{A}_2)$. Moreover, for a 2WAPA $\fk{A}$, we write $\overline{\fk{A}}$ for the 2WAPA, which is also constructible in polynomial time, that accepts the complement of $\ca{L}(\fk{A})$. We can now show the following:

\begin{proposition}\label{pro:cont-to-emptiness}
Consider $Q_1,Q_2 \in (\class{G},\class{BCQ})$. We can construct in double-exponential time a 2WAPA $\fk{A}$, which has exponentially many states, such that
\[
Q_1 \subseteq Q_2 \iff \ca{L}(\fk{A}) = \emptyset.
\]
\end{proposition}

\begin{proofsk}
Let $Q_i = (\insS,\dep_i,q_i)$, for $i \in \{1,2\}$, and $l = (\arity{\insS \cup \sch{\dep_1}} \cdot |q_1|)$. Then
$\fk{A}$ is defined as
%\[
$(\fk{C}_{\insS,l}\ \cap\ \fk{A}_{Q_1,l})\ \cap\ \overline{\fk{A}_{Q_2,l}}$.
%\]
Since $\Gamma_{\insS,l}$ has double-exponential size, Lemmas~\ref{lem:automaton-1} and~\ref{lem:automaton-2} imply that $\fk{A}$ can be constructed in double-exponential time, while it has exponentially many states. Lemma~\ref{lem:consistent-labeled-trees} implies that $Q_1 \subseteq Q_2$ iff $\ca{L}(\fk{A}) = \emptyset$.
%, and the claim follows.
\end{proofsk}

Proposition~\ref{pro:cont-to-emptiness} implies that $\cont((\class{G},\class{BCQ}))$ is in 2\EXP, and Theorem~\ref{the:cont-guarded} follows.
Thus, there exists a double-exponential time algorithm for solving $\cont((\class{G},\class{CQ}))$. Interestingly, the runtime is double-exponential only in the size of the CQs and the maximum arity of the schema. This can be obtained by a providing a more refined complexity analysis of the construction of the 2WAPA $\fk{A}$ in Proposition~\ref{pro:cont-to-emptiness}.

%%% Local Variables:
%%% fill-column: 72
%%% TeX-PDF-mode: t
%%% TeX-debug-bad-boxes: t
%%% TeX-master: "main.tex"
%%% TeX-parse-self: t
%%% TeX-auto-save: t
%%% reftex-plug-into-AUCTeX: t
%%% End:

\section{Combining Languages}\label{sec:different-languages}

In the previous two sections, we studied the containment problem relative to a language $\class{O}$, i.e., both OMQs fall in $\class{O}$. However, it is natural to consider the version of the problem where the involved OMQs fall in different languages. This is the goal of this section. Our analysis proceeds by considering the two cases where the left-hand side (LHS) query falls in a UCQ rewritable OMQ language, or it is guarded.

\subsection{The LHS Query is UCQ Rewritable}

As an immediate corollary of Theorem~\ref{the:ucq-rewritable-languages-decidability} we obtain the following result: $\cont((\class{C}_1,\class{CQ}),(\class{C}_2,\class{CQ}))$, for $\class{C}_1 \neq \class{C}_2$, $\class{C}_1 \in \{\class{L},\class{NR},\class{S}\}$ and $\class{C}_2 \in \{\class{L},\class{NR},\class{S},\class{G}\}$, is decidable.
By exploiting the algorithm underlying Theorem~\ref{the:ucq-rewritable-languages-decidability}, we establish optimal upper bounds for all the problems at hand with the only exception of $\cont((\class{S},\class{CQ}),(\class{NR},\class{CQ}))$.
For the latter, we obtain an \EXPSPACE~upper bound, by providing a similar analysis as for $\cont((\class{NR},\class{CQ}))$, while a \NEXP~lower bound is inherited from query evaluation by exploiting Proposition~\ref{pro:eval-to-cont}.
It is rather tedious, and not very interesting from a technical point of view, to go through all the containment problems in question\footnote{There are eighteen different cases obtained by considering all the possible pairs $(\class{O}_1,\class{O}_2)$ of OMQ languages, where $\class{O}_1 \neq \class{O}_2$ and $\class{O}_1$ is UCQ rewritable, and the two cases whether the arity of the schema is fixed or not.} and explain in details how the exact upper bounds are obtained; we leave this as an exercise to the interested reader. 

%a summarization of those complexity results can be found in~\cite{BaBP17}.

Regarding the matching lower bounds, in most of the cases they are inherited from query evaluation or its complement by exploiting Propositions~\ref{pro:eval-to-cont} and~\ref{pro:coeval-to-cocont}, respectively. There are, however, some exceptions:
\begin{itemize}%\itemsep-\parsep
\item $\cont((\class{S},\class{CQ}),(\class{L},\class{CQ}))$ in the case of unbounded arity, where the problem is \text{\rm co}\NEXP-hard, even for sets of tgds that use only two constants. This is shown by a reduction from the standard tiling problem for the exponential grid $2^n \times 2^n$.

\item $\cont((\class{L},\class{CQ}),(\class{S},\class{CQ}))$ and $\cont((\class{S},\class{CQ}),(\class{L},\class{CQ}))$ in the case of bounded arity, where both problems are $\Pi_{2}^{P}$-hard even for constant-free tgds; implicit in~\cite{BiLW12}.
\end{itemize}

\subsection{The LHS Query is Guarded}

We proceed with the case where the LHS query is guarded, and we show the following result:

\begin{theorem}
\label{th:lhsguarded}
$\cont((\class{G},\class{CQ}),(\class{C},\class{CQ}))$ is $\C$-complete:
\begin{eqnarray*}
\C &=& \left\{
\begin{array}{ll}
\text{\rm {\em 2}\EXP}, & \class{C} \in \{\class{L},\class{S}\},\\
& \\
\text{\rm {\em 3}\EXP}, & \class{C} = \class{NR}.
\end{array} \right.
\end{eqnarray*}
The lower bounds hold even if the arity of the schema is fixed. Moreover, for $\class{C} = \class{L}$ (resp., $\class{C} \in \{\class{NR},\class{S}\}$) it holds even for tgds with one constant (resp., without constants).
\end{theorem}

\noindent \paragraph{Upper bounds.}
The 2\EXP~membership when $\class{C} = \class{L}$ is an immediate corollary of Theorem~\ref{the:cont-guarded}. This is not true when $\class{C} \in \{\class{NR},\class{S}\}$ since the right-hand side query is not guarded. But in this case, since $(\class{NR},\class{CQ})$ and $(\class{S},\class{CQ})$ are UCQ rewritable, one can rewrite the right-hand side query as a UCQ, and then apply the machinery developed in Section~\ref{sec:guardedness} for solving $\cont((\class{G},\class{CQ}))$.
More precisely, given OMQs $Q_1 \in (\class{G},\class{CQ})$ and $Q_2 \in (\class{C},\class{CQ})$, where $\class{C} \in \{\class{NR},\class{S}\}$, $Q_1 \subseteq Q_2$ iff $Q_1 \subseteq q$, where $q$ is a UCQ rewriting of $Q_2$. Thus, an immediate decision procedure, which exploits the algorithm $\mathsf{XRewrite}$, is the following:
\begin{enumerate}
\item Let $q = \mathsf{XRewrite}(Q_2)$;

\item For each $q' \in q$: if $Q_1 \subseteq q'$, then proceed; otherwise, reject; and

\item Accept.
\end{enumerate}
The above procedure runs in triple-exponential time. The first step is feasible in double-exponential time~\cite{GoOP14}. Now, for a single CQ $q' \in q$ (which is a guarded OMQ with an empty set of tgds) the check whether $Q_1 \subseteq q'$ can be done by using the machinery developed in Section~\ref{sec:guardedness}, which reduces our problem to checking whether the language of a 2WAPA $\fk{A}$ is empty. However, it should not be forgotten that $q'$ is of exponential size, and thus,
the automaton $\fk{A}$ has double-exponentially many states. This  in turn implies that checking whether $\ca{L}(\fk{A}) = \emptyset$ is in 3\EXP, as claimed.

Although the above algorithm establishes an optimal upper bound for non-recursive OMQs, a more refined analysis is needed for sticky OMQs. In fact, we need a more refined complexity analysis for the problem $\cont((\class{G},\class{CQ}),\class{UCQ})$, that is, to decide whether a guarded OMQ is contained in a UCQ. To this end, we provide an automata construction different from the one employed in Section~\ref{sec:guardedness}, which allows us to establish a refined complexity upper bound for the problem in question.
Consider a $(\class{G},\class{CQ})$ query $Q$, and a UCQ $q = q_1 \vee \cdots \vee q_n$. As usual, we write $||Q||$ and $||q_i||$ for the number of symbols that occur in $Q$ and $q_i$, respectively, and we write $\mi{var}_{\geq 2}(q_i)$ for the set of variables that appear in more than one atom of $q_i$. By exploiting our new automata-based procedure, we show that the problem of checking if $Q \subseteq q$ is feasible in double-exponential time in $(||Q|| + \max_{1 \leq i \leq n} \{|\mi{var}_{\geq 2}(q_i)|\})$, exponential time in $\max_{1 \leq i \leq n} \{||q_i||\}$, and polynomial time in $n$.

This result allows us to show that the above procedure establishes 2\EXP-membership when the right-hand side OMQ is sticky. But first we need to recall the following
key properties of the UCQ rewriting $q = \mathsf{XRewrite}(Q_2)$, constructed during the first step of the algorithm:
\begin{enumerate}
\item $q$ consists of double-exponentially many CQs,

\item each CQ of $q$ is of exponential size, and

\item for each $q' \in q$, $\mi{var}_{\geq 2}(q')$ is a subset of the variables of the original CQ that appears in $Q_2$.
\end{enumerate}
By combining these key properties with the complexity analysis performed above, it is now straightforward to show that $\cont((\class{G},\class{CQ}),(\class{S},\class{CQ}))$ is in 2\EXP.

\medskip
\noindent \paragraph{Lower Bounds.} We establish matching lower bounds by refining techniques from~\cite{CV97}, where it is shown that containment of Datalog in UCQ is 2\EXP-complete, while containment of Datalog in non-recursive Datalog is 3\EXP-complete; the lower bounds hold for fixed-arity predicates, and constant-free rules.
Interestingly, the LHS query can be transformed into a Datalog query such that each rule has a body-atom that contains all the variables, i.e., is guarded. This is achieved by increasing the arity of some predicates in order to have enough positions for all the body-variables. However, for each rule, the number of unguarded variables that we need to guard is constant, and thus, the arity of the schema remains constant.
We conclude that $\cont((\class{G},\class{CQ}),(\class{NR},\class{CQ}))$ is 3\EXP-hard.
Moreover, containment of guarded OMQs in UCQs is 2\EXP-hard, which in turn allows us to show, by exploiting the construction underlying Proposition~\ref{pro:from-ucq-to-cq}, that $\cont((\class{G},\class{CQ}),(\class{L},\class{CQ}))$ is 2\EXP-hard, even if the set of linear tgds uses only one constant, while $\cont((\class{G},\class{CQ}),(\class{S},\class{CQ}))$ is 2\EXP-hard, even for tgds without constants. 
\section{Applications}\label{sec:applications}

Interestingly, our results on $\cont((\class{G},\class{CQ}))$ can be applied to other important static analysis tasks, in particular, distribution over components and UCQ rewritability. Each one of those tasks is considered in the following two sections.

\subsection{Distribution Over Components}

The notion of distribution over components has been introduced in~\cite{AKNZ14}, and it states that the answer to a query can be computed by parallelizing it over the (maximally connected) components of the input database. But let us first make precise what a component is.
A set of atoms $A$ is {\em connected} if for all $c,d \in \adom{A}$,
%\footnote{This notation can be naturally defined for sets of atoms or single atoms.}
there exists a sequence $\alpha_1,\ldots,\alpha_n$ of atoms in $A$ such that $c \in \adom{\alpha_1}$, $d \in \adom{\alpha_n}$, and $\adom{\alpha_i} \cap \adom{\alpha_{i+1}} \neq \emptyset$, for each $i \in \{1, \ldots, n-1\}$. We call $B \subseteq A$ a {\em component} of $A$ if (i) $B$ is connected, and (ii) for every $\alpha \in A \setminus B$, $B \cup \{\alpha\}$ is not connected.\footnote{For technical clarity, the notion of component is defined only for sets of atoms that do not contain $0$-ary atoms.} Let $\co(A)$ be the set of components of $A$. We are now ready to introduce the notion of distribution over components.
%\begin{definition}
Consider an OMQ $Q = (\insS,\dep,q) \in (\class{TGD},\class{CQ})$. We say that $Q$ {\em distributes over components} if $Q(D) = Q(D_1) \cup \cdots \cup Q(D_n)$, where $\co(D) = \{D_1,\ldots,D_n\}$, for every $\insS$-database $D$.
In this case, $Q(D)$ can be computed without any communication over a network using a distribution where every computing node is assigned some of the components of the database, and every component is assigned to at least one computing node. In other words, $Q$ can be evaluated in a distributed and coordination-free manner; for more details on coordination-free evaluation see~\cite{AKNZ14,AKNZ15,AmNB13}. Therefore, it would be quite beneficial if we can decide whether an OMQ distributes over components, and thus, we obtain the following interesting static analysis task:

\begin{center}
\fbox{\begin{tabular}{ll}
{\small PROBLEM} : & $\dist(\class{C},\class{CQ})$
\\
{\small INPUT} : & An OMQ $Q \in (\class{C},\class{CQ})$.
\\
{\small QUESTION} : &  Does $Q$ distributes over components?
\end{tabular}}
\end{center}

The above problem has been studied in~\cite{BP16}, where tight complexity bounds for $(\class{L},\class{CQ})$ and $(\class{S},\class{CQ})$ have been established. However, its exact complexity for guarded OMQs has been left open. Our results on containment for guarded OMQs allow us to close this problem. But first we need to recall a key result that semantically characterizes distribution over components. An OMQ $Q$ with data schema $\insS$ is {\em unsatisfiable} if there is no $\insS$-database $D$ such that $Q(D) \neq \emptyset$. Moreover, for a CQ $q$, we write $\co(q)$ for its components. The next result has been shown in~\cite{BP16}:

\begin{proposition}\label{pro:distribution-over-components}
Let $Q = (\insS,\dep,q({\bar x})) \in (\class{G},\class{CQ})$. The following are equivalent:
\begin{enumerate}
\item $Q$ distributes over components.

\item $Q$ is unsatisfiable or there exists $\hat{q}(\bar x) \in \co(q)$ such that $(\insS,\dep,\hat{q}(\bar x)) \subseteq Q$.
\end{enumerate}
\end{proposition}

Checking unsatisfiability can be easily reduced to containment. Thus, the above result, together with Theorem~\ref{the:cont-guarded}, implies that $\dist(\class{G},\class{CQ})$ is in 2\EXP, while a matching lower bound is implicit in~\cite{BP16}. Then:

\begin{theorem}
$\dist(\class{G},\class{CQ})$ is {\em 2}\EXP-complete.
\end{theorem}

\subsection{Deciding UCQ Rewritability}

Query rewriting is a well-studied method for evaluating OMQs using standard database technology. The key idea is the following: given an OMQ
$Q = (\ins{S},\Sigma,q(\bar x))$, combine $\Sigma$ and $q$ into a new query $q_\dep(\bar x)$, the so-called rewriting, which can then be evaluated over $D$ yielding the same answer as $Q$ over $D$, for {\em every} $\insS$-database $D$.
For this approach to be realistic, though, it is essential that the rewriting is expressed in a language that can be handled by standard database systems. The typical language that is considered in this setting is first-order (FO) queries~\cite{CDLL*07}. Notice, however, that due to Rossman's Theorem~\cite{Rossman08}, and the fact that OMQs are closed under homomorphisms, FO and UCQ rewritability coincide.
Recall that some OMQ languages are UCQ rewritable, such as the ones based on linear, non-recursive and sticky sets of tgds, while others are not, e.g., guarded OMQs. For those languages $\class{O}$ that are not UCQ rewritable, it is important to be able to check whether a query $Q \in \class{O}$ can be rewritten as a UCQ, in which case we say that it is UCQ rewritable. This gives rise to the following fundamental static analysis task for an OMQ language $(\class{C},\class{CQ})$, where $\class{C} \subseteq \class{TGD}$:

\begin{center}
\fbox{\begin{tabular}{ll}
{\small PROBLEM} : & $\rew(\class{C},\class{CQ})$
\\
{\small INPUT} : & An OMQ $Q \in (\class{C},\class{CQ})$.
\\
{\small QUESTION} : &  Is it the case that $Q$ is UCQ rewritable?
\end{tabular}}
\end{center}

Bienvenu et al.~have recently carried out an in-depth study of the above problem for OMQ languages based on central Horn-DLs~\cite{BLW16}. One of their main results is that the above problem for the OMQ language based on $\ca{ELHI}$, one of the most expressive members of the $\ca{EL}$-family of DLs, is 2\EXP-complete.
Interestingly, by adapting the tree automata techniques developed in Section~\ref{sec:guardedness}, we can generalize the above result: deciding UCQ rewritability for the OMQ language based on guarded tgds over unary and binary relations is in 2\EXP. Let $\class{G}_2$ be the class of (finite) sets of guarded tgds over unary and binary relations. Then:

\begin{theorem}\label{theo:fo-rew}
$\rew(\guarded_2,\class{CQ})$ is {\em 2}\EXP-complete.
\end{theorem}

Since the lower bound is inherited from~\cite{BLW16}, we concentrate  on the upper bound. As in Section \ref{sec:guardedness}, we can focus on BCQs, i.e., it suffices to show that $\rew(\class{G}_2,\class{BCQ})$ is in 2\EXP. Our proof proceeds in two steps:
\begin{enumerate}
\item We semantically characterize UCQ rewritability for queries in $(\class{G}_2,\class{CQ})$ in terms of a certain {\em boundedness} property for the set of $C$-trees defined in Section~\ref{sec:guardedness}.

\item We extend the techniques developed in Section~\ref{sec:guardedness} and construct in double-exponential time a 2WAPA $\fk{A}$ that has
    exponentially many states, such that the aforementioned boundedness property does not hold iff $\L(\fk{A})$ is infinite. (Such an {\em infinity problem} for tree automata has been used to obtain the decidability of the {\em boundedness} problem for monadic Datalog
  \cite{CGKV88,Vardi92}).
\end{enumerate}
Our \twoexptime~upper bound then follows since the infinity
problem for a 2WAPA $\fk{A}$, i.e., checking if $\L(\fk{A})$ is
infinite, is feasible in exponential time in the number of states,
and in polynomial time in the size of the alphabet. This follows from
two known results: (a) The 2WAPA $\fk{A}$ can be converted into an
equivalent non-deterministic
%(top-down)
tree automata $\fk{B}$ with a single-exponential blow up in the
number of states \cite{Vardi98}, and (b) solving the infinity problem
for non-deterministic tree automata is feasible in polynomial time;
cf.~\cite{Vardi92}.

It is worth contrasting our proof with the one in~\cite{BLW16} for $\ca{ELHI}$, which
%for showing that $\rew$ for the OMQ language based on $\ca{ELHI}$ is in \twoexptime.
%Unlike ours,
does not make use of the infinity problem for 2WAPA, but applies a different argument based on pumping. This leads to a finer complexity analysis in terms of the size of the different components of the OMQ, but, in our opinion, makes the proof conceptually harder.
%in our opinion.

\medskip
\noindent
{\bf The semantic characterization.}
To establish the semantic characterization from step 1, we need to
define the notion of {\em distance from the root} for an element $u$ in
a $C$-tree database $D$. Intuitively, this corresponds to the minimal
distance between a node that contains $u$ and the root of a tree
decomposition $T$ of $D$ that witnesses the fact that $D$ is a
$C$-tree. We do not consider all such tree decompositions, however, but
concentrate on a well-behaved subclass, which we call the {\em lean}
tree decompositions of the $C$-tree $D$; the formal definition can be found in~\cite{BaBP17}, as it does not add much to the explanation we provide here.
Due to the fact that we focus on unary and binary relations, such lean tree decompositions ensure the invariance of the notion of distance from the root, by severely limiting the level of redundancy allowed in a tree
representation of $D$. Therefore, it does not matter which lean tree decomposition we choose, since in all of them the distance of an element $u$ from the root will be the same. Let ${D}_{\leq k}$ be the subinstance of $D$ induced by the set of elements whose distance from the root is at most
$k$, and let $D_{> k}$ be the subinstance of $D$ induced by the set of elements whose distance from the root is at least $k+1$.

Another useful notion is the {\em branching degree} of a tree decomposition $T$, that is, the maximum number of child nodes over all nodes of $T$. Again, lean tree decompositions ensure the invariance of the branching degree.
This allows us to define the branching degree of a $C$-tree database $D$ as the branching degree of a lean tree decomposition that witnesses the fact that $D$ is a $C$-tree.

It follows from~\cite{BLW16} that being able to decide containment for the OMQ
language $(\guarded_{2},\class{BCQ})$ (as we have done in Section
\ref{sec:guardedness}) allows us to concentrate on {\em connected} CQs when deciding UCQ rewritability. This simplifies technicalities considerably and, in turn, allows us to obtain our desired semantic characterization of UCQ rewritability:

\begin{proposition}\label{prop:fosemantic}
Let $Q = \tup{\sche{S},\Sigma,q} \in (\guarded_{2},\class{BCQ})$, where $q$ is connected. The following are equivalent:
\begin{enumerate}
\item $Q$ is UCQ rewritable.

\item There exist $k,m \geq 0$ (which depend only on $Q$) s.t.
\[
Q(D) \neq \emptyset \ \Longrightarrow \
\big(Q(D_{\leq k}) \neq \emptyset \text{~~or~~} Q(D_{>0}) \neq \emptyset\big),
\]
for each $C$-tree $\insS$-database $D$ with $|\adom{C}| \leq 2 \cdot |q|$ and branching degree at most $m$.
\end{enumerate}
\end{proposition}

\noindent
{\bf The reduction to the infinity problem.}
We now proceed with step 2, and we explain how the boundedness property established in item (2) of Proposition \ref{prop:fosemantic} can be reduced to the infinity problem for 2WAPAs. As in Section \ref{sec:guardedness}, we do not reason with $C$-tree databases directly, but we deal with their encodings as consistent $\Gamma_{\ins{S},l}$-labeled trees. In fact, using the same ideas as in Lemma~\ref{lem:consistent-labeled-trees}, we can show by exploiting Proposition~\ref{prop:fosemantic} that the following are equivalent:
\begin{enumerate}
\item[(i)] $Q$ is UCQ rewritable.

\item[(ii)] There are $k,m \geq 0$ such that
\[
Q(\dec{L}) \neq \emptyset \ \Longrightarrow \
\big(Q(\dec{L}_{\leq k}) \neq \emptyset \text{ or } Q(\dec{L}_{>0}) \neq
\emptyset\big),
\]
for every consistent $\Gamma_{\insS,l}$-labeled tree $L$ with
$l = 2 \cdot |q|$ and whose branching degree is bounded by $m$.
\end{enumerate}

Let us write {\sc Boundedness} for the property expressed in item (ii)
above, which can be reduced to the problem of checking whether some tree language is finite. Let $\L_Q$ be the set of all $\Gamma_{\sche{S},l}$-labeled trees $L$ of branching degree at most $m$
such that: (1) $Q(\dec{L}) = \emptyset$ and (2) there is some
``extension'' $L'$ of $L$, with branching degree $m$, such that $Q(\dec{L'}) \neq \emptyset$ and $Q(\dec{L'}_{> 0}) = \emptyset$. Notice that $L'$ can increase the depth but not the branching degree of $L$. It is not difficult to
show that {\sc Boundedness} holds iff $\L_Q$ is finite.
We then devise in double-exponential time a 2WAPA $\fk{C}_{Q,l}$, which has exponentially many states, such that $\L_Q = \L(\fk{C}_{Q,l})$. Therefore, the following holds:

\begin{proposition}\label{pro:ucqrew-to-infinity}
Consider $Q \in (\class{G}_2,\class{BCQ})$. We can construct in double-exponential time a 2WAPA $\fk{A}$, which has exponentially many states, such that
\[
Q \text{ is UCQ rewritable } \iff\ \ca{L}(\fk{A}) \text{ is finite}.
\]
\end{proposition}

Since checking whether $\L(\fk{A})$ is infinite is feasible in exponential time in the number of states and in polynomial time in the size of the alphabet, Proposition~\ref{pro:ucqrew-to-infinity} implies that $\rew(\guarded_2,\class{CQ})$ is in 2\EXP, as needed.

%%% Local Variables:
%%% fill-column: 72
%%% TeX-PDF-mode: t
%%% TeX-debug-bad-boxes: t
%%% TeX-master: "main.tex"
%%% TeX-parse-self: t
%%% TeX-auto-save: t
%%% reftex-plug-into-AUCTeX: t
%%% End:

\section{Conclusions}\label{sec:conclusions}

We have concentrated on the fundamental problem of containment for OMQ languages based on the main decidable classes of tgds.
%
%We have developed specially tailored techniques for OMQ containment that allowed us to obtain a relatively complete picture for the complexity of the problem at hand.
%
%Our main conclusion is that for the OMQ languages in question, i.e., those based on linear, sticky, non-recursive and guarded tgds, the containment problem is harder (under widely accepted complexity assumptions) than query evaluation, with the exception of linear OMQs over schemas of unbounded arity, where both problems are complete for \PSPACE.
%
We have also used our techniques to close problems related to distribution over components and UCQ rewritability.
%In particular, it is shown that the problem of deciding whether a guarded OMQ distributes over components is 2\EXP-complete, while the problem of deciding whether a guarded OMQ that mentions only unary and binary relations is UCQ rewritable is also 2\EXP-complete.
%
We believe that our techniques for solving containment under guarded OMQs can be extended to frontier-guarded OMQs, an interesting extension of guardedness~\cite{BLMS11}.
We are also convinced that our solution to the problem of deciding UCQ rewritability of guarded OMQs over unary and binary relations can be extended to guarded (or even frontier-guarded) OMQs over arbitrary schemas.
We are currently investigating these challenging problems. 

%\smallskip
%\noindent
%{\small
%{\bf Acknowledgements:} Barcel\'o would like to thank
%D. Figueira, M. Romero, S. Rudolph, and N. Schweikardt for insightful
%discussions about the nature of
%semantic acyclicity under constraints.
%}

\onecolumn
\newpage
\section*{APPENDIX}\label{sec:appendix}

\section*{PRELIMINARIES}

\subsection*{Definition of Non-recursiveness}

In the main body of the paper, we define non-recursive sets of tgds via the notion of predicate graph. Here, we give an alternative definition, based on the well-known notion of stratification, which is more convenient for the combinatorial analysis that we are going to perform in the proof of Proposition~\ref{pro:function-nr}.

\begin{definition}\label{def:stratification}
Consider a set $\dep$ of tgds. A {\em stratification} of $\dep$ is a partition $\{\dep_1,\ldots,\dep_n\}$, where $n > 0$, of $\dep$ such that, for some function $\mu : \sch{\dep} \ra \{0,\ldots,n\}$, the following hold:
\begin{enumerate}
\item For each predicate $R \in \sch{\dep}$, all the tgds with $R$ in their head belong to $\dep_{\mu(R)}$, i.e., they belong to the same set of the partition.

\item If there exists a tgd in $\dep$ such that the predicate $R$ appears in its body, while the predicate $P$ appears in its head, then $\mu(R) < \mu(P)$.
\end{enumerate}
We say that $\dep$ is {\em stratifiable} if it admits a stratification. \hfill\markfull
\end{definition}

It is an easy exercise to show that the predicate graph of a set $\dep$ of tgds is acyclic iff $\dep$ is stratifiable. Then:

\begin{lemma}\label{lem:stratification}
$\dep$ is non-recursive iff $\dep$ is stratifiable.
\end{lemma}

\subsection*{Definition of Stickiness}

In the main body of the paper, we provide an intuitive explanation of stickiness. Here, we recall the formal definition of sticky sets of tgds, introduced in~\cite{CaGP12}.
Fix a set $\dep$ of tgds; w.l.o.g., we assume that, for every pair $(\sigma,\sigma') \in \dep \times \dep$, $\sigma$ and $\sigma'$ do not share variables.
For notational convenience, given an atom $\alpha$ and a variable
$x$ occurring in $\alpha$, $\pos{\alpha}{x}$ is the set of positions
in $\alpha$ at which $x$ occurs; a position $P[i]$ identifies the $i$-th attribute of the predicate $P$.
The definition of stickiness hinges on the notion of marked variables in a set of tgds.

\begin{definition}\label{def:marked-variables}
Consider a tgd $\sigma \in \dep$, and a variable $x$ occurring in the body of $\sigma$. We inductively define when $x$ is \emph{marked in
$\dep$} as follows:
\begin{enumerate}
\item If there exists an atom  $\alpha$ in the head of $\sigma$ such that $x$ does not occur in $\alpha$, then $x$ is marked in $\dep$; and
\item Assuming that there exists an atom $\alpha$ in the head of $\sigma$ such that $x$ occurs in $\alpha$, if there exists $\sigma' \in \dep$ (not necessarily different than $\sigma$) and an atom $\beta$ in the body of $\sigma'$ such that (i) $\alpha$ and $\beta$ have the same predicate and, (ii) each variable in $\beta$ that occurs at a position of $\pos{\alpha}{x}$ is marked in $\dep$, then $x$ is marked in $\dep$. \hfill\markfull
\end{enumerate}
\end{definition}

We are now ready to recall when a set of tgds is sticky:

\begin{definition}\label{def:sticky-tgds}
A set $\dep$ of tgds is \emph{sticky} if, for each $\sigma \in
\dep$, and for each variable $x$ occurring in the body of $\sigma$, the following holds: if $x$ is marked in $\dep$, then $x$ occurs only once in the body of $\sigma$. \hfill\markfull
\end{definition}

\section*{PROOFS OF SECTION~\ref{sec:containment-basics}}

\subsection*{Proof of Proposition~\ref{pro:eval-to-cont}}

Consider an OMQ $Q = (\insS,\dep,q({\bar x})) \in (\class{C},\class{CQ})$, where $\class{C}$ is a class of tgds, an $\insS$-database $D$, and a tuple ${\bar c} \in \adom{D}^{|{\bar x}|}$. We show that:
\[
{\bar c} \in Q(D)\
\iff\
%\textrm{~~~iff~~~}
\underbrace{(\sch{\dep},\emptyset,q_{D,{\bar c}})}_{Q_1} \subseteq \underbrace{(\sch{\dep},\dep,q)}_{Q_2}.
\]

($\Rightarrow$) Assume that $Q_1 \not\subseteq Q_2$. This implies that there exists a $\sch{\dep}$-database $D'$, and a tuple ${\bar t}$ of constants such that ${\bar t} \in q_{D,{\bar c}}(D')$ and ${\bar t} \not\in q(\chase{D'}{\dep})$. Due to the monotonicity of CQs, ${\bar t} \in q_{D,{\bar c}}(\chase{D'}{\dep})$. Since, by construction, the instance $\chase{D'}{\dep}$ satisfies $\dep$, we conclude that $q_{D,{\bar c}} \not\subseteq_{\dep} q$.\footnote{This is the standard notation for the fact that $q_{D,{\bar c}}(I) \not\subseteq q(I)$, for every (possibly infinite) instance $I$ that satisfies $\dep$.} By exploiting the well-known characterization of CQ containment in terms of the chase, we get that ${\bar c} \not\in q(\chase{D}{\dep})$, which is equivalent to ${\bar c} \not\in Q(D)$, as needed.

($\Leftarrow$) Conversely, assume that ${\bar c} \not\in Q(D)$, or, equivalently, ${\bar c} \not\in q(\chase{D}{\dep})$. This implies that ${\bar c} \not\in Q_2(D)$. Observe that ${\bar c} \in q_{D,{\bar c}}(D)$ holds trivially, which in turn implies that ${\bar c} \in Q_1(D)$. Therefore, $Q_1 \not\subseteq Q_2$, and the claim follows.

\subsection*{Proof of Proposition~\ref{pro:from-ucq-to-cq}}

The construction underlying Proposition~\ref{pro:from-ucq-to-cq} relies on the idea of encoding boolean operations (in our case the `or' operator) using a set of atoms; this idea has been exploited in several other works; see, e.g.,~\cite{BeGo10,BMMP16,GoPa03}.
Let $Q = (\insS,\dep,q) \in (\class{C},\class{UCQ})$. Our goal is to construct in polynomial time $Q' = (\insS,\dep',q') \in (\class{C},\class{CQ})$ such that $Q \equiv Q'$. We assume, w.l.o.g., that the predicates of $\insS$ do not appear in the head of a tgd of $\dep$; we can copy the content of a relation $R/k \in \insS$ into an auxiliary predicate $R^{\star}/k$, using the tgd $R(x_1,\ldots,x_k) \ra R^{\star}(x_1,\ldots,x_k)$, while staying inside $\class{C}$, and then rename each predicate $P$ in $\dep$ and $q$ with $P^{\star}$.
The set $\dep'$ consists of the following tgds:

\begin{enumerate}
\item For every $R/k \in \insS$:
\[
R(x_1,\ldots,x_k)\ \ra\ R'(x_1,\ldots,x_k,1), \text{\rm True}(1).
\]
These tgds are annotating the database atoms with the truth constant true, indicating that these are true atoms.

\item Assuming that $q = \exists {\bar y} \, \phi({\bar x}, {\bar y})$, a tgd:
\[
\text{\rm True}(t)\ \ra \exists {\bar x} \exists {\bar y} \exists f \, \phi'_{\wedge}({\bar x},{\bar y},f), \psi(t,f),
\]
where $\phi'_{\wedge}$ is the conjunction of atoms in $\phi$, after replacing each atom $R(v_1,\ldots,v_k)$ with $R'(v_1,\ldots,v_k,f)$, and $\psi$ is the conjunction of atoms
\[
\text{\rm Or}(t,t,t), \text{\rm Or}(t,f,t), \text{\rm Or}(f,t,t), \text{\rm Or}(f,f,f).
\]
This tgd generates a ``copy'' of the atoms in $q$, while annotating them with a null value that represents the truth constant false, indicating that are not necessarily true atoms. Moreover, the truth table of `or' is generated.

\item Finally, for each tgd $\phi({\bar x},{\bar y}) \ra \exists {\bar z} \, \psi({\bar x},{\bar x})$ in $\dep$, a tgd
\[
\phi'({\bar x},{\bar y},w)\ \ra\ \exists {\bar z} \, \psi'({\bar x},{\bar z},w),
\]
where $\phi'$ and $\psi'$ are obtained from $\phi$ and $\psi$, respectively, by replacing each atom $R(v_1,\ldots,v_k)$ with $R'(v_1,\ldots,v_k,w)$. In fact, this is the actual set of tgds $\dep$, with the difference that the value at the last position of each atom (which indicates whether it is true or false) is propagated to the inferred atoms.
\end{enumerate}

Now, assuming that $q = q_1 \vee \cdots \vee q_n$, the CQ $q'$ is defined as follows; let ${\bar x} = x_1 \ldots x_n$ and ${\bar y} = y_1 \ldots y_{n+1}$:
\begin{eqnarray*}
\exists {\bar x} \exists {\bar y}\, (\text{\rm False}(y_1)\, \wedge
\bigwedge_{1 \leq i \leq n} (q'_{i}[x_i] \wedge \text{\rm Or}(y_i,x_i,y_{i+1}))\, \wedge
\text{\rm True}(y_{n+1})),
\end{eqnarray*}
where ${\bar x}$ and ${\bar y}$ are fresh variables not in $q$, and $q'_{i}[x_i]$ is obtained from $q_i$ by replacing each atom $R(v_1,\ldots,v_k)$ with $R'(v_1,\ldots,v_k,x_i)$. This completes our construction.

It is not difficult to show that $Q \equiv Q'$, or, equivalently, for every $\insS$-database $D$, $q(\chase{D}{\dep}) = q'(\chase{D}{\dep'})$. The key observation is that in order to satisfy $\text{\rm True}(y_{n+1})$ in the CQ $q'$, at least one of the ${\bar x_i}$'s must be mapped to $1$, which means that at least one $q_i$ is satisfied by $\chase{D}{\dep}$. Finally, it is easy to verify that, for each $\class{C} \in \{\class{G},\class{L},\class{NR},\class{S}\}$, $\dep \in \class{C}$ implies $\dep' \in \class{C}$, and Proposition~\ref{pro:from-ucq-to-cq} follows.

\section*{PROOFS OF SECTION~\ref{sec:ucq-rewritability}}

\subsection*{Proof of Proposition~\ref{pro:small-witness-property}}

We assume that $q({\bar x}) = \bigvee_{i=1}^{n} q_i({\bar x})$ is a UCQ rewriting of $Q$. Since, by hypothesis, $Q \not\subseteq Q'$, we conclude that $q \not\subseteq Q'$, which in turn implies that there exists $i \in \{1,\ldots,n\}$ such that $q_i \not\subseteq Q'$.
Let $c({\bar x})$ be a tuple of constants obtained by replacing each variable $x$ in ${\bar x}$ with the constant $c(x)$, and $D_{q_i}$ the $\insS$-database obtained from $q_i$ after replacing each variable $x$ in $q_i$ with the constant $c(x)$. We show that:

\begin{lemma}\label{lem:auxiliary-small-witness-property}
$c({\bar x}) \not\in Q'(D_{q_i})$.
\end{lemma}

\begin{proof}
Since $q_i \not\subseteq Q'$, there exists an $\insS$-database $D$, and a tuple of constants ${\bar t}$ such that ${\bar t} \in q_i(D)$ and ${\bar t} \not\in Q'(D)$. Clearly, there exists a homomorphism $h$ such that $h(q_i) \subseteq D$ and $h({\bar x}) = {\bar t}$. Observe also that $\rho(D_{q_i}) \subseteq D$, where $\rho = h \circ c^{-1}$. Towards a contradiction, assume that $c({\bar x}) \in Q'(D_{q_i})$. This implies that there exists a homomorphism $\gamma$ such that $\gamma(q') \subseteq \chase{D_{q_i}}{\dep}$ and $\gamma({\bar y}) = c({\bar x})$, where $Q' = (\insS,\dep,q'({\bar y}))$. It is not difficult to see that there exists an extension $\rho'$ of $\rho$ such that $\rho'(\chase{D_{q_i}}{\dep}) \subseteq \chase{D}{\dep}$ and $\rho'({\bar x}) = {\bar t}$. Hence, $\rho'(\gamma(q')) \subseteq \chase{D}{\dep}$, which implies that ${\bar t} \in q'(\chase{D}{\dep})$; thus, ${\bar t} \in Q'(D)$. But this contradicts the fact that ${\bar t} \not\in Q'(D)$, and the claim follows.
\end{proof}

Observe that $c({\bar x}) \in q(D_{q_i})$, which immediately implies that $c({\bar x}) \in Q(D_{q_i})$. Consequently, by Lemma~\ref{lem:auxiliary-small-witness-property}, $Q(D_{q_i}) \not\subseteq Q'(D_{q_i})$. The claim follows since, by construction, $D_{q_i}$ is an $\insS$-database such that $|D_{q_i}| \leq f_{\class{O}}(Q)$.

\subsection*{The Algorithm $\mathsf{XRewrite}$}

\begin{algorithm}[t]
\caption{The algorithm $\mathsf{XRewrite}$
\label{alg:tgd-rewrite}} \small
    \KwIn{An OMQ $Q = (\insS,\dep,q(\bar x)) \in (\class{TGD},\class{CQ})$}
    \KwOut{A UCQ $q'(\bar x)$ such that $Q(D) = q'(D)$, for every $\insS$-database $D$}
    \vspace{2mm}
    $i := 0$\;
    $Q_{\textsc{rew}} := \{\langle q,\mathsf{r},\mathsf{u} \rangle\}$\;
    \Repeat{$Q_{\textsc{temp}} = Q_{\textsc{rew}}$}{
        $Q_{\textsc{temp}} := Q_{\textsc{rew}}$\;
        \ForEach{$\langle q,x,\mathsf{u} \rangle \in Q_{\textsc{temp}},$ \emph{where} $x \in \{\mathsf{r},\mathsf{f}\}$}{
            \ForEach{$\sigma \in \Sigma$}{
                \tcc{\textrm{rewriting step}}
                \ForEach{$S \subseteq \body{q}$ \emph{such that} $\sigma$ \emph{is applicable to} $S$}{
                        $i := i + 1$\;
                        $q^{\prime} := \gamma_{S,\sigma^i}(q[S/\body{\sigma^i}])$\;
                        \If{\emph{there is no} $\tup{q'',\mathsf{r},\star} \in Q_{\textsc{rew}}$ \emph{such that} $q' \simeq q''$}{
                            $Q_{\textsc{rew}} := Q_{\textsc{rew}} \cup \{\langle q',\mathsf{r},\mathsf{u} \rangle\}$\;
                        }
                }
                \tcc{\textrm{factorization step}}
                \ForEach{$S \subseteq \body{q}$ \emph{that is factorizable w.r.t.} $\sigma$}{
                    $q^{\prime} := \gamma_{S}(q)$\;
                    \If{\emph{there is no} $\tup{q'',\star,\star} \in Q_{\textsc{rew}}$ \emph{such that} $q' \simeq q''$}{
                        $Q_{\textsc{rew}} := Q_{\textsc{rew}} \cup \{\langle q^{\prime},\mathsf{f},\mathsf{u} \rangle\}$\;
                    }
                }
            }
            \tcc{\textrm{query $q$ is now explored}}
            $Q_{\textsc{rew}} := (Q_{\textsc{rew}} \setminus \{\tup{q,x,\mathsf{u}}\}) \cup
            \{\tup{q,x,\mathsf{e}}\}$\;
        }
    }
    $Q_{\textsc{fin}} := \{q ~|~ \langle q,\mathsf{r},\mathsf{e} \rangle \in Q_{\textsc{rew}}, \text{ and } q \text{ contains only predicates of } \insS\}$\;
    \Return{$Q_{\textsc{fin}}$}
\end{algorithm}

In view of the fact that the rewriting algorithm $\mathsf{XRewrite}$ is heavily used in our complexity analysis, we would like to recall its definition. This algorithm is based on resolution, and thus, before we proceed further, we need to recall the crucial notion of unification.
A set of atoms $A = \{\alpha_1,\ldots,\alpha_n\}$, where $n \geqslant 2$, \emph{unifies} if there exists a substitution $\gamma$, called \emph{unifier} for $A$, such that $\gamma(\alpha_1) = \cdots = \gamma(\alpha_n)$. A \emph{most general unifier (MGU)} for $A$ is a unifier for $A$, denoted as $\gamma_{A}$, such that for each other unifier $\gamma$ for $A$, there exists a substitution $\gamma'$ such that $\gamma = \gamma' \circ \gamma_{A}$. Notice that if a set of atoms unify, then there exists a MGU. Furthermore, the MGU for a set of atoms is unique (modulo variable renaming).

The algorihtm proceeds by exhaustively applying two steps: {\em rewriting} and {\em factorization}, which in turn rely on the technical notions of {\em applicability} and {\em factorizability}, respectively.
We assume, w.l.o.g., that tgds and CQs do not share variables. Given a CQ $q$, a variable $x$ is called \emph{shared} in $q$ if $x$ is a free variable of $q$, or it occurs more than once in $q$.
In what follows, we assume, w.l.o.g., that tgds are in normal form, i.e., they have only one atom in the head, and only one occurrence of an existentially quantified variable~\cite{CaGP12}. We write $\pi_{\exists}(\sigma)$ for the position at which the existentially quantified variable of $\sigma$ occurs; in case $\sigma$ does not mention an existentially quantified variable, then $\pi_{\exists}(\sigma) = \varepsilon$. (Recall that a position $P[i]$ identifies the $i$-th attribute of a predicate $P$.)
We are now ready to recall applicability and factorizability; in what follows, we write $\body{q}$ for the set of atoms occurring in $q$, and $\head{\sigma}$ for the head-atom of $\sigma$.

\begin{definition}(Applicability)\label{def:applicability}
Consider a CQ $q$ and a tgd $\sigma$. Given a set of atoms $S
\subseteq \body{q}$, we say that $\sigma$ is \emph{applicable} to
$S$ if the following conditions are satisfied:
\begin{enumerate}
\item the set $S \cup \{\head{\sigma}\}$ unifies, and
\item for each $\alpha \in S$, if the term at position $\pi$ in
$\alpha$ is either a constant or a shared variable in $q$, then
$\pi \neq \pi_{\exists}(\sigma)$. \hfill\markfull
\end{enumerate}
\end{definition}

Roughly, whenever $\sigma$ is \emph{applicable} to $S$, this means that the atoms of $S$ may be generated during the chase procedure by applying $\sigma$. Therefore, we are allowed to apply a rewriting step (which is essentially a resolution step) that resolves $S$ using $\sigma$, i.e., $S$ is replaced by $\body{\sigma}$, and a new CQ that is closer to the input database is obtained.

If we start applying rewriting steps blindly, without checking for applicability, then the soundness of the rewriting procedure is not guaranteed. However, it is possible that the applicability condition is not satisfied, but still we should apply a rewriting step.
This may happen due to the presence of redundant atoms in a query. For example, given the CQ
\[
q\ =\ \exists x \exists y \exists z \, (R(x,y)\ \wedge\ R(x,z))
\]
and the tgd
\[
\sigma\ =\ P(u,v) \ra \exists w \, R(w,u)
\]
the applicability condition fails since the shared variable $x$ in $q$ occurs at the position $\pi_{\exists}(\sigma) = R[1]$. However, $q$ is essentially the CQ $ q = \exists x \exists y \, R(x,y)$, and now the applicability condition is satisfied.
From the above informal discussion, we conclude that the applicability condition may prevent the algorithm from being complete since some valid rewriting steps are blocked. Because of this reason, we need the so-called factorization step, which aims at converting some shared variables into non-shared variables, and thus, satisfy the applicability condition.
In general, this can be achieved by exhaustively unifying all the atoms that unify in the body of a CQ. However, some of these unifications do not contribute in any way to satisfying the applicability condition, and, as a result, many superfluous CQs are generated. It is thus better to apply a restricted form of factorization that generates a possibly small number of CQs
that are vital for the completeness of the rewriting algorithm. This corresponds to the identification of all the atoms in the query whose shared existential variables come from the same atom in the chase, and they can be unified with no loss of information. Summing up, the key idea underlying the notion of factorizability is as follows: in order to apply the factorization step, there must exist a tgd that can be applied to its output.\footnote{Let us clarify that for the purposes of the present work we can rely on the naive approach of exhaustively unifying all the atoms that unify in the body of a CQ. However, we would like to be consistent with~\cite{GoOP14}, where the algorithm $\mathsf{XRewrite}$ is proposed, and thus, we stick on the slightly more involved notion of factorizability.}

\begin{definition}(Factorizability)\label{def:factorizability}
Consider a CQ $q$ and a tgd $\sigma$. Given a set of atoms $S
\subseteq \body{q}$, where $|S| \geqslant 2$, we say that $S$ is
\emph{factorizable} w.r.t.~$\sigma$ if the following conditions are
satisfied:
\begin{enumerate}
\item $S$ unifies,
\item $\pi_{\exists}(\sigma) \neq \varepsilon$, and
\item there exists a variable $x \not\in \var{\body{q} \setminus S}$
that occurs in every atom of $S$ only at position
$\pi_{\exists}(\sigma)$. \hfill\markfull
\end{enumerate}
\end{definition}

Having the above key notions in place, we are now ready to recall the algorithm $\mathsf{XRewrite}$, which is depicted in Algorithm~\ref{alg:tgd-rewrite}. As said above, the UCQ rewriting of an OMQ $q = (\insS,\dep,q)$ is computed by exhaustively applying (i.e., until a fixpoint is reached) the rewriting and the factorization steps.
Notice that the CQs that are the result of the factorization step,
are nothing else than auxiliary queries which are critical for the
completeness of the final rewriting, but are not needed in the final
rewriting. Thus, during the iterative procedure, the queries are labeled with $\mathsf{r}$ (resp., $\mathsf{f}$) in order to keep track which of them are generated by the rewriting (resp., factorization) step. The CQ that is part of the input OMQ, although is not a result of the rewriting step, is labeled by $\mathsf{r}$ since it must be part of the final rewriting.
Moreover, once the two crucial steps have been exhaustively applied on a CQ $q$, it is not necessary to revisit $q$ since this will lead to redundant queries. Hence, the queries are also labeled with $\mathsf{e}$ (resp., $\mathsf{u}$) indicating that a query has been already explored (resp., is unexplored).
Let us now describe the two main steps of the algorithm. In the sequel, consider a triple $\tup{q,x,y}$, where $\tup{x,y} \in \{\mathsf{r},\mathsf{f}\} \times \{\mathsf{e},\mathsf{u}\}$ (this is how we indicate that $q$ is labeled by $x$ and $y$), and a tgd $\sigma \in \dep$. We assume that $q$ is of the form $\exists {\bar x} \, \varphi({\bar x},{\bar y})$.

\begin{description}
\item[\textbf{Rewriting Step.}] For each $S \subseteq \body{q}$ such
that $\sigma$ is applicable to $S$, the $i$-th application of the
rewriting step generates the query $q' =
\gamma_{S,\sigma^i}(q[S/\body{\sigma^i}])$, where $\sigma^i$ is the
tgd obtained from $\sigma$ by replacing each variable $x$ with
$x^i$, $\gamma_{S,\sigma^i}$ is the MGU for the set $S \cup
\{\head{\sigma^i}\}$ (which is the identity on the variables that
appear in the body but not in the head of $\sigma^i$), and
$q[S/\body{\sigma^i}]$ is obtained from $q$ be replacing $S$ with
$\body{\sigma^{i}}$.
By considering $\sigma^i$ (instead of $\sigma$) we basically rename,
using the integer $i$, the variables of $\sigma$. This renaming step
is needed in order to avoid undesirable clutters among the variables
introduced during different applications of the rewriting step.
Finally, if there is no $\tup{q'',\mathsf{r},\star} \in
Q_{\textsc{rew}}$, i.e., an (explored or unexplored) query that is
the result of the rewriting step, such that $q'$ and $q''$ are the
same (modulo bijective variable renaming), denoted $q' \simeq q''$,
then $\tup{q',\mathsf{r},\mathsf{u}}$ is added to $Q_{\textsc{rew}}$.\\

\item[\textbf{Factorization Step.}] For each $S \subseteq \body{q}$
that is factorizable w.r.t.~$\sigma$, the factorization step
generates the query $q' = \gamma_S(q)$, where $\gamma_S$ is the MGU
for $S$. If there is no $\tup{q'',\star,\star} \in Q_{\textsc{rew}}$, i.e., a query that is the result of the rewriting or the factorization step, and is explored or unexplored, such that $q' \simeq q''$, then $\tup{q', \mathsf{f},\mathsf{u}}$ is added to $Q_{\textsc{rew}}$.
\end{description}

\subsection*{Proof of Proposition~\ref{pro:function-nr}}

We assume, w.l.o.g., that the predicates of $\insS$ do not appear in the head of a tgd of $\dep$. Since $\dep \in \class{NR}$, by Lemma~\ref{lem:stratification}, $\dep$ admits a stratification $\{\dep_1,\ldots,\dep_n\}$ with stratification function $\mu : \sch{\dep} \ra \{0,\ldots,n\}$.
Let us briefly explain how the rewriting tree $T_Q$ of the OMQ $Q = (\insS,\dep,q)$ is defined. $T_Q$ is a rooted tree with $q$ being its root. The $i$-th level of $T_Q$ consists of the CQs obtained from the CQs of the $(i-1)$-th level by applying rewriting steps (see the algorithm $\mathsf{XRewrite}$ for details on the rewriting step) using only tgds from $\dep_{n-i+1}$. It is easy to verify that the CQs of the $i$-th level contain only predicates $P$ such that $\mu(P) < n-i+1$. It is now clear that the $n$-th level of $T_Q$ (i.e., the leaves of $T_Q$) consists only of CQs obtained during the execution of $\mathsf{XRewrite}(Q)$ that contain only predicates of $\insS$. Thus, in order to obtained the desired upper bound, it suffices to show that the number of atoms that occur in a CQ that is a leaf of $T_Q$ is at most $|q| \cdot \left(\max_{\tau \in \dep} \{|\body{\tau}|\}\right)^{|\sch{\dep}|}$.
To this end, let us focus on one branch $B$ of $T_Q$ from the root $q$ to a leaf $q'$. Such a branch can be naturally represented as a $k$-ary forest $F_Q^B$, where the root nodes are the atoms of $q$, and whenever an atom $\alpha$ is resolved during the rewriting step using a tgd $\tau$, the atoms of $\body{\tau}$, after applying the appropriate MGU, are the child nodes of $\alpha$. Therefore, to obtain the desired upper bound, it suffices to show that the number of leaves of $F_Q^B$ is at most $|q| \cdot \left(\max_{\tau \in \dep} \{|\body{\tau}|\}\right)^{|\sch{\dep}|}$.
By construction, $F_Q^B$ consists, in general, of $|q|$ $k$-ary rooted trees, where $k = \max_{\tau \in \dep} \{|\body{\tau}|\}$, of depth $n$. Hence, the number of leaves of $F_Q^B$ is at most $|q| \cdot \left(\max_{\tau \in \dep} \{|\body{\tau}|\}\right)^{n}$. Since $n \leq |\sch{\dep}|$, the claim follows.

\subsection*{Proof of Theorem~\ref{the:cont-nr}}

A proof sketch for the \text{\rm co}\NEXP$^{\textsc{NP}}$ upper bound is given in the main body of the paper. We proceed to establish the \text{\rm P}$^{\textsc{NEXP}}$-hardness.
Our proof is by reduction from a tiling problem that has been recently introduced in~\cite{EiLP16}, which in turn relies on the standard {\em Exponential Tiling Problem}. Let us first recall the latter problem.

An instance of the Exponential Tiling Problem is a tuple $(n,m,H,V,s)$, where $n,m$ are numbers (in unary), $H,V$ are subsets of $\{1,\ldots,m\} \times \{1,\ldots,m\}$, and $s$ is a sequence of numbers of $\{1,\ldots,m\}$.
Such a tuple specifies that we desire a $2^n \times 2^n$ grid, where each cell is tiled with a tile from $\{1,\ldots,m\}$. $H$ (resp., $V$) is the horizontal (resp., vertical) compatibility relation, while $s$ represents a constraint on the initial part of the first row of the grid.
A \emph{solution} to such an instance of the Exponential Tiling Problem is a function $f : \{0,\ldots,2^n-1\} \times \{0,\ldots,2^n-1\} \ra \{1,\ldots,m\}$ such that:
\begin{enumerate}
\item $f(i,0) = s[i]$, for each $0 \leq i \leq (|s|-1)$;

\item $(f(i,j),f(i+1,j)) \in H$, for each $0 \leq i \leq 2^n-2$ and $0 \leq j \leq 2^n-1$; and

\item $(f(i,j),f(i,j+1)) \in V$, for each $0 \leq i \leq 2^n-1$ and $0 \leq j \leq 2^n-2$.
\end{enumerate}
We will refer to $\{0,\ldots,2^n-1\} \times \{0,\ldots,2^n-1\}$ as a \emph{grid}, with the pairs in it being \emph{cells}. A cell consists of two \emph{coordinates}, the column-coordinate (for short col-coordinate) and the row-coordinate, and any function on a grid is a \emph{tiling}. The Exponential Tiling Problem is defined as follows: given an instance $T$ as above, decide whether $T$ has a solution. It is known that this problem is \NEXP-hard (see, e.g., Section~3.2 of~\cite{John90}).

We are now ready to recall the tiling problem introduced in~\cite{EiLP16}, called {\em Extended Tiling Problem} (ETP), which is \text{\rm P}$^{\textsc{NEXP}}$-hard. An instance of this problem is a tuple $(k,n,m,H_1,V_1,H_2,V_2)$, where $k,n,m$ are numbers (in unary), and $H_1,V_1,H_2,V_2$ are subsets of $\{1,\ldots,m\} \times \{1,\ldots,m\}$. The question is as follows: is it the case that for every sequence $s$, where $|s| = k$, of numbers of $\{1,\ldots,m\}$, $(n,m,H_1,V_1,s)$ has no solution {\em or} $(n,m,H_2,V_2,s)$ has a solution?

We give a reduction from the ETP to $\cont(\class{NR},\class{CQ})$. More precisely, given an instance $T = (k,n,m,H_1,V_1,H_2,V_2)$ of the ETP, our goal is to construct in polynomial time two queries $Q_i = (\insS,\dep_i,q_i) \in (\class{NR},\class{CQ})$, for $i \in \{1,2\}$, such that $T$ has a solution iff $Q_1 \subseteq Q_2$.

\subsubsection*{Data Schema $\mathbf{S}$}

The data schema $\mathbf{S}$ consists of:

\begin{itemize}
\item $0$-ary predicates $C_i^j$, for each $i \in \{0,\ldots,k-1\}$ and $j \in \{1,\ldots,m\}$; the atom $C_i^j$ indicates that $s_i = j$.
\end{itemize}

\subsubsection*{The Query $Q_1$}

The goal of the query $Q_1$ is twofold: (i) to check that the so-called {\em existence property} of the input database, i.e., for every $i \in \{0,\ldots,k-1\}$, there exists at least one atom of the form $C_i^j$, is satisfied, and (ii) to check whether $(n,m,H_1,V_1,s)$, where $s$ is the sequence of tilings encoded in the input database, has a solution. To this end, the query $Q_1$ will mention the following predicates:

\begin{itemize}
\item $0$-ary predicate $C_i$, indicating that there exists at least one atom of the form $C_i^j$ in the input database.

\item $0$-ary predicate ${\rm Existence}$, indicating that the input database enjoys the existence property.

\item Unary predicate ${\rm Tile}_i$, for each $i \in \{1,\ldots,m\}$; the atom ${\rm Tile}_i(x)$ states that $x$ is the tile $i$.

\item Binary predicate $H$; the atom $H(x,y)$ encodes the fact that $(x,y) \in H_1$.

\item Binary predicate $V$; the atom $V(x,y)$ encodes the fact that $(x,y) \in V_1$.

\item $5$-ary predicate $T_i$, for each $i \in \{1,\ldots,n\}$; the atom $T_i(x,x_1,x_2,x_3,x_4)$ states that $x$ is a $2^i \times 2^i$ tiling obtained from the $2^{i-1} \times 2^{i-1}$ tilings $x_1,\ldots,x_4$ -- details on the inductive construction of $2^i \times 2^i$ tilings from $2^{i-1} \times 2^{i-1}$ tilings are given below.

\item Unary predicate ${\rm Initial}_i$, for each $i \in \{0,\ldots,k-1\}$; the atom ${\rm Initial}_i(x)$ states that $s[i] = x$, i.e., the $i$-th element of the sequence $s$ is $x$.

\item Binary predicate ${\rm Top}_i^j$, for each $i \in \{1,\ldots,n\}$ and $j \in \{0,\ldots,k-1\}$; the atom ${\rm Top}_i^j(x,y)$ states that in the $2^i \times 2^i$ tiling $x$ the tile at position $(j,0)$ is $y$.

\item $0$-ary predicate ${\rm Tiling}$, indicating that there exists a $2^n \times 2^n$ tiling that is compatible with the initial tiling $s$ encoded in the input database.

\item $0$-ary predicate ${\rm Goal}$, which is derived whenever the predicates ${\rm Existence}$ and ${\rm Tiling}$ are derived.
\end{itemize}

\medskip

\medskip

$Q_1$ is defined as the query $(\insS,\dep_1,{\rm Goal})$, where $\dep_1$ consists of the following tgds:
\begin{itemize}
\item Checking for the existence property of the input database

\smallskip
For each $i \in \{0,\ldots,k-1\}$ and $j \in \{1,\ldots,m\}$:
\begin{eqnarray*}
C_i^j &\ra& C_i
\end{eqnarray*}
and the tgd that checks for the existence property
\begin{eqnarray*}
C_0,\ldots,C_{k-1} &\ra& {\rm Existence}
\end{eqnarray*}

\item Generate the tiles
\begin{eqnarray*}
&\ra& \exists x_1 \ldots \exists x_m \, ({\rm Tile}_1(x_1),\ldots,{\rm Tile}_m(x_m))
\end{eqnarray*}

\item Generate the compatibility relations

\smallskip

For each $(i,j) \in H_1$:
\begin{eqnarray*}
{\rm Tile}_i(x), {\rm Tile}_j(y) &\ra& H(x,y)
\end{eqnarray*}
For each $(i,j) \in V_1$:
\begin{eqnarray*}
{\rm Tile}_i(x), {\rm Tile}_j(y) &\ra& V(x,y)
\end{eqnarray*}

\begin{figure}[t]
 \epsfclipon
  \centerline
  {\hbox{
  \leavevmode
  \epsffile{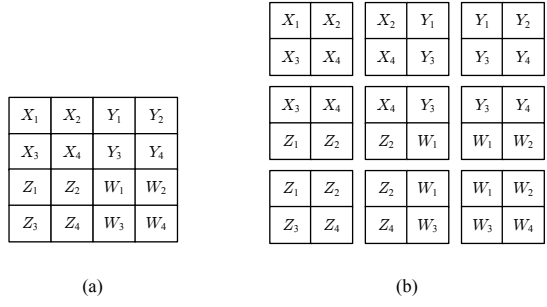}
  }} \epsfclipoff \caption{Inductive construction of tilings.}
  \label{fig:tilings}
\end{figure}

\item Generate the $2^n \times 2^n$ tilings. The key idea is to inductively construct $2^i \times 2^i$ tilings from $2^{i-1} \times 2^{i-1}$ tilings. It is easy to verify that the grid in Figure~\ref{fig:tilings}(a) is a $2^i \times 2^i$ tiling iff the nine subgrids of it, shown in Figure~\ref{fig:tilings}(b), are $2^{i-1} \times 2^{i-1}$ tilings. This has been already observed in~\cite{DaVo97}, where Datalog with complex values is studied.

\smallskip

First, we construct tilings of size $2 \times 2$ (the base case of the inductive construction):
\begin{eqnarray*}
H(x_1,x_2), H(x_3,x_4), V(x_1,x_3), V(x_2,x_4) &\ra& \exists x \, T_1(x,x_1,x_2,x_3,x_4)
\end{eqnarray*}
Then, we inductively construct tilings of larger size until we get tilings of size $2^n \times 2^n$. This is done using the following tgds. For each $i \in \{2,\ldots,n\}$:
\begin{align*}
&T_{i-1}(x_1,x_{11},x_{12},x_{21},x_{22}),T_{i-1}(x_2,x_{12},x_{13},x_{22},x_{23}),T_{i-1}(x_3,x_{13},x_{14},x_{23},x_{24})\\
&T_{i-1}(x_4,x_{21},x_{22},x_{31},x_{32}),T_{i-1}(x_5,x_{22},x_{23},x_{32},x_{33}),T_{i-1}(x_6,x_{23},x_{24},x_{33},x_{34}),\\
&T_{i-1}(x_7,x_{31},x_{32},x_{41},x_{42}),T_{i-1}(x_8,x_{32},x_{33},x_{42},x_{43}),T_{i-1}(x_9,x_{33},x_{34},x_{43},x_{44}) \ra\\
& \hspace{110mm} \exists x \, T_i(x,x_1,x_3,x_7,x_9)
\end{align*}

\item Extract from the $2^n \times 2^n$ tilings the tiles at positions $(0,0),(1,0),\ldots,(k-1,0)$. This is done using the following tgds:
\begin{eqnarray*}
T_1(x,x_1,x_2,x_3,x_4) &\ra& {\rm Top}_1^0(x,x_1),{\rm Top}_1^1(x,x_2)\\
T_2(x,x_1,x_2,x_3,x_4),{\rm Top}_1^0(x_1,y_0),{\rm Top}_1^1(x_1,y_1) &\ra& {\rm Top}_2^0(x,y_0),{\rm Top}_2^1(x,y_1)\\
T_2(x,x_1,x_2,x_3,x_4),{\rm Top}_1^0(x_2,y_0),{\rm Top}_1^1(x_2,y_1) &\ra& {\rm Top}_2^2(x,y_0),{\rm Top}_2^3(x,y_1)\\
&\vdots&\\
T_\ell(x,x_1,x_2,x_3,x_4),{\rm Top}_{\ell-1}^0(x_1,y_0), \cdots, {\rm Top}_{\ell - 1}^{2^{\ell-1}-1}(x_1,y_{2^{\ell-1}-1}) &\ra& {\rm Top}_{\ell}^0(x,y_0),\cdots,{\rm Top}_{\ell}^{2^{\ell-1}-1}(x,y_{2^{\ell-1}-1})\\
T_\ell(x,x_1,x_2,x_3,x_4),{\rm Top}_{\ell-1}^0(x_1,y_0), \cdots, {\rm Top}_{\ell - 1}^{k-2^{\ell-1}-1}(x_1,y_{k-2^{\ell-1}-1}) &\ra& {\rm Top}_{\ell}^{2^{\ell-1}}(x,y_0),\cdots,{\rm Top}_{\ell}^{k-1}(x,y_{k-2^{\ell-1}-1}),
\end{eqnarray*}
where $\ell = \lceil \log k \rceil$. Moreover, for each $i \in \{\ell+1,\ldots,n\}$:
\begin{eqnarray*}
T_i(x,x_1,x_2,x_3,x_4), {\rm Top}_{i-1}^{0}(x_1,y_0),\cdots,{\rm Top}_{i-1}^{k-1}(x_1,y_{k-1}) &\ra& {\rm Top}_{i}^{0}(x,y_0),\ldots,{\rm Top}_{i}^{k-1}(x,y_{k-1})
\end{eqnarray*}

\item Check whether there exists a $2^n \times 2^n$ tiling that is compatible with the sequence of tilings $s$

\smallskip

For each $i \in \{0,\ldots,k-1\}$ and $j \in \{1,\ldots,m\}$:
\begin{eqnarray*}
C_i^j, {\rm Tile}_j(x) &\ra& {\rm Initial}_i(x)
\end{eqnarray*}
and the tgd
\begin{eqnarray*}
{\rm Top}_{n}^{0}(x,y_0), {\rm Initial}_0(y_0), \cdots,
{\rm Top}_{n}^{k-1}(x,y_{k-1}), {\rm Initial}_{k-1}(y_{k-1}) &\ra& {\rm Tiling}
\end{eqnarray*}

\item Finally, we have the output tgd
\begin{eqnarray*}
{\rm Existence}, {\rm Tiling} &\ra& {\rm Goal}
\end{eqnarray*}
\end{itemize}
This concludes the construction of $Q_1$.

\subsubsection*{The Query $Q_2$}

The goal of the query $Q_2$ is twofold: (i) to check that the so-called {\em uniqueness property} of the input database, i.e., for every $i \in \{0,\ldots,k-1\}$, there exists at most one atom of the form $C_i^j$, is satisfied, and (ii) to check whether $(n,m,H_2,V_2,s)$, where $s$ is the sequence of tilings encoded in the input database, has a solution. The query $Q_2$ mentions the same predicates as $Q_1$, and is defined as $(\insS,\dep_2,{\rm Goal})$, where $\dep_2$ consists of the following tgds:

\begin{itemize}
\item Checking the uniqueness property

\smallskip

For each $i \in \{0,\ldots,k-1\}$ and $j,\ell \in \{1,\ldots,m\}$ with $j < \ell$:
\begin{eqnarray*}
C_i^j, C_i^\ell &\ra& {\rm Goal}
\end{eqnarray*}

\item The rest of $\dep_2$ encodes the tiling problem $(n,m,H_2,V_2,s)$ in exactly the same way as $\dep_1$ encodes $(n,m,H_1,V_1,s)$.
\end{itemize}
This concludes the construction of $Q_2$.

\subsection*{Proof of Proposition~\ref{pro:function-sticky-lower-bound}}

The set $\dep^n$ consists of the following tgds; for brevity, we write ${\bar x}_{i}^{j}$ for $x_i,x_{i+1},\ldots,x_j$:\footnote{A similar construction has been used in~\cite{GoOP14} for showing a lower bound on the size of a CQ in the UCQ rewriting of a $(\class{S},\class{CQ})$ OMQ.}
\begin{eqnarray*}
S(x_1,\ldots,x_n) &\ra& P_n(x_1,\ldots,x_n)\\
P_i({\bar x}_{1}^{i-1},z,{\bar x}_{i+1}^{n},z,o), P_i({\bar x}_{1}^{i-1},o,{\bar x}_{i+1}^{n},z,o) &\ra& P_{i-1}({\bar x}_{1}^{i-1},z,{\bar x}_{i+1}^{n},z,o), \quad 1 \leq i \leq n,\\
P_0(\underbrace{z,\ldots,z}_{n},z,o) &\ra& \text{\rm Ans}(z,o),
\end{eqnarray*}
while $q = \text{\rm Ans}(0,1)$.
It can be verified that, for every $\{S\}$-database $D$, $Q^{n}(D) \neq \emptyset$ implies that
\[
D\ \supseteq\ \{S(c_1,\ldots,c_{n-2},0,1) \mid (c_1,\ldots,c_{n-2}) \in \{0,1\}^{n-2}\},
\]
and thus, $|D| \geq 2^{n-2}$.
Let $Q = (\{S\},\dep',q')$, where $\dep'$ is a set of tgds and $q'$ a Boolean CQ, and $D$ an $\{S\}$-database. Clearly, $Q^{n}(D) \not\subseteq Q(D)$ iff $Q^{n}(D) \neq \emptyset$ and $Q(D) = \emptyset$. This implies that $|D| \geq 2^{n-2}$, and the claim follows.

\subsection*{Proof of Theorem~\ref{the:cont-sticky}}
The \text{\rm co}\NEXP~upper bound, as well as the $\Pi_{2}^{P}$-hardness in case of fixed-arity predicates, are discussed in the main body of the paper. Here, we show the \text{\rm co}\NEXP-hardness. The proof proceeds in two steps:
\begin{enumerate}
\item First, we show that $\cont((\class{FNR},\class{CQ}),(\class{L},\class{UCQ}))$ is \text{\rm co}\NEXP-hard, where $\class{FNR}$ denotes the class of full non-recursive sets of tgds, i.e., non-recursive sets of tgds without existentially quantified variables.

\item Then, we reduce $\cont((\class{FNR},\class{CQ}),(\class{L},\class{UCQ}))$ to $\cont((\class{S},\class{CQ}),(\class{L},\class{UCQ}))$ by showing that (under some assumptions that are explained below) every query in $(\class{FNR},\class{CQ})$ can be rewritten as an $(\class{S},\class{CQ})$ query.
\end{enumerate}
By Proposition~\ref{pro:from-ucq-to-cq}, we immediately get that $\cont((\class{S},\class{CQ}),(\class{L},\class{CQ}))$ is \text{\rm co}\NEXP-hard, as needed.

\subsection*{Step 1: $\cont((\class{FNR},\class{CQ}),(\class{L},\class{UCQ}))$ is \text{\rm co}\NEXP-hard}

We show that $\cont((\class{FNR},\class{CQ}),(\class{L},\class{UCQ}))$ is \text{\rm co}\NEXP-hard, even if we focus on {\em 0-1 queries}, that is, queries $Q$ with following property: for every database $D$, $Q(D) = Q(D_{01})$, where $D_{01} \subseteq D$ is the restriction of $D$ on the binary domain $\{0,1\}$, i.e., $D_{01} = \{R({\bar c}) \in D \mid {\bar c} \subseteq \{0,1\}\}$.
The proof is by reduction from the Exponential Tiling Problem, and is a non-trivial adaptation of the one given in~\cite{BeGo10} for showing that containment of non-recursive Datalog queries is co\NEXP-hard.

\begin{theorem}\label{the:fnr-into-lin}
$\cont((\class{FNR},\class{CQ}),(\class{L},\class{UCQ}))$ is \text{\rm co}\NEXP-hard, even for 0-1 queries.
\end{theorem}

\begin{proof}
Given an instance $T = (n,m,H,V,s)$ of the Exponential Tiling Problem, we are going to construct a $(\class{FNR},\class{CQ})$ 0-1 query $Q_{T} = (\insS,\dep,q)$ and a $(\class{L},\class{UCQ})$ 0-1 query $Q'_{T} = (\insS,\dep_T,q_T)$ such that $T$ has a solution iff $Q_T \not\subseteq Q'_T$.

\subsubsection*{Data Schema $\mathbf{S}$}

The data schema $\mathbf{S}$ consists of:

\begin{itemize}
%\item $0$-ary predicate ${\rm Start}$.
\item $2n$-ary predicates ${\rm TiledBy}_i$, for each $i \leq m$; the atom ${\rm TiledBy}_i(x_1,\ldots,x_n,y_1,\ldots,y_n)$ indicates that the cell with coordinates $((x_1,\ldots,x_n),(y_1,\ldots,y_n)) \in \{0,1\}^{n} \times \{0,1\}^n$ is tiled by tile $i$. Notice that we use $n$-bit binary numbers to represent a coordinate; this is the key difference between our construction and the one of~\cite{BeGo10}.
\end{itemize}

\subsubsection*{The Query $Q_T$}

The goal of the query $Q_T$ is to assert whether the input database encodes a candidate tiling, i.e., whether the entire grid is tiled, without taking into account the constraints, that is, the compatibility relations and the constraint on the initial part of the first row. To this end, the query $Q_T$ will mention the following predicates:
\begin{itemize}
\item Unary predicate ${\rm Bit}$; the atom ${\rm Bit}(x)$ simply says that $x$ is a bit, i.e., $x \in \{0,1\}$.

\item $2n$-ary predicate ${\rm TiledAboveCol}_i$, for each $i \leq n$; the atom ${\rm TiledAboveCol}_i({\bar x},{\bar y})$ says that for the row-coordinate ${\bar y}$ there are tiled cells with coordinates $({\bar x'},{\bar y})$ for every col-coordinate ${\bar x'}$ that agrees with ${\bar x}$ on the first $i-1$ bits. In other words, for the row corresponding to ${\bar y}$, every column extending the first $i-1$ bits of ${\bar x}$ is tiled. In particular, ${\rm TiledAboveCol}_1({\bar x},{\bar y})$ says that the entire row ${\bar y}$ is tiled.

\item $2n$-ary predicate ${\rm TiledAboveRow}_i$, for each $i \leq n$; the atom ${\rm TiledAboveRow}_i({\bar y})$ says that for every ${\bar y'}$ that agrees with ${\bar y}$ on the first $i-1$ bits, the row ${\bar y'}$ is fully tiled.

\item $n$-ary predicate ${\rm RowTiled}$; the atom ${\rm RowTiled}({\bar y})$ says that the row ${\bar y}$ is fully tiled.

\item $0$-ary predicate ${\rm AllTiled}$, which asserts that the entire grid is tiled.

\item $0$-ary predicate ${\rm Goal}$, which is derived whenever the predicate ${\rm AllTiled}$ is derived.
\end{itemize}

\medskip

$Q_T$ is defined as the query $(\insS,\dep,{\rm Goal})$, where $\dep$ consists of the following rules:
\begin{itemize}
\item Generate ${\rm Bit}$ atoms
\begin{eqnarray*}
&\ra& {\rm Bit}(0)\\
&\ra& {\rm Bit}(1).
\end{eqnarray*}

\item ${\rm RowTiled}$

\smallskip

For each $j,k \leq m$:
\begin{multline*}
{\rm TiledBy}_j(x_1,\ldots,x_{n-1},1,y_1,\ldots,y_n),
{\rm TiledBy}_k(x_1,\ldots,x_{n-1},0,y_1,\ldots,y_n),\\
{\rm Bit}(x_1), \ldots, {\rm Bit}(x_{n-1}), {\rm Bit}(y_1), \ldots, {\rm Bit}(y_n), {\rm Bit}(w) \ra\\
{\rm TiledAboveCol}_n(x_1,\ldots,x_{n-1},w,y_1,\ldots,y_n)
\end{multline*}

For each $2 \leq i \leq n$:
\begin{multline*}
{\rm TiledAboveCol}_i(x_1,\ldots,x_{i-1},1,x_{i+1},\ldots,x_n,y_1,\ldots,y_n),\\
{\rm TiledAboveCol}_i(x_1,\ldots,x_{i-1},0,x'_{i+1},\ldots,x'_n,y_1,\ldots,y_n),\\
{\rm Bit}(w_i), \ldots, {\rm Bit}(w_n) \ra\\
{\rm TiledAboveCol}_{i-1}(x_1,\ldots,x_{i-1},w_i,\ldots,w_n,y_1,\ldots,y_n)
\end{multline*}

A row is fully tiled:
\begin{eqnarray*}
{\rm TiledAboveCol}_{1}(x_1,\ldots,x_n,y_1,\ldots,y_n) &\ra& {\rm RowTiled}(y_1,\ldots,y_n)
\end{eqnarray*}

\item ${\rm AllTiled}$
\begin{eqnarray*}
{\rm RowTiled}(y_1,\ldots,y_{n-1},1),{\rm RowTiled}(y_1,\ldots,y_{n-1},0), {\rm Bit}(w) &\ra& {\rm TiledAboveRow}_{n}(y_1,\ldots,y_{n-1},w)
\end{eqnarray*}

For each $2 \leq i \leq n$:
\begin{multline*}
{\rm TiledAboveRow}_i(y_1,\ldots,y_{i-1},1,y_{i+1},\ldots,y_n),\\
{\rm TiledAboveRow}_i(y_1,\ldots,y_{i-1},0,y'_{i+1},\ldots,y'_n),\\
{\rm Bit}(w_i), \ldots, {\rm Bit}(w_n) \ra\\
{\rm TiledAboveRow}_{i-1}(y_1,\ldots,y_{i-1},w_i,\ldots,w_n)
\end{multline*}

The entire grid is tiled:
\begin{eqnarray*}
{\rm TiledAboveRow}_{1}(y_1,\ldots,y_n) &\ra& {\rm AllTiled}\\
{\rm AllTiled} &\ra& {\rm Goal}
\end{eqnarray*}
\end{itemize}
This concludes the construction of the query $Q_T$.

\subsubsection*{The Query $Q'_T$}

$Q'_T$ is defined in such a way that $Q'_T(D)$ is non-empty exactly when the input database $D$ encodes an invalid tiling, i.e., when one of the constraints on the tiles is violated. The query $Q'_T$ will mention the following intensional predicates:
\begin{itemize}
\item Unary predicate ${\rm Bit}$; as above, ${\rm Bit}(x)$ says that $x$ is a bit.

\item $2i$-ary predicate ${\rm LastFirst}_i$, for each $1 \leq i \leq n$; the atom ${\rm LastFirst}_i(x_1,\ldots,x_i,y_1,\ldots,y_i)$ says that $(x_1,\ldots,x_i) = (1,\ldots,1)$ and $(y_1,\ldots,y_i) = (0,\ldots,0)$.

\item $2i$-ary predicate ${\rm Succ}_i$, for each $1 \leq i \leq n$; the atom ${\rm Succ}_i({\bar x},{\bar y})$ says that the $i$-bit binary number ${\bar y}$ is the successor of the $i$-bit binary number ${\bar x}$.

\item $0$-ary predicate ${\rm Goal}$.
\end{itemize}

\medskip

$Q'_T$ is defined as the query $(\insS,\dep',q')$.
The set $\dep'$ consists of the following linear tgds:
\begin{itemize}
\item Generate ${\rm Bit}$ atoms:
\begin{eqnarray*}
&\ra& {\rm Bit}(0)\\
&\ra& {\rm Bit}(1).
\end{eqnarray*}

\item Generate the successor predicates:
\begin{eqnarray*}
&\ra& {\rm Succ}_1(0,1)\\
&\ra& {\rm LastFirst}_1(1,0).
\end{eqnarray*}
For each $1 \leq i \leq n-1$:
\begin{eqnarray*}
{\rm Succ}_i(x_1,\ldots,x_i,y_1,\ldots,y_i) &\ra& {\rm Succ}_{i+1}(0,x_1,\ldots,x_i,0,y_1,\ldots,y_i)\\
{\rm Succ}_i(x_1,\ldots,x_i,y_1,\ldots,y_i) &\ra& {\rm Succ}_{i+1}(1,x_1,\ldots,x_i,1,y_1,\ldots,y_i)\\
{\rm LastFirst}_i(x_1,\ldots,x_i,y_1,\ldots,y_i) &\ra& {\rm Succ}_{i+1}(0,x_1,\ldots,x_i,1,y_1,\ldots,y_i)\\
{\rm LastFirst}_i(x_1,\ldots,x_i,y_1,\ldots,y_i) &\ra& {\rm LastFirst}_{i+1}(1,x_1,\ldots,x_i,0,y_1,\ldots,y_i).
\end{eqnarray*}
\end{itemize}

The UCQ $q'$ consists of the following (Boolean) CQs; for brevity, the existential quantifiers in front of the CQs are omitted:

\begin{itemize}
\item Tile Consistency

\smallskip

For each $i \neq j \leq m$:
\begin{multline*}
{\rm TiledBy}_i(x_1,\ldots,x_n,y_1,\ldots,y_n), {\rm TiledBy}_j(x_1,\ldots,x_n,y_1,\ldots,y_n),\\
{\rm Bit}(x_1), \ldots, {\rm Bit}(x_n), {\rm Bit}(y_1), \ldots, {\rm Bit}(y_n)
\end{multline*}

\item Tile Compatibility

\smallskip

For each $(i,j) \not\in V$:
\begin{multline*}
{\rm Succ}_n(x_1,\ldots,x_n,y_1,\ldots,y_n),\\
{\rm TiledBy}_i(w_1,\ldots,w_n,x_1,\ldots,x_n), {\rm TiledBy}_i(w_1,\ldots,w_n,y_1,\ldots,y_n),\\
{\rm Bit}(w_1), \ldots, {\rm Bit}(w_n)
\end{multline*}

For each $(i,j) \not\in H$:
\begin{multline*}
{\rm Succ}_n(x_1,\ldots,x_n,y_1,\ldots,y_n),\\
{\rm TiledBy}_i(x_1,\ldots,x_n,w_1,\ldots,w_n), {\rm TiledBy}_i(y_1,\ldots,y_n,w_1,\ldots,w_n),\\
{\rm Bit}(w_1), \ldots, {\rm Bit}(w_n)
\end{multline*}

\item Tiling of First Row

For each $j \leq n$, let $f_j$ be the function from $\{1,\ldots,n\}$ into $\{0,1\}$ such that $f_j(1) \ldots f_j(n)$ is the number $j$ in binary representation, and let $k \in \{1,\ldots,m\}$ other than $s[j]$; recall that $s$ is a sequence of numbers of $\{1,\ldots,m\}$ that represents a constraint on the initial part of the first row of the grid. Then, we have the CQ:
\begin{eqnarray*}
{\rm TiledBy}_k(x_1,\ldots,x_n,\underbrace{z,\ldots,z}_{n}), {\rm Succ}_1(z,o)
\end{eqnarray*}
where, for each $i \in \{1,\ldots,n\}$, $x_i = z$ if $f_j(i) = 0$, and $x_i = o$ if $f_j(i) = 1$.
\end{itemize}
This concludes the definition of the query $Q'_T$.
\end{proof}

\subsection*{Step 2: $\cont((\class{S},\class{CQ}),(\class{L},\class{UCQ}))$ is \text{\rm co}\NEXP-hard}

Our goal is show that every 0-1 query $(\insS,\dep,q) \in (\class{F},\class{CQ})$ can be equivalently rewritten as a 0-1 query $(\insS,\dep',q')$, where all the tgds of $\dep'$ are {\em lossless}, i.e., all the body-variables appear also in the head, which in turn implies that $\dep'$ is sticky.

\begin{proposition}\label{f-to-sticky}
Consider a 0-1 query $Q \in (\class{F},\class{CQ})$. We can construct in polynomial time a 0-1 query $Q' \in (\class{S},\class{CQ})$ such that $Q \equiv Q'$.
\end{proposition}

\begin{proof}
Let $Q = (\insS,\dep,q)$, and assume that $n$ is the maximum number of variables occurring in the body of a tgd of $\dep$. We are going to construct in polynomial time a 0-1 query $Q' = (\insS,\dep',q') \in (\class{S},\class{CQ})$ such that $Q \equiv Q'$.

\medskip

The set $\dep'$ consists of the following tgds:

\begin{itemize}
\item Initialization Rules

\smallskip

We first transform every database atom of the form $R({\bar c})$ into an atom $R'({\bar c},\underbrace{0,\ldots,0}_{n},0,1)$. This is done as follows:
\begin{eqnarray*}
&\ra& {\rm Bit}(0)\\
&\ra& {\rm Bit}(1)
\end{eqnarray*}
and, for each $k$-ary predicate $R \in \mathbf{S}$, we have the lossless tgd
\begin{eqnarray*}
R(x_1,\ldots,x_k), {\rm Bit}(x_1), \ldots, {\rm Bit}(x_k) &\ra& R'(x_1,\ldots,x_k,\underbrace{0,\ldots,0}_{n})
\end{eqnarray*}
Notice that we can safely force the variables $x_1,\ldots,x_k$ to take only values from $\{0,1\}$ due to the 0-1 property.

\medskip

\item Transformation into Lossless Tgds

\smallskip
For each tgd $\sigma \in \dep$ of the form
\begin{eqnarray*}
R_1({\bar x_1}), \ldots, R_k({\bar x_k}) &\ra& R_0({\bar x_0})
\end{eqnarray*}
we have the lossless tgd
\begin{eqnarray*}
R'_1({\bar x_1},\underbrace{0,\ldots,0}_{n}), \ldots, R'_k({\bar x_k},\underbrace{0,\ldots,0}_{n}) &\ra& R'_0({\bar x_0},y_1,\ldots,y_n),
\end{eqnarray*}
where, if $\{v_1,\ldots,v_{\ell}\}$, for $\ell \in \{1,\ldots,n\}$, is the set of variables occurring in the body of $\sigma$ (the order is not relevant), then $y_i = v_i$, for each $i \in \{1,\ldots,\ell\}$, and $y_j = v_1$, for each $j \in \{\ell+1,\ldots,n\}$.

\medskip

\item Finalization Rules

\smallskip

Observe that each atom obtained during the chase due to one of the lossless tgds introduced above is of the form $R'({\bar x},{\bar y})$, where ${\bar y} \in \{0,1\}^n$. If ${\bar y} \neq (0,\ldots,0)$, then we need to ensure that eventually the atom
\[
R'({\bar x},\underbrace{0,\ldots,0}_{n})
\]
will be inferred. This is achieved by adding to $\dep'$ the following tgds: For each $k$-ary predicate $R$ occurring in $\dep$, and for each $1 \leq i \leq n$, we have the rule:
\begin{eqnarray*}
R'(x_1,\ldots,x_k,y_1,\ldots,y_{i-1},1,y_{i+1},\ldots,y_n) &\ra& R'(x_1,\ldots,x_k,y_1,\ldots,y_{i-1},0,y_{i+1},\ldots,y_n).
\end{eqnarray*}
\end{itemize}
This concludes the definition of $\dep'$.

\medskip

The CQ $q'$ is defined analogously. More precisely, assuming that $q$ is of the form (the existential quantifiers are omitted)
\begin{eqnarray*}
R_1({\bar x_1}), \ldots, R_k({\bar x_k})
\end{eqnarray*}
the CQ $q'$ is defined as
\begin{eqnarray*}
R'_1({\bar x_1},\underbrace{0,\ldots,0}_{n}), \ldots, R'_k({\bar x_k},\underbrace{0,\ldots,0}_{n}).
\end{eqnarray*}

\medskip

It is easy to verify that $\dep'$ consists of lossless tgds, and thus, $Q' \in (\class{S},\class{CQ})$. It also not difficult to see that, for every database $D$ over $\mathbf{S}$, $Q(D_{01}) = Q'(D_{01})$; thus, by the 0-1 property, $Q(D) = Q'(D)$, and the claim follows.
\end{proof}

By Theorem~\ref{the:fnr-into-lin} and Proposition~\ref{f-to-sticky}, we immediately get that $\cont((\class{S},\class{CQ}),(\class{L},\class{UCQ}))$ is \text{\rm co}\NEXP-hard, as needed.

\section*{PROOFS OF SECTION~\ref{sec:guardedness}}

Recall that, for the sake of technical clarity, we focus on constant-free tgds and CQs, but all the results can be extended to the general case at the price of more involved definitions and proofs.
Moreover, we assume that tgds have only one atom in the head. This does not affect the generality of our proof since every set of guarded tgds can be transformed in polynomial time into a set of guarded tgds with the above property; see, e.g.,~\cite{CaGK13}.
Finally, for convenience of presentation, we also assume that the body of a tgd is non-empty, i.e., the body of a tgd is always an atom and not the symbol $\top$.

%\begin{align*}
%  \varphi(\ve{x},\ve{y}) \limpl \exists\ve{z}\,\alpha(\ve{x},\ve{z}),
%\end{align*}
%where $\alpha$ is a relational atom and neither the head nor the body of
%the rule contains constants. We can easily rewrite a constant-free set
%of guarded rules $\Sigma$ into rules of the above form by introducing
%auxiliary predicates. Notice however that such a transformation may
%increase the maximum arity of $\sch{\Sigma}$. For convenience of
%presentation, we also assume that $\varphi(\ve{x},\ve{y})$ as above is not
%empty, i.e., does not equal $\top$ and that the CQs do not contain
%constants as well.

\subsection*{Proof of Proposition~\ref{pro:tree-witness-property}}

Let us start by recalling the key notion of tree decomposition.
%We first extend some definitions regarding notions on instances by a slight nuance.
Notice that the definition of the tree decomposition that we give here
is slightly different than the one in the main body of the paper. The
reason is because, for convenience of presentation, we prefer to employ a slightly different notation.
%This is due to notational convenience in the proofs that follow.

\begin{definition}
  Let $I$ be an instance. A \emph{tree decomposition of $I$ that omits $V$}, where $V \subseteq \adom{I}$, is a pair $\delta = \tup{\ca{T}, \tup{X_t}_{t \in T}}$, where $\ca{T} = \tup{T, E^{\ca{T}}}$ is a tree and $\tup{X_t}_{t \in T}$ a family of subsets of $\adom{I}$ (called the \emph{bags} of the decomposition) such that:
  \begin{enumerate}
  \item For every $v \in \adom{I} \setminus V$, the set
    $\set{t \in T \mid v \in X_t}$ is non-empty and connected.
  \item For every atom $P(s_1,\ldots,s_n) \in I$, there is a $t \in T$ such
    that $\set{s_1,\ldots,s_n} \subseteq X_t$.
  \end{enumerate}
  The \emph{width} of a tree decomposition
  $\delta = \tup{\ca{T},\tup{X_t}_{t\in T}}$ omitting $V$ is
  $\max\{|X_t| : t \in T\} - 1$. The \emph{tree-width} of $I$ is the
  minimum among the widths of all tree decompositions of $I$ that omit
  $V$. We call a tree decomposition omitting $\emptyset$ simply
  \emph{tree decomposition} of $I$. For $v \in T$, we denote by
  $I_\delta(v)$ the subinstance of $I$ induced by $X_v$.\hfill\markfull
\end{definition}

\textit{Notation.} We usually denote the strict partial order among the nodes of a tree $\ca{T}$ of a tree decomposition $\delta = \tup{\ca{T}, \tup{X_t}_{t \in T}}$ by $\prec$. Accordingly, we write $v \preceq w$ iff $v \prec w$ or $v = w$. For brevity, $\varepsilon$ will usually denote the root of a tree decomposition at hand. If ambiguities could possibly arise, we shall use subscripts in these notations. Furthermore, when $\delta$ is clear from
context, we shall omit it from the expression $I_\delta(v)$.

\medskip
Let $\delta = \tup{\ca{T},\tup{X_t}_{t \in T}}$ be a tree decomposition
of $I$ and $V \subseteq T$. Recall that $\delta$ is \emph{$[V]$-guarded}
(or \emph{guarded except for $V$}), if for every node
$v \in T \setminus V$, there is an atom $P(s_1,\ldots,s_n) \in I$ such
that $X_v \subseteq \set{s_1,\ldots,s_n}$. A $[\emptyset]$-guarded tree
decomposition of $I$ is simply called \emph{guarded tree decomposition}.

Also recall the crucial notion of $C$-tree:

\begin{definition}
  An $\sche{S}$-instance $I$ is a \emph{$C$-tree}, where
  $C \subseteq I$, if there is a tree decomposition
  $\delta = \tup{\ca{T},\tup{X_t}_{t \in T}}$ of $I$ such that
  \begin{enumerate}
  \item $I_\delta(\varepsilon) = C$, i.e., the subinstance of $I$
    induced by $X_\varepsilon$ equals $C$.
  \item $\delta$ is guarded except for $\set{\varepsilon}$.
  \end{enumerate}
  If $\delta$ or $C$ is clear from context, we shall often refer to
  $|\adom{C}|$ as the \emph{diameter} of $D$ and to $C$ as the
  \emph{core} of $D$. \hfill\markfull

  %The \emph{diameter of $D$ relative to
  %  $C$}, denoted $\diam{C}{D}$, is $|\adom{C}|$. Again, we may simply speak of
  %the diameter of $D$, if $C$ or $\delta$ is clear from context.
\end{definition}

\textit{Remark.} The notion of $C$-tree defined here refers to both
instances and databases, i.e., a $C$-tree may be a (finite) database or an
instance. We often do not explicitly mention whether a $C$-tree at hand
is a database or an instance. However, it will be clear from context
whether a $C$-tree is a database or an instance.

\medskip

We proceed to establish the following technical lemma, which in turn
allows us to show Proposition~\ref{pro:tree-witness-property}. It is an
adaption of a result in~\cite{BBL16} to the case of guarded
tgds. Henceforth, for brevity, given a query
$Q = \tup{\sche{S},\Sigma,q} \in (\class{G},\class{BCQ})$ and an
$\insS$-database $D$, we write $D \models Q$ for the fact that
$Q(D) \neq \emptyset$.

\begin{lemma}\label{lem:main-lemma}
  Let $Q = \tup{\sche{S},\Sigma,q}$ be an OMQ from
  $\tup{\class{G},\class{BCQ}}$. Let $D$ be an $\sche{S}$-database and
  suppose $D \models Q$. Then there is a finite $\sche{S}$-instance
  $\hat{I}$ such that $\hat{I} \models Q$ and:
  \begin{enumerate}
  \item $\hat{I}$ is a $C$-tree such that $|\adom{C}| \leq \arity{\sche{S} \cup \sch{\Sigma}}\cdot|q|$.
  \item There is a homomorphism from $\hat{I}$ to $D$.
  \end{enumerate}
\end{lemma}

Before we proceed with its formal proof, let us explain why
Proposition~\ref{pro:tree-witness-property} is an easy consequence of
Lemma~\ref{lem:main-lemma}.
The fact that the first item implies the second is trivial. Conversely,
suppose that $Q_1 \not\subseteq Q_2$, which implies that there exists an
$\insS$-database $D$ such that $D \models Q_1$ and $D \not\models Q_2$.
By Lemma~\ref{lem:main-lemma}, there exists a $C$-tree $\hat{I}$, where
$|\adom{C}| \leq \arity{\sche{S} \cup \sch{\Sigma_1}}\cdot|q_1|$, such
that $\hat{I} \models Q_1$. Moreover, there is a homomorphism from
$\hat{I}$ to $D$; hence, since $Q_2$ is closed under homomorphisms, it
immediately follows that $\hat{I} \not\models Q_2$.
Consequently, the $\insS$-database $\hat{D}$ obtained from $\hat{I}$
after replacing each null $z$ with a distinct constant $c_z$ is a
$C$-tree such that $Q_1(\hat{D}) \not\subseteq Q_2(\hat{D})$, and
Proposition~\ref{pro:tree-witness-property} follows.

\medskip

We now proceed with the proof of Lemma~\ref{lem:main-lemma} which is our
main task in this section. Before that, we introduce some additional
auxiliary concepts.

\subsubsection*{The Guarded Chase Forest}

Given a database $D$ and a set $\dep$ of guarded tgds, the
\emph{guarded chase forest for $D$ and $\Sigma$} is a forest (whose
edges and nodes are labeled) constructed as follows:
\begin{enumerate}
\item For each fact $R(\ve{a})$ in $D$, add a node labeled with
  $R(\ve{a})$.
\item For each node $v$ labeled with $\alpha \in \chase{D}{\Sigma}$ and
  for every atom $\beta$ resulting from a one-step application of a rule
  $\tau \in \Sigma$, if $\alpha$ is the image of the guard in this
  application of $\tau$, then add a node $w$ labeled with $\beta$ and
  introduce an arc from $v$ to $w$ labeled with $\tau$.
\end{enumerate}

We can assume that the guarded chase forest is always built inductively
according to a fixed, deterministic version of the chase procedure. The
non-root nodes are then totally ordered by a relation $\prec$ that
reflects their order of generation. Furthermore, we can extend $\prec$
to database atoms by picking a lexicographic order among them. Notice
that one atom can be the label of multiple nodes. Using the order
$\prec$ we can, however, always refer to the $\prec$-least node.

\subsubsection*{Guarded Unraveling}

Let $I$ be an instance over $\sche{S}$. We say that
$X \subseteq \adom{I}$ is \emph{guarded in $I$}, if there are
$a_1,\ldots,a_s \in \adom{I}$ such that
\begin{itemize}
  \item $X \subseteq \set{a_1,\ldots,a_s}$ and
  \item there is an $R/s \in \sche{S}$ such that
    $I \models R(a_1,\ldots,a_s)$.
\end{itemize}
A tuple $\ve{t}$ is guarded in $I$ if the set containing the elements of
$\ve{t}$ is guarded in $I$.

In the following paragraph, we largely follow the notions introduced
in~\cite{ABeBB16,BeBB16}. Fix an $\sche{S}$-instance
$I$ and some $X_0 \subseteq \adom{I}$. Let $\Pi$ be the set of finite
sequences of the form $X_0X_1\cdots X_n$, where, for $i > 0$, $X_i$ is a
guarded set in $I$, and, for $i \geq 0$, $X_{i + 1} = X_i \cup \set{a}$
for some $a \in \adom{I} \setminus X_i$, or $X_i \supseteq X_{i+1}$. The
sequences from $\Pi$ can be arranged in a tree by their natural prefix
order and each sequence $\pi = X_0X_1\cdots X_n$ identifies a unique
node in this tree. In this context, we say that $a \in \adom{I}$ is
\emph{represented at $\pi$} whenever $a \in X_n$. Two sequences
$\pi,\pi'$ are \emph{$a$-equivalent}, if $a$ is represented at each
node on the unique shortest path between $\pi$ and $\pi'$. For $a$
represented at $\pi$, we denote by $[\pi]_a$ the $a$-equivalence class
of $\pi$. The \emph{guarded unraveling around $X_0$} is the instance
$I^\ast$ over the elements $\set{[\pi]_a \mid \text{$a$ is represented
    at $\pi$}}$, where
\begin{align*}
  I^\ast \models R([\pi_1]_{a_1},\ldots,[\pi_n]_{a_n})\ \defequ\ \ &\text{$I \models R(a_1,\ldots,a_n)$ and}\\
                                                                     &\text{$\exists\pi \in \Pi, \forall i \in \set{1,\ldots,n}\colon [\pi]_{a_i} = [\pi_i]_{a_i}$,}
\end{align*}
for all $R/n \in \sche{S}$.
\begin{lemma}
  \label{lem:unraveling}
  For every $\sche{S}$-instance $I$ and any $X_0 \subseteq \adom{I}$,
  the guarded unraveling $I^\ast$ around $X_0$ is a $C$-tree over
  $\sche{S}$, where $C$ is the subinstance of $I^\ast$ induced by the
  elements $\set{[X_0]_a \mid a \in X_0}$.
\end{lemma}
\begin{proof}
  Let $\delta = \tup{\ca{T},\tup{X_t}_{t \in T}}$, where $\ca{T}$ is the
  natural tree that arises from ordering the sequences in $\Pi$ by their
  prefixes. For $\pi \in T$, let
  $X_\pi \coloneqq \set{[\pi]_a \mid \text{$a$ is represented at
      $\pi$}}$. Let $\varepsilon$ denote the root of $\ca{T}$. We need
  to show that $\delta$ is an appropriate tree decomposition witnessing
  that $I^\ast$ is a $C$-tree. First, note that it is clear that
  $I(\varepsilon) = \set{[X_0]_a \mid a \in X_0}$ by construction. Let
  $[\pi]_a \in \adom{I}$ and consider the set
  $A \coloneqq \set{t \in T \mid [\pi]_a \in X_t}$. This set is
  certainly non-empty. Moreover, for $t_1,t_2 \in A$, we know that
  $[t_1]_a = [t_2]_a$, hence $t_1$ and $t_2$ are $a$-connected in
  $\ca{T}$. Suppose
  $I^\ast \models R([\pi_1]_{a_1},\ldots,[\pi_n]_{a_n})$ for some
  $R/n \in \sche{S}$. Then there is a $\pi \in T$ such that
  $[\pi]_{a_i} = [\pi_i]_{a_i}$, for all $i = 1,\ldots,n$. Hence,
  $a_1,\ldots,a_n$ are all represented at $\pi$ and so
  $\set{[\pi_1]_{a_1},\ldots,[\pi_n]_{a_n}} \subseteq X_\pi$. It remains to show that $\delta$ is guarded except for
  $\set{\varepsilon}$. Let $\pi \neq \varepsilon$ and consider the set
  $X_t$. Since $\pi$ is a sequence of length greater than one, its last
  element $Y$ is a guarded set in $I$. Hence, there are $a_1,\ldots,a_s$
  such that $Y \subseteq \set{a_1,\ldots,a_s}$ and
  $I \models R(a_1,\ldots,a_s)$ for some $R/s \in \sche{S}$. Let
  $\set{a_1,\ldots,a_s} \setminus Y = \set{b_1,\ldots,b_m}$ and define
  $\rho \coloneqq \pi \cdot (Y \cup \set{b_1}) \cdot (Y \cup
  \set{b_1,b_2}) \cdots (Y \cup \set{b_1,\ldots,b_m})$. Then
  $I^\ast \models R([\rho]_{a_1},\ldots,[\rho]_{a_s})$, as desired.
\end{proof}
Notice that this lemma implies that the tree-width of $I^\ast$ is
bounded by $|X_0| + \arity{\sche{S}} - 1$.

We are now ready to prove Lemma~\ref{lem:main-lemma}:

\medskip
\begin{proofcustom}{of Lemma~\ref{lem:main-lemma}}
  Let $q = \exists\ve{y}\, \varphi(\ve{y})$ and $\mu$ a homomorphism
  mapping $\varphi(\ve{y})$ to $\chase{D}{\Sigma}$. Let
  $R_1(\ve{b}_1),\ldots,R_k(\ve{b}_k)$ exhaust all facts from $D$ that
  are the roots of those $\prec$-least facts from $\mu(\varphi(\ve{y}))$
  in the guarded chase forest of $D$ and $\Sigma$ that have an element
  from $\adom{D}$ as argument. Let
  $G_\mu \coloneqq \bigcup_{1 \leq i \leq k}\set{\ve{b}_i}$ and let
  $I^\ast$ be the unraveling of $D$ around $G_\mu$, regarding all
  elements from $\adom{I^\ast}$ as labeled nulls. Henceforth, for every
  $a \in G_\mu$, we denote by $\lambda_a$ the element $[G_\mu]_a$. We
  say that $\lambda_a$ \emph{represents} $a$. Let $C$ be the
  substructure of $I^\ast$ induced by the set
  $\set{\lambda_a \mid a \in G_\mu}$. Notice that $I^\ast$ is an
  infinite instance that is a $C$-tree by Lemma~\ref{lem:unraveling}. We
  will show later how to get a finite instance from $I^\ast$ that
  satisfies our constraints. We proceed to show that
  $I^\ast \cup \Sigma$ logically entails $q$, denoted
  $I^\ast,\Sigma \models q$:

  \begin{lemma}\label{lem:unravelingentailsq}
    $I^\ast,\Sigma \models q$.
  \end{lemma}
  \begin{proof}
    We will first construct a universal model $J$ of $I^\ast$ and
    $\Sigma$. Recall that an instance $U$ is a \emph{universal model} of
    $I$ and $\Sigma$, if it can be homomorphically mapped to every model
    of $I \cup \Sigma$; in particular, it is well-known and easy to
    prove that $\chase{I}{\Sigma}$ is always a universal model of $I$
    and $\Sigma$. Before constructing $J$, we introduce some additional
    notions. In the following, given a guarded set
    $G = \set{a_1,\ldots,a_k}$ in $D$, a \emph{copy of $G$ in $I^\ast$}
    is a set $\Gamma = \set{\alpha_1,\ldots,\alpha_k}$ which is guarded
    in $I^\ast$ such that, for $i = 1,\ldots,k$, we have that
    $\alpha_i = [\pi_i]_{a_i}$ for some sequences $\pi_i$ and
    $D \models R(a_{i_1},\ldots,a_{i_m})$ iff
    $I^\ast \models R(\alpha_{i_1},\ldots,\alpha_{i_m})$ for all
    $R \in \sche{S}$ and $i_j \in \set{1,\ldots,k}$. Copies of guarded
    tuples are defined accordingly. Consider the structure
    $\chase{D}{\Sigma}$. Let $G$ be a guarded set in $D$ and
    $D \upharpoonright G$ denote the subinstance of $D$ induced by
    $G$. It is well-known and easy to prove that
    $\chase{D\upharpoonright G}{\Sigma}$ is acyclic (cf.,
    e.g.,~\cite{CaGK08a}). Henceforth, we loosely call
    $\chase{D \upharpoonright G}{\Sigma}$ the \emph{tree attached to
      $G$}. The model $J$ is constructed as follows. Let $J_0$ be the
    instance $C$. Furthermore, for each guarded set
    $G = \set{a_1,\ldots,a_k}$ in $D$ and each copy
    $\Gamma = \set{\alpha_1,\ldots,\alpha_k}$ of $G$ in $I^\ast$,
    construct a new instance $J_\Gamma$ that is isomorphic to the tree
    attached to $G$ such that
  \begin{enumerate*}[label={(\roman*)}]
  \item the elements $a_i$ of $G$ are renamed to $\alpha_i$ in $J_\Gamma$,
  \item
    $\adom{J_0} \cap \adom{J_\Gamma} =
    \set{\alpha_1,\ldots,\alpha_k}$, and
  \item $\Gamma \cap \Theta = \adom{J_\Gamma} \cap \adom{J_\Theta}$, for every copy
    $\Theta$ of $G$ in $I^\ast$.
  \end{enumerate*}
  The model $J$ is the union of $J_0$ and all the
  $J_\Gamma$. If a guarded set $X$ in $J_\Gamma$ arises from
  renaming elements of a guarded set $Y$ in
  $\chase{D \upharpoonright G}{\Sigma}$, we also say that $X$ is a
  copy of $Y$ in $J$. Furthermore, the copies of $D$ that are
  contained in $I^\ast$ (i.e., in $J_0$) are also called
  copies in $J$. Observe that $J$ is a model of $I^\ast$
  by construction. We show that it is a model of $\Sigma$. To this end,
  we show the following claim.
  \begin{claim}
    Let $\ve{t}$ be a guarded tuple in $J$ and let $q(\ve{x})$ be a
    guarded conjunctive query\footnote{By a \emph{guarded conjunctive
        query} we mean here a CQ that contains an atom that contains all
        the variables occurring in the CQ as argument.} over
      $\sche{S} \cup \sch{\Sigma}$. Suppose $\ve{t}$ is a copy of
      $\ve{s}$ in $J$, where $\ve{s}$ is over $\adom{\chase{D}{\Sigma}}$
      and $|\ve{t}| = |\ve{s}|$. Then $J \models q(\ve{t})$ iff
      $\chase{D}{\Sigma} \models q(\ve{s})$.
  \end{claim}
  \begin{proof}
    Suppose $J \models q(\ve{t})$. Let
    $\set{\ve{t}} = \set{\alpha_1,\ldots,\alpha_k}$ be a copy of
    $\set{\ve{s}} = \set{a_1,\ldots,a_k}$ in $J$. Since $q(\ve{x})$
    is guarded, there is a
    $\Gamma \supseteq (\set{\alpha_1,\ldots,\alpha_k} \cap
    \adom{J_0})$ such that $J_\Gamma\models q(\ve{t})$. Let
    $G \supseteq \set{\ve{s}}$ be the guarded set in $D$ of which
    $\Gamma$ is a copy in $J_0$. It clearly holds that
    $\chase{D \upharpoonright G}{\Sigma} \models q(\ve{s})$, whence
    $\chase{D}{\Sigma} \models q(\ve{s})$ follows.

    Suppose that $\chase{D}{\Sigma} \models q(\ve{s})$. Let
    $\ve{t} = \alpha_1,\ldots,\alpha_k$ and $\ve{s} = a_1,\ldots,a_k$
    and suppose that $\alpha_i = [\pi_i]_{a_i}$ ($i = 1,\ldots,k$). The
    set $\set{a_1,\ldots,a_k}$ is guarded in
    $\chase{D}{\Sigma}$. Hence, there is a guarded
    $G \supseteq \set{a_1,\ldots,a_k} \cap \adom{D}$ in $D$
    such that
    $\chase{D \upharpoonright G}{\Sigma} \models q(\ve{s})$. We
    show that there is a
    $\Gamma \supseteq \set{\alpha_1,\ldots,\alpha_k} \cap \adom{I^\ast}$
    which is a copy of $G$ in $I^\ast$. Suppose
    $G = \set{b_1,\ldots,b_l}$. Let $\pi = X_0X_1\cdots X_m$ be such
    that $[\pi]_{a_i} = [\pi_i]_{a_i}$ for all $i = 1,\ldots,k$. For
    $i = 1,\ldots,l$, define
    \begin{align*}
      \rho_i \coloneqq \pi \cdot (X_m \cup \set{b_1}) \cdot (X_m \cup \set{b_1,b_2}) \cdots
      (X_m \cup \set{b_1,\ldots,b_i}).
    \end{align*}
    Then $b_i$ is represented at $\rho_i$. For $i = 1,\ldots,l$, let
    $\beta_i \coloneqq [\rho_i]_{b_i}$. We claim that
    $\Gamma \coloneqq \set{\beta_1,\ldots,\beta_l}$ is a copy of $G$ in
    $I^\ast$. Let $R/s \in \sche{S}$ and suppose
    $I^\ast \models
    R([\rho_{i_1}]_{b_{i_1}},\ldots,[\rho_{i_s}]_{b_{i_s}})$. Then we
    immediately obtain $D \models
    R(b_{i_1},\ldots,b_{i_s})$. Conversely, if
    $D \models R(b_{i_1},\ldots,b_{i_s})$, let
    $\rho \coloneqq \rho_\ell$, where
    $\ell \coloneqq \max\set{i_1,\ldots,i_s}$. Take any
    $j \in \set{i_1,\ldots,i_s}$. It is easy to see that $\rho$ and
    $\rho_j$ are $b_j$-equivalent. Hence,
    $[\rho_j]_{b_j} = [\rho]_{b_j}$ and it follows that
    $I^\ast \models
    R([\rho_{i_1}]_{b_{i_1}},\ldots,[\rho_{i_s}]_{b_{i_s}})$, as
    required. It follows that $\Gamma$ is a copy of $G$ in $I^\ast$ and
    so there is a structure $J_\Gamma$ contained in $J$ that is
    isomorphic to $\chase{D \upharpoonright G}{\Sigma}$ with
    $b_1,\ldots,b_l$ respectively renamed to
    $\beta_1,\ldots,\beta_l$. Hence, $J \models q(\ve{t})$ as required.
  \end{proof}

  Now let
  $\sigma\colon \varphi(\ve{x},\ve{y}) \limpl
  \exists\ve{z}\,\alpha(\ve{x},\ve{z})$ be a guarded rule from
  $\Sigma$. Suppose that
  $J \models \exists\ve{y}\,\varphi(\ve{t},\ve{y})$. Since every guarded
  tuple in $J$ is a copy of some guarded tuple in $\chase{D}{\Sigma}$,
  there is an $\ve{s}$, of which $\ve{t}$ is a copy, such that
  $\chase{D}{\Sigma} \models
  \exists\ve{y}\,\varphi(\ve{s},\ve{y})$. Since $\chase{D}{\Sigma}$ is a
  model of $\Sigma$, we know that
  $\chase{D}{\Sigma} \models \exists\ve{z}\,\alpha(\ve{s},\ve{z})$. It
  follows that $J \models \exists\ve{z}\,\alpha(\ve{s},\ve{z})$ by the
  above claim, as required. It remains to show that $J$ is
  universal:
  \begin{claim}
    $J$ is universal.
  \end{claim}
  \begin{proof}
    It suffices to show that $J$ can be homomorphically mapped to
    $\chase{I^\ast}{\Sigma}$ via a homomorphism $\eta$. We let $\eta_0$
    be the homomorphism that maps every element of $J_0$ to itself. It
    remains to treat the structures $J_\Gamma$. Consider a copy
    $\Gamma = \set{\alpha_1,\ldots,\alpha_k}$ in $I^\ast$ of a set
    $G = \set{b_1,\ldots,b_k}$ which is guarded in $D$. It suffices to
    show that $J_\Gamma$ can be mapped to $\chase{I^\ast}{\Sigma}$. To
    this end, it we show how to map
    $\chase{D \upharpoonright G}{\Sigma}$ to
    $\chase{I^\ast}{\Sigma}$. We do so by induction on the number of
    rule applications of $\chase{D \upharpoonright G}{\Sigma}$. For the
    base case, we map $D\upharpoonright G$ to $I^\ast$ as follows. Let
    $\eta_G^0(b_i) \coloneqq \alpha_i$, for $i = 1,\ldots,k$. Suppose
    $D \upharpoonright G \models R(b_{i_1},\ldots,b_{i_l})$ for some
    $R \in \sche{S}$ and $i_j \in \set{1,\ldots,k}$, where
    $j = 1,\ldots,l$. Recall that $\Gamma$ is guarded in
    $I^\ast$. Reviewing the construction of $I^\ast$, it is easy to see
    that this holds iff
    $I^\ast \models R(\alpha_{i_1},\ldots,\alpha_{i_l})$. Hence,
    $\eta_G^0$ is indeed a homomorphism from $D \upharpoonright G$ to
    $I^\ast$. The induction step is obvious---we can easily obtain a
    homomorphism $\eta_G^i$ that maps
    {$\pchase{k}{D \upharpoonright G}{\Sigma}$} to
    $\chase{I^\ast}{\Sigma}$. The desired homomorphism $\eta_G$ is the
    union of the {$\eta_G^i$} ($i \geq 0$). We then obtain a
    homomorphism $\eta_\Gamma$ from $\eta_G$ by appropriately renaming
    the elements from the domain of the latter as we did in the
    construction of $J_\Gamma$---which is nothing else than an
    isomorphic copy of $\chase{D \upharpoonright
      G}{\Sigma}$. Furthermore, each of these homomorphisms maps each
    element of $\Gamma$ to itself. The desired homomorphism $\eta$ that
    witnesses that $J$ is universal is the union of $\eta_0$ and the
    $\eta_\Gamma$.
  \end{proof}

  In order to prove $I^\ast,\Sigma \models q$, it remains to show that there is a homomorphism $\hat{\mu}$ that maps $q$ to $J$. There
  are guarded sets $G_1,\ldots,G_l$ in $D$ such that $\mu$ can be
  understood to map $q$ to
  $\chase{\bigcup_{1 \leq i \leq l}(D \upharpoonright G_i)}{\Sigma}$. By
  construction, we know that $G_1,\ldots,G_l$ can be chosen in such a
  way that {$G_\mu \subseteq \bigcup_{i = 1}^l G_i$}. Since $\Sigma$ is
  guarded, $\mu$ can be understood to map $q$ to
  $\bigcup_{1 \leq i \leq l} \chase{D \upharpoonright
    G_i}{\Sigma}$---assuming that the labeled nulls occurring in these
  instances are mutually new. Let
  $\ca{C}_\mu \coloneqq \set{\set{\ve{b}_1},\ldots,\set{\ve{b}_k}}$. For
  every $X \in \ca{C}_\mu$, let
  $\Gamma_X \coloneqq \set{\lambda_b \mid b \in X}$. Notice that
  $\Gamma_X$ is a copy of $X$ in $I^\ast$. By construction, all the
  facts from $q$ that are mapped via $\mu$ to $\chase{D}{\Sigma}$ and
  which have an element from $\adom{D}$ in their image under $\mu$ are
  already mapped to
  $\bigcup_{X \in \ca{C}_\mu} \chase{D\upharpoonright X}{\Sigma}$. For
  the other facts, the names of the constants in the databases do not
  matter.\footnote{Here, it is of course essential to assume
    constant-free rules.} Let {$\Theta_1,\ldots,\Theta_s$} be arbitrary
  copies of the sets {$\set{G_1,\ldots,G_l} \setminus \ca{C}_\mu$} in
  {$I^\ast$}. It follows that we can find our desired match $\hat{\mu}$
  in the union of $\bigcup_{X \in \ca{C}_\mu}J_{\Gamma_X}$ and
  $\bigcup_{1 \leq i \leq s} J_{\Theta_i}$. Notice that
  $\bigcup_{X \in \ca{C}_\mu}J_{\Gamma_X}$ is isomorphic to
  $\bigcup_{X \in \ca{C}_\mu} \chase{D \upharpoonright X}{\Sigma}$ with
  each $b \in G_\mu$ represented by $\lambda_b$.
  \end{proof}

  Now the database $I^\ast$ has the desired form with $C$ being its
  core. However, $I^\ast$ is infinite. Since $I^\ast,\Sigma \models q$
  due to Lemma~\ref{lem:unravelingentailsq}, by compactness, there is a
  finite {$\hat{B} \subseteq I^\ast$} such that
  {$\hat{B},\Sigma \models q$}. Consider a tree decomposition
  {$\delta = \tup{\ca{T}, \tup{X_t}_{t\in T}}$} witnessing that $I^\ast$
  is a $C$-tree. There is a maximum $\ell$ such that {$\hat{B}$}
  contains all the subinstances induced by the bags of depth less or
  equal $\ell$. Let $\hat{I}$ be the instance that actually contains
  \emph{all} the subinstances induced by the bags of level up to
  $\ell$. Hence, {$\hat{I}$} is itself a $C$-tree and
  $\hat{I}, \Sigma \models q$, since $\hat{B} \subseteq \hat{I}$.

  Now there is a natural homomorphism mapping {$\hat{I}$} to $D$: we
  simply specify $[\pi]_a \mapsto a$ for all $a \in \adom{D}$. The
  instance {$\hat{I}$} is the one we are looking for.
\end{proofcustom}
%

%\subsection*{Finalizing the Proof}
%
%
%
%Now Proposition~\ref{pro:tree-witness-property} easily follows:
%
%\medskip
%\begin{proofcustom}{of Proposition~\ref{pro:tree-witness-property}}
%  The fact that the first item implies the second is trivial. Suppose
%  $Q_1 \not\subseteq Q_2$. Let $D$ be an $\sche{S}$-database such that
%  $D \models Q_1$ and $D \not\models Q_2$. Then there is a $C$-tree
%  $\hat{I}$ where $|\adom{C}|$ is bounded as in the statement of the
%  theorem. Since there is a homomorphism from $\hat{I}$ to $D$ and since
%  $Q_2$ is closed under homomorphisms, it immediately follows that
%  $\hat{I} \not\models Q_2$.
%\end{proofcustom}

\subsection*{Proof of Lemma~\ref{lem:consistent-labeled-trees}}

One can naturally encode instances of bounded tree-width into trees over
a finite alphabet such that the alphabet's size depends only on the
tree-width. Our goal here is to appropriately encode $C$-trees in order
to make them accessible to tree automata techniques. Since the
tree-width of a $C$-tree over $\sche{S}$ depends only on the size of
$\adom{C}$ and the maximum arity of $\sche{S}$, the alphabet of the
encoding will depend on the same.

\medskip
\noindent
\textbf{Labeled trees.} Let $\Gamma$ be an alphabet and
$(\mbb{N} \setminus \set{0})^\ast$ be the set of finite sequences of
positive integers, including the empty sequence
$\varepsilon$.\footnote{We specify that $0$ is included in $\mbb{N}$ as
  well.} Let us recall that a \emph{$\Gamma$-labeled tree} is a pair
$t = \tup{T, \mu}$, where $\mu \colon T \ra \Gamma$ is the
\emph{labeling function} and
$T \subseteq (\mbb{N} \setminus \set{0})^\ast$ is closed under prefixes,
i.e., $x \cdot i \in T$ implies $x \in T$, for all
$x \in (\mbb{N} \cup \set{0})^\ast$ and $i \in (\mbb{N} \cup \set{0})$.
The elements contained in $T$ identify the \emph{nodes} of $t$. For
$i \in \mbb{N} \setminus \set{0}$, nodes of the form $x \cdot i \in T$
are the \emph{children} of $x$.  A \emph{path} of length $n$ in $T$ from
$x$ to $y$ is a sequence of nodes $x = x_1,\ldots,x_n = y$ such that
$x_{i+1}$ is a child of $x_i$. A \emph{branch} is a path from the root
to a leaf node. For $x \in T$, we set $x \cdot i \cdot -1 \coloneqq x$,
for all $i \in \mbb{N}$, and $x \cdot 0 \coloneqq x$---notice that
$\varepsilon \cdot -1$ is not defined.

\medskip
\noindent
\textbf{Encoding.} Let $l \geq 0$ and fix a schema $\sche{S}$. Let
$U_{\sche{S},l}$ be the disjoint union of two sets $C_l$ and
$T_{\sche{S}}$, respectively containing $l$ and $2 \cdot \arity{\sche{S}}$
elements. The elements from $U_{\sche{S},l}$ will be called
\emph{names}. Elements from the set $C_l$ will describe core elements,
while those of $T_{\sche{S}}$ will describe the others. Furthermore,
neighboring nodes may describe overlapping pieces of the instance. In
particular, if one name is used in neighboring nodes, this means that
the name at hand refers to the same element---this is why we use $2w$
elements for the non-root bags. Let $\mbb{K}_{\sche{S},l}$ be the finite
schema capturing the following information:
\begin{itemize}
\item For all $a \in U_{\sche{S},l}$, there is a unary relation
  $D_a \in \mbb{K}_{\sche{S},l}$.
\item For all $a \in C_l$, there is a unary relation $C_a \in \mbb{K}_{\sche{S},l}$.
\item For each $R \in \sche{S}$ and every $n$-tuple
  $\ve{a} \in U^n_{\sche{S},l}$, there is a unary relation
  $R_{\ve{a}} \in \mbb{K}_{\sche{S},l}$.
\end{itemize}
Let $\Gamma_{\sche{S},l} \coloneqq 2^{\mbb{K}_{\sche{S},l}}$ be an
alphabet and suppose that $D$ is a (finite) $C$-tree over $\sche{S}$
such that $|\adom{C}| \leq l$. Consider a tree decomposition
$\delta = \tup{\ca{T},\tup{X_t}_{t \in T}}$ witnessing that $D$ is
indeed a $C$-tree and let $\varepsilon$ be the root of $\ca{T}$. Fix a
function $f \colon \adom{D} \ra U_{\sche{S},l}$ such that (i)
$f \upharpoonright \adom{C}$ is injective and (ii) different elements
that occur in neighboring bags of $\delta$ are always assigned different
names from $U_{\sche{S},l}$. Using $f$, we can encode $D$ and $\delta$
into a $\Gamma_{\sche{S},l}$-labeled tree $t = \tup{\hat{T},\mu}$ such
that each node from $\ca{T}$ corresponds to exactly one node in
$\hat{T}$ and vice versa. For a node $v$ from $\ca{T}$, we denote the
corresponding node of $T$ by $\hat{v}$ in the following and vice
versa. In this light, the symbols from $\mbb{K}_{\sche{S},l}$ have the
following intended meaning:
\begin{itemize}
\item $D_a \in \mu(\hat{v})$ means that $a$ is used as a name for some
  element of the bag $X_v$.
\item $C_a \in \mu(\hat{v})$ indicates that $a$ is used as name for an element
  of the bag $X_v$ that also occurs in $X_\varepsilon$, i.e., $a$ names
  an element from the core of $D$.
\item $R_{\ve{a}} \in \mu(\hat{v})$ indicates that $R$ holds in $D$ for the
  elements named by $\ve{a}$ in bag $X_v$.
\end{itemize}
Under certain assumptions, we can decode a $\Gamma_{\sche{S},l}$-labeled
tree $t = \tup{T,\mu}$ into a $C$-tree whose width is bounded by
$\arity{\sche{S}} - 1$. Let $\names{v} \coloneqq \set{a \mid D_a \in \mu(v)}$. We say
that $t$ is \emph{consistent}, if it satisfies the following properties:
\begin{enumerate}
\item For all nodes $v$ it holds that
  $|\names{v}| \leq \arity{\sche{S}}$, except for the root whose number
  of names are accordingly bounded by $l$. Furthermore,
  $\names{\varepsilon} \subseteq C_l$.
\item For all $R_{\ve{a}} \in \mbb{K}_{\sche{S},l}$ and all $v \in T$ it
  holds that $R_{\ve{a}} \in \mu(v)$ implies that
  $\set{\ve{a}} \subseteq \names{v}$.
\item For all $a \in C_l$ and all $v \in T$ it holds that
  $D_a \in \mu(v)$ iff $C_a \in \mu(v)$.
\item If $C_a \in \mu(v)$, then $C_a \in \mu(w)$ for all $w \in T$ on
  the unique shortest path between $v$ and the root.
% \item Neighboring nodes agree on their common elements, i.e., for all
%   $R/n \in \sche{S}$, all $\ve{a} \in U_{\sche{S},l}^n$, and all
%   neighboring nodes $v, w$, we have that $R_{\ve{a}} \in \mu(v)$ and
%   $\set{\ve{a}} \subseteq \names{w}$ implies $R_{\ve{a}} \in \mu(w)$.
\item For all nodes $v \neq \varepsilon$, there is an
  $R_{\ve{a}} \in \mbb{K}_{\sche{S},l}$ and a node $w$ such that
  $R_{\ve{a}} \in \mu(w)$, $\names{v} \subseteq \set{\ve{a}}$, and, for
  all $b \in \names{v}$, $v$ and $w$ are $b$-connected.
\end{enumerate}

\medskip
\noindent
\textbf{Decoding trees.} Suppose now that $t$ is consistent. We show how
we can decode $t$ into a database $\dec{t}$ which is a $C$-tree whose
diameter is bounded by $l$. Let $a$ be a name used in $t$. We say that
two nodes $v, w$ of $t$ are \emph{$a$-equivalent} if $D_a \in \mu(u)$
for all nodes $u$ on the unique shortest path between $v$ and
$w$. Clearly, $a$-equivalence defines an equivalence relation and we let
$[v]_a \coloneqq\set{\tup{w,a} \mid \text{$w$ is $a$-equivalent to
    $v$}}$ and $[v]_a^\ast \coloneqq \set{w \mid \tup{w,a} \in [v]_a}$.
The domain of $\dec{t}$ is the set
$\set{[v]_a \mid v \in T, a \in \mu(v)}$ and, for $R/n \in \sche{S}$, we
define
\begin{align*}
\dec{t} \models R([v_1]_{a_1},\ldots,[v_n]_{a_n})\ \defequ\ \text{there is some
$v \in [v_1]_{a_1}^\ast \cap \cdots \cap [v_n]_{a_n}^\ast$ such that
$ R_{a_1,\ldots,a_n} \in \mu(v)$}.
\end{align*}
\begin{lemma}
  \label{lem:consistency}
  Let $t$ be a consistent $\Gamma_{\sche{S},l}$-labeled tree with root
  node $\varepsilon$. Then $\dec{t}$ is well-defined and a $C$-tree over
  $\sche{S}$, where $C$ is the subinstance of $\dec{t}$ induced by the
  set $\set{[\varepsilon]_a \mid a \in \names{\varepsilon}}$. Moreover,
  $|\adom{C}|$ is bounded by $l$.
\end{lemma}
\begin{proof}
  Let $t = \tup{T, \mu}$ be a consistent, $\Gamma_{\sche{S},l}$-labeled
  tree. The fact that $\dec{t}$ is well-defined is left to the
  reader. We are going to construct an appropriate decomposition
  $\delta = \tup{\ca{T}, \tup{X_t}_{t \in T}}$ for $\dec{t}$. The tree
  $\ca{T}$ has the same structure as $t$. Furthermore, for $v \in T$, we
  set $X_v \coloneqq \set{[v]_a \mid a \in \names{v}}$. We need to show
  that $\delta$ is indeed a tree decomposition that satisfies the
  desired properties.

  Let $[v]_a \in \adom{\dec{t}}$ and consider two nodes $v_1,v_2 \in T$
  such that $[v]_a \in X_{v_1}$ and $[v]_a \in X_{v_2}$.  Then
  $v_1,v_2 \in [v]_a$ and so $v_1$ and $v_2$ are $a$-connected. Hence,
  $w \in [v]_a$ for all $w \in T$ which lie on the unique shortest path
  between $v_1$ and $v_2$. Since $a \in \names{w}$ for all such $w$, it
  follows that $[v]_a \in X_w$, and so $[v]_a$ is contained in all bags
  on the unique path between $v_1$ and $v_2$.
  Suppose $\dec{t} \models R([v_1]_{a_1},\ldots,[v_n]_{a_n})$. Then
  there is a $v \in [v_1]_{a_1}^\ast \cap \cdots \cap [v_n]_{a_n}^\ast$
  such that $R_{a_1,\ldots,a_n} \in \mu(v)$. By consistency,
  $\set{a_1,\ldots,a_n} \subseteq \names{v}$. Moreover, we know that
  $[v_i]_{a_i} = [v]_{a_i}$, for $i = 1,\ldots,n$. It follows that
  $\set{[v_1]_{a_1},\ldots,[v_n]_{a_n}} \subseteq X_v$.
  Now let $v \in T \setminus \set{\varepsilon}$. By consistency, there
  is an $R_{a_1,\ldots,a_n} \in \mbb{K}_{\sche{S},l}$ and a $w \in T$
  such that
  $\names{v} \coloneqq \set{a_{i_1},\ldots,a_{i_s}} \subseteq
  \set{a_1,\ldots,a_n} \subseteq \names{w}$,
  $R_{a_1,\ldots,a_n} \in \mu(w)$, and $v$ and $w$ are
  $b_{i_j}$-connected for $j = 1,\ldots,s$. By construction,
  $X_v = \set{[v]_{a_{i_1}},\ldots,[v]_{a_{i_s}}}$ and
  $\set{[w]_{a_1},\ldots,[w]_{a_n}} \subseteq X_w$. The claim follows
  now since $[v]_{a_{i_j}} = [w]_{a_{i_j}}$ for $j = 1,\ldots,s$. It is
  immediate that $|\adom{C}|$ is bounded by $l$.
\end{proof}

\textit{Notation.} Given a consistent $\Gamma_{\sche{S},l}$-labeled tree
$t = \tup{T,\mu}$ and a label $\rho \in \mu(T)$, in order to ease
notation we often regard $\rho$ as a database consisting of the facts
$\set{R(\ve{a}) \mid R_{\ve{a}} \in \rho}$. Furthermore, we let
$\names{\rho} \coloneqq \set{a \mid D_a \in \rho}$.

\medskip
\begin{proofcustom}{of Lemma~\ref{lem:consistent-labeled-trees}}
  The lemma is an easy consequence of Lemma~\ref{lem:consistency} and
  the fact that, when encoding a $C$-tree $D$ over $\sche{S}$, together
  with a tree decomposition witnessing that $D$ is a $C$-tree, into a consistent $\Gamma_{\sche{S},l}$-labeled tree $t$, then $\dec{t}$ and $D$ are isomorphic.
\end{proofcustom}

Roughly, Lemma~\ref{lem:consistent-labeled-trees} states that containment among OMQs from $\tup{\class{G},\class{BCQ}}$ can be semantically characterized via the decodings of consistent $\Gamma_{\sche{S},l}$-labeled trees. This makes the problem of deciding containment amenable to tree automata techniques.

\subsection*{Proof of Lemma~\ref{lem:automaton-1}}

Before proceeding to the proof of Lemma~\ref{lem:automaton-1}, we first introduce the relevant automata model.

\subsubsection*{Automata Techniques}

For a set of propositional variables $X$, we denote by $\mbb{B}^+(X)$
the set of Boolean formulas using variables from $X$, the connectives
$\land, \lor$, and the constants $\ptrue, \pfalse$. Let us now introduce our
automata model.

\begin{definition}
  \label{def:2WAPA}
  A \emph{two-way alternating parity automaton (2WAPA) on trees} is a
  tuple $\fk{A} = \tup{S, \Gamma, \delta, s_0, \Omega}$, where $S$ is a
  finite set of \emph{states}, $\Gamma$ an alphabet (the \emph{input
    alphabet} of $\fk{A}$),
  $\delta \colon S \times \Gamma \rightarrow
  \mbb{B}^+(\mathsf{tran}(\fk{A}))$ the \emph{transition function},
  where we set
  $\mathsf{tran}(\fk{A}) \coloneqq \set{\ndia{\alpha}s, \nbox{\alpha} s
    \mid s \in S, \alpha \in \set{-1,0,\ast}}$, $s_0 \in S$ the
  \emph{initial state}, and $\Omega \colon S \rightarrow \mbb{N}$ the
  \emph{parity condition} that assigns to each $s \in S$ a
  \emph{priority} $\Omega(s)$. Elements from $\mathsf{tran}(\fk{A})$ are
  called \emph{transitions}.\hfill\markfull
\end{definition}

Intuitively, a transition of the form $\ndia{0}s$ means that a copy of
the automaton should change to state $s$ and stay at the current node. A
transition of the form $\ndia{-1} s$ means that a copy should be sent to
the parent node, which is then required to exist, and proceed in state
$s$, while one of the form $\ndia{\ast} s$ means that a copy of the
automaton that assumes state $s$ is sent to some child node. The
transition $\nbox{0}s$ means the same as $\ndia{0}s$, while $\nbox{-1}s$
means that a copy of the automaton that assumes state $s$ should be sent
to the parent node which is there not required to exist at
all. Likewise, $\nbox{\ast}s$ means that a copy of the automaton
assuming state $s$ should be sent to all child nodes.

\medskip \textit{Notation.} We write $\Diamond s$ for
$\bigvee \set{\ndia{\alpha}s \mid \ndia{\alpha}s \in
  \mathsf{tran}(\fk{A}), s \in S}$, $\Box s$ for
$\bigwedge \set{\nbox{\alpha}s \mid \nbox{\alpha}s \in
  \mathsf{tran}(\fk{A}), s \in S}$, and simply $s$ for $\ndia{0}s$.
Furthermore, for $\alpha \in \mbb{N} \cup \set{-1, \ast}$, we define
\begin{align*}
  T_\alpha(x) \coloneqq \begin{cases}
    \set{x \cdot \alpha}, & \text{if $\alpha = -1$ and $x \cdot \alpha \in T$,}\\
    \set{x \cdot i \mid x \cdot i \in T, i \in \mbb{N} \setminus \set{0}}, & \text{if $\alpha = \ast$,}\\
    \emptyset, & \text{otherwise.}
    \end{cases}
\end{align*}
\begin{definition}
  A \emph{run} of a 2WAPA
  $\fk{A} = \tup{S,\Gamma, \delta, s_0, \Omega}$ on a $\Gamma$-labeled
  tree $\tup{T,\eta}$ is a $T \times S$-labeled tree
  $\tup{T_r, \eta_r}$ such that the following holds:
\begin{enumerate}
  \item $\eta_r(\varepsilon) = \tup{\varepsilon, s_0}$,
  \item if $y \in T_r$, $\eta_r(y) = \tup{x,s}$, and
    $\delta(s,\eta(x)) = \varphi$, then there is an
    $I \subseteq \mathsf{tran}(\mbb{A})$ such that $I \models \varphi$
    holds and the following conditions are satisfied:
    \begin{itemize}
    \item If $\ndia{\alpha} s' \in I$ then there is a node
      $x' \in T_\alpha(x)$ and a child node $y' \in T_r$ of $y$ such
      that $\eta_r(y') = \tup{x', s'}$.
    \item If $\nbox{\alpha} s' \in I$ then for all $x' \in T_\alpha(x)$,
      there is a child node $y' \in T_r$ of $y$ such that
      $\eta_r(y') = \tup{x', s'}$.
    \end{itemize}
\end{enumerate}
We say that a run $\tup{T_r,\eta_r}$ is \emph{accepting} on $\fk{A}$, if
on all infinite paths
$\tup{\varepsilon,s_0},\tup{x_1,s_1},\tup{x_2,s_2},\ldots$ in $T_r$, the
maximum priority among $\Omega(s_0),\Omega(s_1),\Omega(s_2),\ldots$ that
appears infinitely often is even. $\fk{A}$ \emph{accepts} a
$\Gamma$-labeled tree $\tup{T,\eta}$, if there is an accepting run on
$\tup{T,\eta}$. We denote by $\ca{L}(\fk{A})$ the set of
$\Gamma$-labeled trees $\fk{A}$ accepts, i.e., the \emph{language
  accepted by $\fk{A}$}.\hfill\markfull
\end{definition}

\textit{Remark.} The automaton model defined above resembles that
in~\cite{Wi01}. However, we explicitly provide transitions that allow
the automaton move to the parent node, while the model defined
in~\cite{Wi01} provides transitions for moving to \emph{some}
neighboring node, including the parent node. Therefore, the automata
in~\cite{Wi01} offer transitions of the form $s$, $\Diamond s$, and
$\Box s$ with their intended meaning as defined above. Using techniques
as employed in~\cite{Vardi98,Wi01}, for a 2WAPA $\fk{A}$, one can show
that the problem of deciding whether $\ca{L}(\fk{A}) = \emptyset$ is
feasible in exponential time with respect to the number of states of
$\fk{A}$ and in polynomial time with respect to the size of the input
alphabet of $\fk{A}$.

\medskip
\begin{proofcustom}{of Lemma~\ref{lem:automaton-1}}
We only give an intuitive explanation for the construction of the desired 2WAPA. To check whether a $\Gamma_{\sche{S},l}$-labeled tree is consistent, we can check each condition for consistency separately by a dedicated 2WAPA and then take the intersection of all of them.
Most of the consistency conditions are easy to check. We give here a more
detailed verbal explanation for condition (5). A 2WAPA checking this condition can be constructed as follows. At the beginning of its run, the
automaton branches universally to all nodes (except the root) in a state
whose intended purpose is to find appropriate guards in the input tree
for the names available at the current node. To this end, the automaton
has to do a reachability analysis on the input tree and store, using
exponentially many states in $\arity{\sche{S}}$, the tuple it seeks to
guard. By a \emph{guard} for the node $v$ here, we mean a node $w$ with
an $R_{\ve{a}} \in \mu(w)$ such that
\begin{enumerate*}[label={(\roman*)}]
\item $\set{\ve{a}}$ contains all the
names present at $w$ and
\item is $b$-connected to $v$ for all $b \in \names{v}$.
\end{enumerate*}
Notice that such a reachability analysis can be easily performed once we
have the means to store the information contained in $\names{v}$ in a
single state. This is, however, possible since for this task we need
somewhat $O((\arity{\sche{S}} + l)^{\arity{\sche{S}}})$ states, i.e.,
  polynomially many in the size of $\Gamma_{\sche{S},l}$.
\end{proofcustom}

\subsection*{Proof of Lemma~\ref{lem:automaton-2}}

We first need to introduce some additional auxiliary notions.

\subsubsection*{Strictly Acyclic Queries}

Let $q$ be a CQ over a schema $\sche{S}$. We denote by $\free{q}$ the
free variables of $q$; the same notation is used for first-order
formulas in general. We can naturally view $q$ as an instance $[q]$
whose domain is the set of variables of $q$ and contains the body atoms
of $q$ as facts. In the following, we will often overload notation and
write $q$ for both the query $q$ and the instance $\cqasinst{q}$.  The
notions of tree-width, acyclicity, etc.~then immediately extend to
CQs. Given a tree decomposition $\delta$ of $q$ (i.e., of
$\cqasinst{q}$), we say that $\delta$ is \emph{strict}, if some bag of
$\delta$ contains all variables that are free in
$q$~(cf.~also~\cite{FlFG02}). Accordingly, $q$ is called \emph{strictly
  acyclic} if it has a guarded tree decomposition that is strict.

Strictly acyclic queries have the convenient property to be equivalent
to \emph{guarded formulas} of a special form. Recall that the set
of \emph{guarded formulas} over a schema $\sche{S}$ is built inductively
by including all atomic formulas, relativizing quantifiers by atomic
formulas, and closing under Boolean connectives. More precisely, all
quantifier occurrences have one of the forms
\begin{align*}
  \forall \ve{y}\,(\alpha(\ve{x},\ve{y}) \limpl \varphi)\quad \text{and}\quad \exists \ve{y}\,(\alpha(\ve{x},\ve{y}) \land \varphi),
\end{align*}
such that the free variables of $\varphi$ are among
$\set{\ve{x},\ve{y}}$.

We are interested in the guarded formulas that are build up using
conjunction and existential quantification; we restrict ourselves to
such formulas in the following. We call a formula from this class
\emph{strictly guarded}, if it is of the form
$\exists\ve{y}\,(\alpha(\ve{x},\ve{y}) \land \varphi)$. We explicitly
include the case where $\ve{y}$ is the empty sequence of variables,
i.e., if to formulas of the form $\alpha(\ve{x}) \land \varphi$ with
$\free{\varphi} \subseteq \set{\ve{x}}$. Notice that every guarded
sentence $\varphi$ (i.e., a formula having no free variables) is
strictly guarded, since it is equivalent to
$\exists y\,(y = y \land \varphi)$. Furthermore, notice that every usual
guarded formula that uses only existential quantifiers and conjunction
is equivalent to a conjunction of strictly guarded formulas. The
following lemma is proved in~\cite{FlFG02}.
\begin{lemma}
  Every strictly acyclic CQ can be rewritten in polynomial time into
  an equivalent strictly guarded formula that is built up using
  conjunction and existential quantification only. The converse holds as
  well.
\end{lemma}

\subsubsection*{Squid Decompositions}

Let $q$ be a BCQ over a schema $\sche{S}$ having $n$ body atoms. An
\emph{$\sche{S}$-cover} of $q$ is a BCQ $q^+$ that contains all the
atoms from $q$ and may additionally contain $2n$ other body atoms over
$\sche{S}$.  It is pretty straightforward that, for an
$\sche{S}$-instance $I$, it holds that $I \models q$ iff there is an
$\sche{S}$-cover $q^+$ of $q$ such that $I \models q^+$.
\begin{definition}
  Let $I$ be an instance. For $V \subseteq \adom{I}$, we say that $I$ is
  \emph{$[V]$-acyclic}, if it has a guarded tree decomposition that
  omits $V$.  \hfill\markfull
\end{definition}

\begin{definition}
  Let $q$ be a BCQ over $\sche{S}$. A \emph{squid decomposition} of $q$
  is a tuple $\delta = \tup{q^+,\mu, H, T, V}$, where $q^+$ is an
  $\sche{S}$-cover of $q$,
  $\mu \colon \var{{q^+}} \rightarrow \var{{q^+}}$ a mapping,
  $V \subseteq \var{\mu(q^+)}$, and $\tup{H,T}$ a partition of the atoms
  $\mu(q^+)$ such that
  \begin{itemize}
    \item $H$ is the set of atoms of $\mu(q^+)$ induced by $V$,
    \item $T = \mu(q^+) \setminus H$ and $T$ is $[V]$-acyclic.\hfill\markfull
  \end{itemize}
\end{definition}
Intuitively, a squid decomposition specifies a way how a BCQ can be
mapped to an instance that contains some ``cyclic parts''---the set $H$
specifies those atoms that are mapped to such cyclic parts, while $A$
declares those atoms that are mapped to the acyclic parts of the
instance at hand. We will make this more precise in
Lemma~\ref{lem:squid} below, where we analyze matches in $C$-trees.

Given a CQ $q$ and a set of variables $V \subseteq \var{q}$, the
\emph{$V$-reduct} of $q$, denoted $q^V$, is the conjunctive query that
arises from $q$ by dropping all the existential quantifiers that bind
variables in $V$.

\begin{lemma}
\label{lem:squid}
Let $J$ be a $C$-tree over $\sche{S}$ and $q$ a BCQ over
$\sche{S}$. Let $\tup{\ca{T},\tup{X_t}_{t\in T}}$ be a witnessing tree
decomposition of $J$. It holds that $J \models q$ iff there is a squid
decomposition $\delta = \tup{q^+,\mu,H,A,V \coloneqq \set{\ve{x}}}$ of
$q$ and a homomorphism $\eta \colon \mu(q^+) \rightarrow J$ such that
  \begin{enumerate}
  \item $C \models H$ is witnessed by $\eta$,
    \label{lem:squid:1}
    \item
    $\bigcup_{\varepsilon \prec v} J(v) \models A^V(\eta(\ve{x}))$
    is witnessed by $\eta$, and
    \label{lem:squid:2}
  \item there are strictly guarded formulas $\varphi_1,\ldots,\varphi_l$
    such that
    $A^V(\ve{x}) \equiv \varphi_1 \land \cdots \land \varphi_l$.
    \label{lem:squid:3}
  \end{enumerate}
\end{lemma}
\begin{proof}
  For the direction from right to left, consider such a given squid
  decomposition $\delta$ and a homomorphism $\eta$ as in the hypothesis
  of the lemma. It is immediate that $\eta \circ \mu$ is a homomorphism
  mapping $q^+$ to $J$. Since $q^+$ is an $\sche{S}$-cover of $q$, we
  obtain $J \models q$ as required.

  For the other direction, suppose that $J \models q$ is witnessed by a
  homomorphism $\theta$. For each $v \in T \setminus \set{\varepsilon}$,
  let $\beta_v$ be an atom of $J$ such that $J(v) \models \beta_v$ and
  $\beta_v$ contains all domain elements from $J(v)$ as
  arguments. Notice that the $\beta_v$ exist, since $\delta$ is guarded
  except for $\set{\varepsilon}$. Since $\theta$ maps $q$ to $J$, for
  each atom $\alpha$ of $q$, there is a node $v_\alpha$ such that
  $\theta(\alpha) \in J(v_\alpha)$. Let $W$ be the set of all these
  nodes and their closure under greatest lower bounds with respect to
  $\preceq$, excluding the root node $\varepsilon$ of $\ca{T}$. Consider
  the set of atoms
  $Q^+ \coloneqq \theta(q) \cup \set{\beta_v \mid v \in W}$. Notice that
  at least half of the nodes of $W$ are of the form $v_\alpha$---hence,
  $|Q^+| \leq 3|q|$. Let $q^+$ be a BCQ constructed as follows. Take the
  conjunction of $q$ and for each $\beta_v(a_1,\ldots,a_n)$ ($v \in W$),
  add an atom $\beta_v(x_1,\ldots,x_n)$, where each $x_i$ is a newly
  chosen variable. Then $q^+$ is obviously an $\sche{S}$-cover of
  $q$. Furthermore, by construction, there is a mapping
  $\mu \colon \var{q^+} \rightarrow \var{q^+}$ and an isomorphism
  $\eta \colon \mu(q^+) \rightarrow Q^+$ such that
  $(\eta \circ \mu)(q^+) = Q^+$. Now let $H$ be the greatest set of
  atoms of $\mu(q^+)$ such that $\eta(H) \subseteq
  J(\varepsilon)$. Moreover, let $V \coloneqq \var{H}$ and
  $A \coloneqq \mu(q^+) \setminus H$. We claim that
  $\delta \coloneqq \tup{q^+, \mu, H, A, V}$ is a squid decomposition of
  $q$ that satisfies together with $\eta$ the points mentioned in the
  statement of the lemma.

  To see that $\delta$ is a squid decomposition of $q$, the only
  nontrivial point to prove is that $A$ is indeed $[V]$-acyclic. We will
  prove this below in the course of establishing the third item.

  The first two items are immediate by construction. We prove the third
  item. Suppose $V = \set{\ve{x}}$ and consider the $V$-reduct $A^V(\ve{x})$
  of $A$. By construction, the atoms $\eta(A)$ are contained in
  $\bigcup_{\varepsilon \prec v} J(v)$. Now the set $W$ together with
  the order $\preceq_{\ca{T}}$ gives rise to a forest consisting of
  trees $\ca{T}_1,\ldots,\ca{T}_l$ whose roots are descendants of
  $\varepsilon$, i.e., the root of $\ca{T}$ (recall that $\varepsilon$
  is not contained in $W$). Moreover, we have
  \begin{enumerate*}[label={(\roman*)}]
    \item $\bigcup_{i = 1}^l T_i = W$,
    \item $T_i \cap T_j = \emptyset$, for $i \neq j$, and
    \item
      $\bigcup_{v \in T_i} X_v \cap \bigcup_{v \in T_j} X_v \subseteq
      \adom{C}$, for $i \neq j$.
  \end{enumerate*}

  For $v \in T$, let
  $Q^+(v) \coloneqq \set{\alpha \in Q^+ \mid J(v) \models \alpha}$ and,
  for $i = 1,\ldots,l$, let $Q^+(\ca{T}_i)$ be the set of atoms
  $\bigcup_{v \in T_i}Q^+(v)$. Now it is easy to check using the facts
  stated before that each $Q^+(\ca{T}_i)$ is acyclic and, hence, so is
  $\eta^{-1}(Q^+(\ca{T}_i))$.  Furthermore, denoting by $\varepsilon_i$
  the root of $\ca{T}_i$, it holds that
  $\adom{Q^+(\ca{T}_i)} \cap \eta(V) \subseteq
  \adom{Q^+(\varepsilon_i)}$---indeed, if
  $a \in \adom{Q^+(v)} \cap \eta(V)$ for some $v \succeq \varepsilon_i$,
  then, since $\varepsilon_i \succ \varepsilon$ and
  $a \in X_\varepsilon$, it must be the case that
  $a \in \adom{Q^+(\varepsilon_i)}$ by connectivity. It follows that the
  $V$-reduct of $\eta^{-1}(Q^+(\ca{T}_i))$ (viewed as Boolean query),
  henceforth denoted $q^+_{\ca{T}_i}$, is strictly acyclic and is
  therefore equivalent to a strictly guarded formula $\varphi_i$. Hence,
  the query $A^V(\ve{x})$ is equivalent to
  $\bigwedge_{i=1}^l \varphi_i$. Moreover, it follows that $A$ itself is
  $[V]$-acyclic---notice that
  $A \equiv \exists\ve{x}\bigwedge_{i=1}^l q^+_{\ca{T}_i}$ and that
  $\adom{Q^+(\ca{T}_i)} \cap \adom{Q^+(\ca{T}_j)} \subseteq \eta(V)$,
  for $i \neq j$. Hence,
  $\var{q^+_{\ca{T}_i}} \cap \var{q^+_{\ca{T}_j}} \subseteq V$, for
  $i \neq j$. The claim now follows since every $q^+_{\ca{T}_i}$ is
  acyclic.
\end{proof}

\subsubsection*{Derivation trees}

Let $D$ be an $\sche{S}$-database and $\Sigma$ a set of guarded
rules. Let $q_0(\ve{x})$ be a strictly acyclic query whose free
variables are exactly those from $\ve{x} \coloneqq x_1,\ldots,x_n$ and
let $\ve{a} \coloneqq a_1,\ldots,a_n$ be a tuple from $\adom{D}$. A
\emph{derivation tree for $\tup{\ve{a},q_0(\ve{x})}$ with respect to $D$
  and $\Sigma$} is a finite tree $\ca{T}$ whose nodes are labeled via a
function $\mu$ with pairs of the form
$\tup{b_1,\ldots,b_k;q(y_1,\ldots,y_k)}$, where $b_1,\ldots,b_k$ are
constants from $\adom{D}$ and $q(y_1,\ldots,y_k)$ is a strictly acyclic
query over $\sche{S} \cup \sch{\Sigma}$ having exactly $y_1,\ldots,y_k$
free, such that the following conditions are satisfied:
\begin{enumerate}
\item $\mu(\varepsilon) = \tup{\ve{a},q_0(\ve{x})}$, where $\varepsilon$
  is the root node of $\ca{T}$.
\item If $\mu(v) = \tup{c_1,\ldots,c_m ;q(z_1,\ldots,z_m)}$ for some
  node $v$, then one of the following conditions holds (let
  $\ve{c} \coloneqq c_1,\ldots,c_m$ and
  $\ve{z} \coloneqq z_1,\ldots,z_m$):
\begin{enumerate}
\item $v$ is a leaf node and $q(\ve{z}) \equiv \beta(\ve{z})$, for some
  atomic formula $\beta(\ve{z})$ such that $D \models \beta(\ve{c})$.
\item The node $v$ has a successor labeled by
  $\tup{\ve{c},\ve{b}; p(\ve{z},\ve{y})}$ and it holds that
\begin{align*}
  \Sigma \models \forall \ve{z},\ve{y}\,(p(\ve{z},\ve{y}) \limpl q(\ve{z})).
\end{align*}
\item The query $q(\ve{z})$ is logically equivalent to
  $q_1(z_{i_{1,1}},\ldots,z_{i_{1,k_1}}) \land \cdots \land
  q_l(z_{i_{l,1}},\ldots,z_{i_{l,k_l}})$ and $v$ has $l$ successors
  $v_1,\ldots,v_l$ respectively labeled by
  $\tup{c_{i_{1,1}},\ldots,c_{i_{1,k_1}};
    q_1(z_{i_{1,1}},\ldots,z_{i_{1,k_i}})},\ldots,\tup{c_{i_{l,1}},\ldots,c_{i_{1,k_l}};
    q_l(z_{i_{l,1}},\ldots,z_{i_{1,k_l}})}$.
\end{enumerate}
\end{enumerate}

\begin{lemma}
  \label{lem:derivtreeatom}
  Let $\alpha(x_1,\ldots,x_n)$ be an atomic formula. Then
  $D,\Sigma \models \alpha(a_1,\ldots,a_n)$ iff there is a derivation
  tree for $\tup{a_1,\ldots,a_n; \alpha(x_1,\ldots,x_n)}$ with respect
  to $D$ and $\Sigma$.
\end{lemma}

\begin{proofsk}
  Let $\ve{a} \coloneqq a_1,\ldots,a_n$ and
  $\ve{x} \coloneqq x_1,\ldots,x_n$. The direction from right to left is
  an easy induction on the construction of the derivation tree. We
  sketch the other direction. Consider the guarded chase forest
  $\tup{\ca{F},\eta}$ for $D$ and $\Sigma$, where $\eta$ is a function
  labeling the nodes and edges of $\ca{F}$. We construct a derivation
  tree for $\tup{\ve{a},\alpha(\ve{x})}$ by induction on the number of
  chase steps required to derive $\alpha(\ve{a})$ from $D$ and $\Sigma$.

  For the base case, if $D \models \alpha(\ve{a})$, the claim is obvious
  since we can apply rule 2.(a). Assume that $\alpha(\ve{a})$ is derived
  using a rule
  \begin{align*}
    \sigma \colon \quad \beta_0(\ve{x},\ve{y}),\beta_1,\ldots,\beta_k \limpl \alpha(\ve{x}),
  \end{align*}
  and a homomorphism $\mu$ such that $\mu(\ve{x}) = \ve{a}$, where
  $\beta_0(\ve{x},\ve{y})$ is the guard of $\sigma$. If
  $\mu(\set{\ve{x},\ve{y}}) \subseteq \adom{D}$, the result immediately
  follows by the induction hypothesis. Otherwise, the image of
  $\beta_0(\ve{x},\ve{y})$ under $\mu$ contains some labeled nulls as
  arguments. Assume that all the $\beta_1,\ldots,\beta_k$ contain nulls
  as their arguments---for those that do not, the induction hypothesis
  would yield appropriate derivation trees again. Notice that all the
  nulls occurring in $\beta_1,\ldots,\beta_k$ appear in
  $\mu(\set{\ve{y}})$. By construction of $\ca{F}$, there is a node
  $v_0$ that is an ancestor of the nodes having the atoms
  $\mu(\beta_0),\mu(\beta_1),\ldots,\mu(\beta_k)$ as labels and which
  has a label of the form $\beta_0(\ve{a},\ve{b})$ which contains no nulls
  at all as arguments. There is a corresponding atomic formula
  $\gamma_0(\ve{x},\ve{z})$ whose image under an appropriate
  homomorphism equals $\beta_0(\ve{a},\ve{b})$. Furthermore, there are atoms
  $\gamma_1,\ldots,\gamma_l$ such that
  $\adom{\set{\gamma_1,\ldots,\gamma_l}} \subseteq \set{\ve{a},\ve{b}}$
  and
  \begin{align*}
    \Sigma \models \beta_0(\ve{a},\ve{b}) \land \gamma_1 \land \cdots \land \gamma_l \limpl \exists\ve{y}\,(\beta_0(\ve{a},\ve{y}) \land \beta_1 \land \cdots \land \beta_k).
  \end{align*}
  Now regard the $\gamma_i$ ($i = 1,\ldots,l$) as atomic formulas with
  free variables among $\set{\ve{x},\ve{z}}$. The formula
  $p(\ve{x},\ve{z}) \coloneqq \gamma_0(\ve{x},\ve{z}) \land \gamma_1
  \land \cdots \land \gamma_l$ is then a strictly acyclic query that
  satisfies
  $\Sigma \models \forall\ve{x},\ve{z}\,(p(\ve{x},\ve{z}) \limpl
  \alpha(\ve{x}))$. An application of rule 2.(b) then requires us to
  find a derivation tree for $\tup{\ve{a},\ve{b};p(\ve{x},\ve{y})}$,
  whence an application of rule 2.(c) reduces this task to finding
  derivation trees for the atoms $\gamma_0,\gamma_1,\ldots,\gamma_l$ and
  their corresponding tuples of constants. These trees exist by
  induction hypothesis and we can simply concatenate them appropriately
  in order to arrive at a derivation tree for
  $\tup{\ve{a},\alpha(\ve{x})}$.
\end{proofsk}

Given a guarded formula $\varphi(\ve{x})$ built up from conjunctions and
existential quantification, we define the \emph{nesting depth} of
$\varphi(\ve{x})$, denoted $\nd{\varphi(\ve{x})}$, inductively:
\begin{itemize}
\item If $\varphi(\ve{x})$ is an atomic formula, then
  $\nd{\varphi(\ve{x})} \coloneqq 0$.
\item If $\varphi(\ve{x}) = (\psi_1 \land \psi_2)$, then
  $\nd{\varphi(\ve{x})} \coloneqq \max\{\nd{\psi_1},\nd{\psi_2}\}$.
\item If
  $\varphi(\ve{x}) = \exists\ve{y}\,(\alpha(\ve{x},\ve{y}) \land \psi)$
  and $\ve{y} \neq \emptyset$, then
  $\nd{\varphi(\ve{x})} \coloneqq \nd{\psi} + 1$.
\end{itemize}

\begin{lemma}
  \label{prop:derivtreequery}
  Let $D$ be a database, $\Sigma$ a set of guarded rules, and
  $q(\ve{x})$ a strictly acyclic conjunctive query. Then
  $D,\Sigma \models q(\ve{a})$ iff there is a derivation tree for
  $\tup{\ve{a},q(\ve{x})}$ with respect to $D$ and $\Sigma$.
\end{lemma}

\begin{proofsk}
  We again sketch only the direction from left to right. Let
  $\varphi(\ve{x})$ be the strictly guarded formula corresponding to
  $q(\ve{x})$. We proceed by induction on the nesting depth of
  $\varphi(\ve{x})$. If $\nd{\varphi(\ve{x})} = 0$, then
  $\varphi(\ve{x})$ is quantifier free and thus a conjunction of atoms
  $\alpha_0(\ve{x}) \land \alpha_1 \land \cdots \land \alpha_k$, where
  $\var{\alpha_i} \subseteq \set{\ve{x}}$ for $i = 1,\ldots,k$. An
  application of rule 2.(c) reduces the problem of building a derivation
  tree for $\tup{\ve{a},\varphi(\ve{x})}$ to the problem of building
  corresponding trees for the $\alpha_i$ and their corresponding
  constants from $\ve{a}$. The existence of these trees is guaranteed
  by Lemma~\ref{lem:derivtreeatom}.

  Now suppose that $\nd{\varphi(\ve{x})} = n + 1$. Let
  $\varphi(\ve{x}) = \exists \ve{y}\,(\alpha(\ve{x},\ve{y}) \land \psi)$
  and $\ve{y} \coloneqq y_1,\ldots,y_k$. Assume, without loss of
  generality, that all the bound variables from $\varphi(\ve{x})$ are
  pairwise distinct. In the following, we will describe how to construct
  a derivation tree for $\tup{\ve{a},q(\ve{x})}$. If
  $D,\Sigma \models q(\ve{a})$, then there is a homomorphism $\mu$
  mapping each atom of $q(\ve{x})$ to $\chase{D}{\Sigma}$ such that
  $\mu(\ve{x}) = \ve{a}$. Furthermore, $\mu$ maps each atom of
  $q(\ve{x})$ to a node of the guarded chase forest $\ca{F}$ of $D$ and
  $\Sigma$. Let $\alpha_\mu(\ve{a},\lambda_1,\ldots,\lambda_k)$ denote
  the atom labeling the node of $\ca{F}$ where $\alpha(\ve{x},\ve{y})$
  is mapped to via $\mu$. Let $\lambda_{i_1},\ldots,\lambda_{i_l}$
  exhaust all elements from $\lambda_1,\ldots,\lambda_k$ that are not
  from $\adom{D}$ and
  $\ve{b} \coloneqq \lambda_{j_1},\ldots,\lambda_{j_m}$ exhaust those
  from $\lambda_1,\ldots,\lambda_k$ that are from $\adom{D}$. Let
  $\varphi'(\ve{x},y_{j_1},\ldots,y_{j_m})$ be the formula
  $\exists y_{i_1},\ldots,y_{i_l}\,(\alpha(\ve{x},\ve{y}) \land
  \psi)$. Clearly,
  $\Sigma \models \forall
  \ve{x},y_{j_1},\ldots,y_{j_m}\,(\varphi'(\ve{x},y_{j_1},\ldots,y_{j_m})
  \limpl \varphi(\ve{x}))$. Hence, we can create a successor of
  $\tup{\ve{a},\varphi(\ve{x})}$ that is labeled by
  $\tup{\ve{a},\ve{b};\varphi'(\ve{x},y_{j_1},\ldots,y_{j_m})}$. Assume
  now that none of the $\lambda_1,\ldots,\lambda_k$ is from
  $\adom{D}$. Furthermore, assume that $k \geq 1$, since otherwise we
  can just simply apply rule 2.(c) to reduce $q(\ve{x})$ to a
  conjunction of queries of the desired form. As in the proof of
  Lemma~\ref{lem:derivtreeatom}, there is a node $v_0$ in $\ca{F}$ whose
  label $\beta_0(\ve{a},\ve{b})$ contains only values from $\adom{D}$ as
  arguments and such that $v_0$ is an ancestor of the node labeling
  $\alpha_\mu(\ve{a},\lambda_1,\ldots,\lambda_k)$. Furthermore, all the
  atoms from $\mu(q(\ve{x}))$ that contain an element from
  $\lambda_1,\ldots,\lambda_k$ as argument are also located in the
  subtree rooted at $v_0$. Let $p$ be the query that results from
  deleting all atoms from $q(\ve{x})$ which are mapped via $\mu$ into
  the subtree rooted at $v_0$. Notice that $p$ may be empty and has free
  variables among $\ve{x}$. Furthermore $p$ is equivalent to a
  conjunction $p_1 \land \cdots \land p_l$ of strictly acyclic
  queries. Let $\beta_0(\ve{x},\ve{z})$ be the atomic formula whose
  image under an appropriate homomorphism equals
  $\beta_0(\ve{a},\ve{b})$. A similar line of reasoning as in the proof
  of Lemma~\ref{lem:derivtreeatom} shows that there are atomic formulas
  $\beta_1,\ldots,\beta_m$ such that
  $\var{\beta_i} \subseteq \set{\ve{x},\ve{z}}$ and
  \begin{align*}
    \forall \ve{x},\ve{z}\,(\beta_0(\ve{x},\ve{z}) \land \beta_1 \land \cdots \land \beta_m \land p \limpl \varphi(\ve{x})),
  \end{align*}
  whence an application of rule 2.(b) and rule 2.(c) reduces the problem
  of constructing a derivation tree for $\tup{\ve{a},\varphi(\ve{x})}$
  to that of constructing corresponding trees for
  $\beta_0,\ldots,\beta_m$ and $p$. Notice that $p$ is a conjunction of
  strictly guarded formulas of nesting depth at most $n$. Hence, the
  induction hypothesis guarantees the existence of such derivation
  trees.
\end{proofsk}

Having the above results in place, it is easy to show the following
statement:
\begin{lemma}
  \label{lem:automatonguide}
  Let $D$ be a $C$-tree over $\sche{S}$ and $Q = \tup{\sche{S},\Sigma,q}$
  an OMQ where $\Sigma$ is guarded and $q$ a BCQ. Then
  $D \models Q$ iff there is a squid decomposition
  $\delta = \tup{q^+,\mu,H,A,V \coloneqq \set{\ve{x}}}$ of $q$ and a
  homomorphism $\eta \colon \mu(q^+) \rightarrow \chase{D}{\Sigma}$
  such that:
  \begin{enumerate}
  \item $F \models H$ is witnessed by $\eta$, where $F$ is
    the subinstance of $\chase{{D}}{\Sigma}$ induced by $\adom{C}$.
  \item There are strictly acyclic queries $q_1,\ldots,q_l$ such that
    \begin{enumerate}
    \item $A^V(\ve{x}) \equiv q_1 \land \cdots \land q_l$ and
    \item for $i =1,\ldots,l$ and $\free{q_i} = \set{\ve{x}_i}$, there
      are derivation trees for $\tup{\eta(\ve{x}_i), q_i}$ with respect
      to $D$ and $\Sigma$.
  \end{enumerate}
  \end{enumerate}
\end{lemma}
\begin{proof}
  We can easily prove by induction on the number of chase steps that
  $\chase{D}{\Sigma}$ is an $F$-tree, where $F$ is the subinstance of
  $\chase{D}{\Sigma}$ induced by $\adom{C}$. Now the lemma at hand is
  immediate by combining this fact with Lemma~\ref{lem:squid} and
  Lemma~\ref{prop:derivtreequery}.
\end{proof}

We are now ready to proceed with the proof of
Lemma~\ref{lem:automaton-2}:

\medskip
\begin{proofcustom}{of Lemma~\ref{lem:automaton-2}}
  Lemma~\ref{lem:automatonguide} will guide the construction of the
  2WAPA we are now going to construct.  Suppose
  $Q = \tup{\sche{S},\Sigma,q}$ is an OMQ from
  $\tup{\class{G},\class{BCQ}}$ and let $l \geq 1$. We are going to
  construct a 2WAPA
  $\fk{A}_{Q,l} = \tup{S,\Gamma_{\sche{S},l},\delta,s_0,\Omega}$ that
  accepts a consistent $\Gamma_{\sche{S},l}$-labeled tree $t$ iff
  $\dec{t} \models Q$. In particular, the number of states of
  $\fk{A}_{Q,l}$ will be at most exponential in the size of $Q$ and at
  most polynomial in $l$, while the construction of $\fk{A}_{Q,l}$ will
  be feasible in \twoexptime.

  \medskip
  \noindent
  \textit{The state set.} Let $\Lambda$ denote the set of all Boolean
  acyclic queries over $\sche{S} \cup \sch{\Sigma}$ that are of size at
  most $3|q|$. Notice that each of these queries is equivalent to a
  strictly guarded formula. Furthermore, assume that $\Lambda$ is closed
  under $V$-reducts, for $V \subseteq \var{q}$, provided that they are
  strictly acyclic as well, i.e., if $p \in \Lambda$ and
  $V \subseteq \var{q}$, then also $p^V \in \Lambda$ provided $p^V$ is
  strictly acyclic. For $\set{\ve{a}} \subseteq U_{\sche{S},l}$, let
  \begin{align*}
    \hat{S}(\ve{a}) \coloneqq \set{p(\ve{x}/\ve{a}) \mid p \in \Lambda, \free{p} = \set{\ve{x}}, |\ve{a}| = |\ve{x}|}
  \end{align*}
  and let $\hat{S}$ be the union of all the sets $\hat{S}(\ve{a})$. Now
  the set of states $S$ consists of an initial state, denoted $s_0$,
  plus the set $\hat{S}$ factorized modulo logical equivalence. We
  denote by $[p]$ the equivalence class of a query $p \in
  \hat{S}$. Furthermore, for a strictly guarded formula $\varphi$, we
  may abuse notation and write $[\varphi]$ for the equivalence class of
  the strictly acyclic query $p \in \hat{S}$ that is equivalent to
  $\varphi$. Notice that the size of $S$ is exponential in the size of
  $Q$, since there are only exponentially many CQs of size at most
  $3|q|$ that are mutually non-equivalent (cf.~\cite{BaGO14}).

  \medskip
  \noindent
  \textit{The parity condition.} We set $\Omega(s) \coloneqq 1$, for all
  $s \in S$. This means that only finite trees are accepted.

  \medskip
  \noindent
  \textit{The transition function.} In the following, for each
  $\rho \in \Gamma_{\sche{S},l}$, we denote by $\hat{\Theta}(\rho)$ the
  set of all pairs that are of the form
  $\tup{\alpha_1 \land \cdots \land \alpha_n, p_1 \land \cdots \land
    p_m}$ for which there is a squid decomposition of the form
  $\tup{q^+,\mu,H,T,\set{\ve{x}}}$ and a function
  $\theta \colon \set{\ve{x}} \ra \names{\rho}$ such that:
\begin{itemize}
\item
  $H^{\set{\ve{x}}}(\theta(\ve{x})) \equiv \alpha_1 \land \cdots \land
  \alpha_n$, where all the $\alpha_i$ are relational ground atoms.
\item
  $T^{\set{\ve{x}}}(\theta(\ve{x})) \equiv p_1 \land \cdots \land p_m$,
  where the $p_i$ are strictly acyclic queries.
\end{itemize}
Call two pairs $\tup{\varphi_1,\psi_1}$ and $\tup{\varphi_2,\psi_2}$ as
above \emph{equivalent} if $\varphi_1 \equiv \varphi_2$ and
$\psi_1 \equiv \psi_2$. Let $\Theta(\rho)$ be the set of equivalence
classes under this relation and denote by $[\varphi,\psi]$ the
equivalence of a pair $\tup{\varphi,\psi}$ under this relation. Now we
fix for each $[p] \in S \setminus \set{s_0}$ a strictly guarded formula
$\chi_{[p]}$ that is equivalent to all queries from $[p]$. Likewise, we
fix a function
$\vartheta_\rho \colon \Theta(\rho) \ra \hat{\Theta}(\rho)$ such that
$\vartheta_\rho([\varphi,\psi]) \in [\varphi,\psi]$, i.e., which picks a
representative for each equivalence class $[\varphi,\psi]$.

Now let $\rho \in \Gamma_{\sche{S},l}$. Specify $\delta(\cdot,\rho)$ as
follows:
\begin{enumerate}
  \item For the initial state $s_0$, set
  \begin{align*}
    \delta(s_0, \rho) \coloneqq \bigvee\set{\bigwedge_{i = 1}^n[\alpha_i] \land \bigwedge_{i = 1}^m [p_i] \mid \tup{\alpha_1 \land \cdots \land \alpha_n, p_1 \land \cdots \land p_m} \in \vartheta_\rho(\Theta(\rho))}.
  \end{align*}
  Intuitively, the automaton selects a squid decomposition where its
  components are instantiated by names occurring in the root node of the
  input tree. The automaton tries to verify the single compartments of
  the squid decomposition, i.e., it tries to match them to the chase
  expansion of the input database under $\Sigma$.

\item Let $[p] \in S \setminus \set{s_0}$. We define $\delta([p],\rho)$
  according to a case distinction:
\begin{enumerate}
\item Suppose that $p \equiv \top$. Then
  $\delta([p],\rho) \coloneqq \ptrue$.
\item Suppose
  $\chi_{[p]} = \exists \ve{y}\,(\alpha(\ve{a},\ve{y}) \land \varphi)$,
  where $\alpha(\ve{a},\ve{y})$ is an atomic formula (including
  equality), $\free{\varphi} \subseteq \set{\ve{y}}$, and $\ve{a}$
  exhausts all names occuring in $\alpha$. If
  $\set{\ve{a}} \not\subseteq \names{\rho}$ then
  $\delta([p],\rho) \coloneqq \pfalse$. Otherwise,
  \begin{align*}
    \delta([p],\rho) \coloneqq &\bigvee\set{[\varphi(\ve{y}/\ve{b})] \mid \rho \models \alpha(\ve{a},\ve{b}), \set{\ve{b}} \subseteq \names{\rho}} \lor \Diamond [p]\ \lor \bigvee \mathrm{impl}(p, \rho),
  \end{align*}
  where
  \setlength{\jot}{2pt}
  \begin{align*}
      \mathrm{impl}(p,\rho) \coloneqq \set{[p_1] \land \cdots \land [p_n] \mid\ &[p_1],\ldots,[p_n] \in S \setminus \set{s_0}, \set{\ve{b}} \subseteq \names{\rho}, \\
& p_1 \land \cdots \land p_n \equiv q, \\
&\Sigma \models \forall \ve{x},\ve{y}\,(q(\ve{a}/\ve{x},\ve{b}/\ve{y}) \limpl p(\ve{a}/\ve{x}))}.
  \end{align*}
\end{enumerate}
We provide some intuitive explanation for this second case.
\begin{enumerate}
\item If $p$ is the
empty query, it can be satisfied at any input node and, hence, the
automaton accepts unconditionally on this computation branch.
\item Otherwise, we first inspect the strictly guarded formula
  $\chi_{[p]}$ at hand. If the names occurring in the guard
  $\alpha(\ve{a},\ve{y})$ are not present at the current node, it
  rejects. Otherwise, it tries to satisfy $\alpha(\ve{a},\ve{y})$ with
  all possible assignments for $\ve{y}$ at the current node and then
  proceed in state $[\varphi(\ve{y}/\ve{b})]$.
  Apart from these possibilities, the automaton can decide to move to
  any neighboring node (i.e., the parent or a child) while remaining in
  state $[p]$. This amounts to an exhaustive search of the input tree
  that tries to satisfy $p$ in the input tree.
Furthermore, the automaton may choose to construct derivation trees for
$p$. There, it uses the information provided by $\Sigma$ in order to
find strictly acyclic queries $p_1,\ldots,p_n$ that imply
$p$. Consequently, it tries to proceed its search with
$[p_1],\ldots,[p_n]$.
\end{enumerate}
\end{enumerate}

\medskip We shall now briefly comment on the running time needed to
construct $\fk{A}_{Q,l}$. The interesting part of the construction
concerns the transition function $\delta$, in particular point 2.(b)
involving $\mathrm{impl}(p,\rho)$. We have seen that in the proofs of
Lemma~\ref{lem:derivtreeatom} and Lemma~\ref{prop:derivtreequery} that
there are double-exponentially many candidates for the query
$q(\ve{a}/\ve{x},\ve{b}/\ve{y})$ that (possibly) implies
$p(\ve{a}/\ve{x})$ under $\Sigma$. Furthermore,
$q(\ve{a}/\ve{x},\ve{b}/\ve{y})$ consists of at most exponentially many
atoms. Each check whether such a query $q$ at hand implies $p$ requires
at most double-exponential time in the size of $p$. This follows from
the well-known fact that checking query implication under a set of
guarded rules is feasible in \textsc{2}\EXP\ with respect to the size of
the right-hand side query, and in polynomial time with respect to the
size of the left-hand side query (cf.~\cite{CaGK08a}), i.e., the \emph{data
complexity} of query answering under guarded tgds is polynomial time.
\end{proofcustom}

%%% Local Variables:
%%% fill-column: 72
%%% TeX-PDF-mode: t
%%% TeX-debug-bad-boxes: t
%%% TeX-master: "main.tex"
%%% TeX-parse-self: t
%%% TeX-auto-save: t
%%% reftex-plug-into-AUCTeX: t
%%% End:

\section*{PROOFS OF SECTION~\ref{sec:different-languages}}

\subsection*{Proof of Theorem~\ref{th:lhsguarded}}

A proof sketch is given in the main body of the paper. However, the fact
that $\cont((\class{G},\class{CQ}),(\class{S},\class{CQ}))$ is in
2\EXP~deserves a formal proof. Recall that to establish the latter
result we need a more refined complexity analysis of the problem of
deciding whether a guarded OMQ is contained in a UCQ; this is discussed
in the main body of the paper. In fact, it suffices to show the
following result.
As in the previous section, we focus on constant-free tgds and CQs, but
all the results can be extended to the general case at the price of more
involved definitions and proofs. Moreover, we assume that tgds have only
one atom in the head.
Recall that we write $\mi{var}_{\geq 2}(q)$ for the variables of $q$
that appear in more than one atom, and we also write $\mi{var}_{=1}(q)$
for the variables of $q$ that appear only in one atom. Then:

\begin{proposition}\label{pro:guarded-into-cq}
Consider $Q \in (\class{G},\class{BCQ})$ and a Boolean CQ $q$. The problem of deciding whether $Q \subseteq q$ is feasible in
\begin{enumerate}
\item double-exponential time in $(||Q|| + |\mi{var}_{\geq 2}(q)|)$; and
\item exponential time in $|\mi{var}_{=1}(q)|$.
\end{enumerate}
\end{proposition}

It is easy to verify that the above result, together with the algorithm
devised in the main body of the paper, implies that
$\cont((\class{G},\class{CQ}),(\class{S},\class{CQ}))$ is in 2\EXP. The
rest of this section is devoted to show the above proposition.
Our crucial task is, given a CQ $q$, to devise an automaton that accepts
consistent labeled trees which correspond to databases that make $q$
true.

\begin{lemma}\label{lem:automaton-for-cqs}
Let $q$ be a Boolean CQ over $\sche{S}$. There is a 2WAPA $\fk{A}_{q,l}$, where $l > 0$, that accepts a consistent $\Gamma_{\insS,l}$-labeled tree $t$ iff $\dec{t} \models q$. The number of states of $\fk{A}_{q,l}$ is exponential in $|\mi{var}_{\geq 2}(q)|$ and polynomial in $(|\mi{var}_{=1}(q)| + \arity{\insS} + l)$. Furthermore, $\fk{A}_{q,l}$ can be constructed in exponential time.
\end{lemma}

%\begin{lemma}
%  Let $q$ be a Boolean CQ over $\sche{S}$ and $l > 0$. There is a 2WAPA
%  $\fk{A}_{q,l} = \tup{S, \Gamma_{\sche{S},l},\delta,s_0,\Omega}$ such
%  that for any consistent $\Gamma_{\sche{S},l}$-labeled tree $t$ it
%  holds that
%  \begin{align*}
%    \dec{t} \models q \iff \text{$\fk{A}_{q,l}$ accepts $t$.}
%  \end{align*}
%  Furthermore, the number of states of $\fk{A}_{q,l}$ depends
%  polynomially on $\mi{var}_{\leq 1}(q)$, $\arity{\sche{S}}$, and $l$,
%  while it depends exponentially on $\mi{var}_{\geq 2}(q)$.
%\end{lemma}

\begin{proof}
We are going to construct $\fk{A}_{q,l} = \tup{S, \Gamma_{\sche{S},l},\delta,s_0,\Omega}$.
  Let $x_1,\ldots,x_n$ be the variables of $\mi{var}_{=1}(q)$
  and fix a total order $x_1 \prec x_2 \prec \cdots \prec x_n$ among
  them. Define the state set $S$ to be
  \begin{align*}
    S \coloneqq \set{s_{y,\theta} \mid \theta \colon V \ra U_{\sche{S},l}, V \subseteq \mi{var}_{\geq 2}(q), y \in \mi{var}_{=1}(q) \cup \set{\sharp}}.
  \end{align*}
  Notice that $|S| = O(|\mi{var}_{=1}(q)|\cdot(\arity{\sche{S}} + l)^{|\mi{var}_{\geq 2}(q)|})$. We set $s_0 \coloneqq s_{\sharp,\emptyset}$, where $\emptyset$ denotes the empty substitution. In the following, we treat
  $q$ as a set of relational atoms and let $X = \mi{var}_{\geq 2}(q)$. For $\rho \in \Gamma_{\sche{S},l}$ and $s_{y,\theta} \in S$, define $\delta(s_{y,\theta},\rho)$ as
  follows:
  \begin{itemize}
    \item If $y = \sharp$, distinguish the following cases:
    \begin{enumerate}
    \item If there is an atom $\alpha \in \theta(q)$ such that
      $\var{\alpha} \cap X \neq \emptyset$ and
      $\adom{\alpha} \cap U_{\sche{S},l} \not\subseteq \names{\rho}$,
      then $\delta(s_{\sharp,\theta},\rho) \coloneqq \pfalse$.
    \item Otherwise, let
      \begin{align*}
        \delta(s_{\sharp,\theta},\rho) \coloneqq\begin{cases}
          \bigvee\set{s_{\sharp, \eta} \mid \eta \supseteq \theta, \rho \models \exists \ve{x} \bigwedge (\eta(q) \setminus \theta(q))} \lor \Diamond s_{\sharp,\theta}, & \text{if $X  \cap \var{\theta(q)} \neq \emptyset$,}\\
          s_{x_1,\theta}, & \text{otherwise.}
        \end{cases}
      \end{align*}
    \end{enumerate}
  \item Suppose $y = x_i$, for some $i = 1,\ldots,n$. Let
    $\alpha_{i,\theta}$ denote the unique atom $\alpha \in \theta(q)$ such
    that $x_i \in \var{\alpha}$. Set
    \begin{align*}
      \delta(s_{x_i,\theta},\rho) \coloneqq
      \begin{cases}
        s_{x_{i+1},\theta}, & \text{if $\rho \models \exists\ve{x}\, \alpha_{i,\theta}$ and $i < n$,}\\\
        \ptrue, & \text{if $\rho \models \exists\ve{x}\, \alpha_{i,\theta}$ and $i = n$},\\
        \Diamond s_{x_i,\theta}, & \text{otherwise.}
      \end{cases}
    \end{align*}
  \end{itemize}
  Set the parity condition $\Omega$ to be $\Omega(s) \coloneqq 1$ for
  all $s \in S$. Intuitively, the automaton works in two passes. The
  first pass consists of the runs working on states of the form
  $s_{\sharp,\theta}$. In this pass, the automaton tries to find an
  assignment for the variables in the query that appear in at least two
  distinct atoms. When a candidate assignment $\theta$ is found, the
  automaton changes to state $s_{x_1,\theta}$ which is the beginning of
  the second pass. A state of the form $s_{x_i,\theta}$ means that the
  assignment $\theta$ can be extended to all variables $x \prec x_i$
  and, in this state, the automaton tries to extend $\theta$ to cover
  the variable $x_i$. The automaton accepts if it is able to extend the
  candidate assignment $\theta$ to all $x_1,\ldots,x_n$.
\end{proof}

Having the above result in place, we can now reduce the problem in question to the emptiness problem for 2WAPA.

\begin{lemma}\label{lem:reduction-to-emptiness-for-cqs}
  Consider $Q \in (\class{G},\class{BCQ})$ and a Boolean CQ $q$. We can
  construct in double-exponential time in $||Q||$ and in exponential
  time in $||q||$ a 2WAPA $\fk{A}$, which has exponentially many states
  in $(||Q|| + |\mi{var}_{\geq 2}(q)|)$ and polynomially many states in
  $|\mi{var}_{=1}(q)|$, such that
\[
Q \subseteq q \iff \ca{L}(\fk{A}) = \emptyset.
\]
\end{lemma}

\begin{proof}
  Let $Q = (\insS,\dep,q')$ and
  $l = \arity{\insS \cup \sch{\dep}} \cdot |q'|$. Then $\fk{A}$ is
  defined as:
\[
(\fk{C}_{\insS,l} \cap \fk{A}_{Q,l}) \cap \overline{\fk{A}_{q,l}}.
\]
It is an easy task to verify that the claim follows from Lemmas~\ref{lem:consistent-labeled-trees}, \ref{lem:automaton-1}, \ref{lem:automaton-2}, and~\ref{lem:automaton-for-cqs}.
\end{proof}

It is clear that Proposition~\ref{pro:guarded-into-cq} is an easy consequence of Lemma~\ref{lem:reduction-to-emptiness-for-cqs}.

%%% Local Variables:
%%% fill-column: 72
%%% TeX-PDF-mode: t
%%% TeX-debug-bad-boxes: t
%%% TeX-master: "main.tex"
%%% TeX-parse-self: t
%%% TeX-auto-save: t
%%% reftex-plug-into-AUCTeX: t
%%% End: 
\section*{PROOFS OF SECTION~\ref{sec:applications}}

Recall that we focus on unary and binary predicates.
%Recall that we are interested in guarded OMQs over unary and binary predicates. %Hence, all schemas, queries, tgds, etc.~occurring in this section will obey this restriction. In particular, we assume that guarded tgds are of the form
%$\varphi(\ve{x},\ve{y}) \limpl \exists\ve{z}\,\alpha(\ve{x},\ve{z})$,
%where $\alpha(\ve{x},\ve{z})$ is a relational atom. Furthermore, as
%before, we assume that the rules and queries at hand are assumed to
%contain no constants.
%
Moreover, we consider constant-free tgds and CQs, and we assume that tgds have only one atom in the head.

\subsection*{Proof of Proposition~\ref{prop:fosemantic}}

\textbf{Basics.} Let $D$ be a $C$-tree of width two. We say that a tree
decomposition $\delta = \tup{\ca{T},\tup{X_t}_{t \in T}}$ witnessing
that $D$ is a $C$-tree is \emph{lean}, if it satisfies the following
conditions:
\begin{itemize}
\item The elements from $\adom{C}$ occur only in the root of $\ca{T}$
  and its immediate successors.
\item If $w$ is a child of $v$ in $\ca{T}$, then there are unique
  $c,d \in \adom{D}$ such that $X_v \cap X_w = \set{c}$ and
  $X_w \setminus X_v = \set{d}$. The element $d$ is called \emph{new at
    $w$}.
\item It follows from the previous item that every node
  $v \neq \varepsilon$ in $\ca{T}$ has a unique new element
  $c \in \adom{D}$. We additionally require that $c$ appears in the bag
  of each child of $v$.
\end{itemize}
Intuitively, $C$-trees $D$ that have lean tree decomposition represent
the \emph{actual} tree structure of $D$. It is fairly straightforward to
see that every $C$-tree has a lean tree decomposition.

Recall that the \emph{Gaifman graph} of $D$ is the graph
$\ca{G}(D) = \tup{V,E}$ with $V \coloneqq \adom{D}$ and $(a,b) \in E$ if
$a$ and $b$ coexist in some atom of $D$.
%
%\begin{align*}
%  aEb\ \defequ\ \text{$a \neq b$ and $a$ and $b$ co-occur in some guarded set of $D$.}
%\end{align*}
%
Given two nodes $a,b$ from $\ca{G}({D})$, the \emph{distance from $a$ to
  $b$} in $\ca{G}(D)$, denoted $d_{\ca{G}(D)}(a,b)$, is the minimum
length of a path between $a$ and $b$, and $\infty$ if such a path does
not exist. For $a, b \in \adom{D}$, we denote by $d_\delta(a,b)$ the
minimum distance among two nodes of $\ca{T}$ that respectively have $a$
and $b$ in their bags. We call $d_\delta(a,b)$ the \emph{distance from
  $a$ to $b$ in $\delta$}.

Notice that in a tree decomposition $\delta$ witnessing that $D$ is a
$C$-tree, any element $a \in \adom{D}$, if $a$ appears in the bag of
$v$, then it occurs only at $v$, at $v$ and its children, or at $v$ and
its parent. Since furthermore the bag of the root node is uniquely
determined by $C$, each node in the tree has a uniquely determined set
of child nodes whose bags are determined by the structure of
$D$ alone. Therefore, the following two lemmas follow immediately.
\begin{lemma}
  \label{lem:gaifmandist}
  Let $\delta = \tup{\ca{T},\tup{X_t}_{t \in T}}$ be a lean tree
  decomposition witnessing that $D$ is a $C$-tree. Then
  $d_{\delta}(a, b) \leq d_{\ca{G}(D)}(a,b)$ for all
  $a, b \in \adom{D}$.
\end{lemma}
\begin{lemma}
  \label{lem:leanislean}
  Let $\delta$ and $\delta'$ be two lean tree decompositions witnessing
  that $D$ is a $C$-tree. Then $d_\delta(a,b) = d_{{\delta'}}(a,b)$ for
  all $a,b \in \adom{D}$.
\end{lemma}

In the following, we denote by $D_{\leq k}$ the subinstance of $D$
induced by the set of elements whose distance from any $a \in \adom{C}$
in \emph{any} lean tree decomposition $\delta$ is bounded by $k$. The
subinstance $D_{>k}$ is defined analogously.%\emph{mutatis mutandis}.
The \emph{branching degree} of a lean tree decomposition is the maximum
number of child nodes of any node contained in the tree of
$\delta$. Notice that two lean tree decompositions of a $C$-tree $D$
always have the same branching degree; the argument is similar as for
the two lemmas above. Hence, we can simply speak about \emph{the}
branching degree of $D$.

\medskip
\noindent
\textbf{Encodings.} Recall that a consistent
$\Gamma_{\sche{S},l}$-labeled tree $t = \tup{T,\mu}$ encodes information
on an $\sche{S}$-database $D$ \emph{and} an appropriate tree
decomposition $\delta$ of $D$. It is clear that $\dec{t}$ has a lean
tree decomposition, but it is not guaranteed that this is reflected in
$\delta$ as well. We call (the consistent) $t$ \emph{lean}, if the tree
decomposition
$\delta_t = \tup{\ca{T} \coloneqq \tup{T,E},\tup{X_v}_{v \in T}}$ is,
where $xEy$ iff $y = x \cdot i$ for some
$i \in \mbb{N} \setminus \set{0}$ and
$X_v \coloneqq \set{[v]_a \mid a \in \names{v}}$. The following is easy to prove:
\begin{lemma}
  \label{lem:automatonlean}
  There is a 2WAPA on trees $\fk{L}_{\sche{S},l}$ that accepts a
  consistent $\Gamma_{\sche{S},l}$-labeled tree iff it is lean. The
  number of states of $\fk{L}_{\sche{S},l}$ is bounded logarithmically
  in the size of $\Gamma_{\sche{S},l}$ and $\fk{L}_{\sche{S},l}$ can be
  constructed in polynomial time in the size of $\Gamma_{\sche{S},l}$.
\end{lemma}

Let $t = \tup{T,\mu}$ be a labeled tree.  The \emph{branching degree} of
a node $x \in T$ is the cardinality of
$\set{i \mid x \cdot i \in T, i \in \mbb{N} \setminus \set{0}}$; the
branching degree of $t$ is the maximum over all branching degrees of its
nodes and $\infty$ is this maximum does not exist. We also say that $t$
is \emph{$m$-ary} if the branching degree of $t$ is bounded by $m$. A
node $x \in T$ is a \emph{leaf node} of $t$ if it has branching degree
zero. The \emph{depth} of $T$ is the maximum length among the lengths of
all branches and $\infty$ if this maximum does not exist.  Let us remark
that the branching degree of the lean $\Gamma_{\sche{S},l}$-labeled tree
$t$ as defined for labeled trees equals the branching degree of
$\dec{t}$ as defined above.

\begin{lemma}
  \label{lem:semanticfoadvanced}
  Let $Q = \tup{\sche{S},\Sigma,q}$ be an OMQ from
  $\tup{\class{G}_2, \class{BCQ}}$. There is an $m \geq 0$ such that the
  following are equivalent:
  \begin{enumerate}
    \item There is an $\sche{S}$-database $D$ such that $D \models Q$.
    \item There is a $C$-tree $\hat{D}$ with $|\adom{C}| \leq 2|q|$ and
      branching degree at most $m$ such that $\hat{D} \models Q$.
  \end{enumerate}
\end{lemma}
\begin{proof}
  Let $l \coloneqq 2|q|$ and, let $\fk{A}_{Q,l}$ be
  the 2WAPA from Lemma~\ref{lem:automaton-2}. Take the intersection
  of $\fk{A}_{Q,l}$ with
  \begin{enumerate*}[label={(\roman*)}]
    \item the 2WAPA $\fk{C}_{\sche{S},l}$ from Lemma~\ref{lem:automaton-1} and
    \item the 2WAPA from Lemma~\ref{lem:automatonlean} that checks
      leanness.
  \end{enumerate*}
  Call the resulting automaton $\fk{B}$. Then $\fk{B}$ accepts a
  $\Gamma_{\sche{S},l}$-labeled tree $t$ iff $t$ is lean and consistent
  and $\dec{t} \models Q$. We let $m$ be the number of states of
  $\fk{B}$ and claim that this is the required bound on the branching
  degree.

  First of all, notice that the first item of the lemma trivially
  implies the second independently from the choice of $m$. For the other
  direction, suppose that $D \models Q$ for some
  $\sche{S}$-database $D$. Then there is a $C$-tree $B$ such that
  $\adom{C} \leq 2|q|$ and $B \models Q$. Being a $C$-tree,
  $B$ has a lean tree decomposition $\delta$ and the encoding of $B$
  together with $\delta$ corresponds to a lean
  $\Gamma_{\sche{S},l}$-labeled tree $t$. It follows that
  $t \in \ca{L}(\fk{B})$. By the results of~\cite{GrWa99}, it follows
  that there is a $t' \in \ca{L}(\fk{B})$ whose branching degree is
  bounded by the number of states of $\fk{B}$, i.e., by $m$. The tree
  $t'$ is lean and consistent, therefore $\dec{t'}$ is a $C'$-tree of
  branching degree at most $m$ for some $C' \subseteq \dec{t}$ such that
  $|\adom{C'}| \leq 2|q|$.  Furthermore, $\dec{t'} \models Q$,
  as required.
\end{proof}

We are now ready to prove Proposition~\ref{prop:fosemantic}:

\medskip

\begin{proofcustom}{of Proposition~\ref{prop:fosemantic}}
  We largely follow~\cite{BLW16} here. Choose $m$ as in
  Lemma~\ref{lem:semanticfoadvanced} above. Suppose first that $Q$ is
  UCQ rewritable. Let $p \coloneqq p_1 \lor \cdots \lor p_n$ be a
  corresponding UCQ rewriting. Since the query $q$ is connected, we can
  assume that $p$ is as well. We choose
  $k > \max\set{|p_i| : i = 1,\ldots,n}$ and suppose that $D \models Q$
  for some $C$-tree $D$. Since $p$ is a UCQ rewriting, $D \models p_i$
  for some $i = 1,\ldots,n$. Fix a homomorphism $\mu$ witnessing that
  $D \models p_i$. We distinguish cases. Suppose first that
  $\mu(\var{p_i}) \cap \adom{C} \neq \emptyset$. Since $p$ is connected,
  it follows $D_{\leq k} \models p_i$ by Lemma~\ref{lem:gaifmandist} and
  so $D_{\leq k} \models p$. On the other hand, if
  $\mu(\var{p}) \cap \adom{C} = \emptyset$, then it is also easy to
  check that $D_{> 0} \models p$.

  For the other direction, suppose that the second item of the
  proposition's statement holds, i.e., there is a $k \geq 0$ such that
  for all $C$-trees $D$ over $\sche{S}$ with $|\adom{C}| \leq 2|q|$ and
  branching degree at most $m$ it holds that $D \models Q$ implies
  $D_{\leq k} \models Q$ or $D_{>0} \models Q$. Let $\Lambda$ be the set
  of all $C$-trees such that $|\adom{C}| \leq 2|q|$ and that have
  branching degree at most $m$ such that $D \models Q$. We regard
  $\Lambda$ as a set of BCQs and regard it as factorized modulo logical
  equivalence. It is clear that $\Lambda$ is finite then and we claim
  that $p \coloneqq \bigvee_{i = 1}^n p_i$ is a UCQ rewriting of
  $Q$. We explicitly include the case where $\Lambda$ is empty, in which
  case $p$ is equivalent to the empty disjunction $\bot$ and there is no
  database $D$ at all such that $D \models Q$. To see that $p$ is indeed
  a UCQ rewriting of $Q$, let $D$ be an $\sche{S}$-database such that
  $D \models p$. Then there is an $i = 1,\ldots,n$ such that
  $D \models p_i$. Furthermore, $[p_i] \models Q$ and so $D \models Q$
  as well, since $Q$ is closed under homomorphisms.
  Suppose now $D \models Q$. We know that there is a $C$-tree $\hat{D}$
  with $|\adom{C}| \leq l \coloneqq 2|q|$ and branching degree at most
  $m$ such that $\hat{D} \models Q$ and---when we regard $\hat{D}$ as an
  instance---there is a homomorphism from $\hat{D}$ to $D$. Let
  $D' \subseteq \hat{D}$ be a minimal connected subset of $\hat{D}$ such
  that $D' \models Q$. $D'$ is again a $C'$-tree for some
  $C' \subseteq D'$. Therefore $D'_{\leq k} \models Q$ or
  $D'_{>0} \models Q$. The latter is impossible by minimality of
  $D'$. Hence, $D'_{\leq k} \models Q$ and so there is a (logically
  equivalent) copy of $D'_{\leq k}$ contained in $\Lambda$. Hence,
  $D'_{\leq k} \models p$, therefore $D' \models p$, and hence
  $\hat{D} \models p$. Recall that, when $\hat{D}$ is regarded as an
  instance, there is a homomorphism from $\hat{D}$ to $D$. Therefore,
  $D \models p$.
\end{proofcustom}

\subsection*{Proof of Proposition~\ref{pro:ucqrew-to-infinity}}

Let $Q = \tup{\sche{S},\Sigma,q}$ be an OMQ from
$\tup{\class{G}_2,\class{BCQ}}$ such that $q$ is connected. We are going
to show that the desired 2WAPA $\fk{A}$ can be constructed in
\twoexptime. Notice that, using similar results as in~\cite{BLW16}, this
gives us a decision procedure for deciding
$\rew\tup{\class{G}_2,\class{CQ}}$ also for non-connected
queries. Let us first introduce some auxiliary notions.

%Toward a proof of Proposition~\ref{pro:ucqrew-to-infinity}, we will introduce some additional notions first.

\medskip
\noindent
\textbf{2WAPAs on $m$-ary trees.} A 2WAPA $\fk{B}$ \emph{on $m$-ary}
trees is just defined as a 2WAPA, except that its transitions
$\mathsf{tran}(\fk{B})$ are
$\set{\ndia{k}s,\nbox{k}s \mid -1 \leq k \leq m, s \in S}$, where $S$ is
the state set of $\fk{B}$. The notion of run is then defined on $m$-ary
trees only and its definition is modified in the obvious way so as to
deal with the transitions $\ndia{k}s, \nbox{k}s$. Intuitively, for
$k = 1,\ldots,m$, a transition $\ndia{k}s$ means that the automaton
should move to the $k$-th child of the current node (which is then
required to exist) and assume state $s$. Correspondingly, $\nbox{k}s$
means that the automaton should move to the $k$-th child and assume
state $s$ provided that this $k$-th child exists at all. We remark that
all 2WAPAs constructed in this paper so far can easily be modified to
work on $m$-ary trees as well and we shall assume in the following that
they do so. Furthermore, deciding whether $\ca{L}(\fk{B})$ is feasible
in exponential time in the number of states of $\fk{B}$ and in
polynomial time in the size of the input alphabet of $\fk{B}$
(cf.~\cite{Vardi98}).

Let $m$ be as in Proposition~\ref{prop:fosemantic}. In the following, we
shall regard all trees mentioned in the following as $m$-ary and let
$l \coloneqq 2|q|$. Before proceeding to a proof of
Proposition~\ref{pro:ucqrew-to-infinity}, we must make the notion of
being an ``extension'' of a labeled tree more precise.

\medskip
\noindent
\textbf{Extensions of trees.}  Let $\fk{B}_Q$ be a 2WAPA that accepts a
$\Gamma_{\sche{S},l}$-labeled tree $t$ iff
\begin{enumerate*}[label={(\roman*)}]
\item $t$ is lean and consistent,
\item $\dec{t} \models Q$, and
\item $\dec{t}_{> 0} \not\models Q$.
\end{enumerate*}
Notice that a 2WAPA $\fk{A}_Q^{>0}$ that accepts a lean and consistent
$\Gamma_{\sche{S},l}$-labeled tree iff $\dec{t}_{>0} \not\models Q$ can
be easily constructed using the construction in
Lemma~\ref{lem:automaton-2}. Hence, $\fk{B}_Q$ can be constructed
intersecting several 2WAPAs we have already encountered.

Let $\Pi$ be the set of all tuples of the form $\tup{s,s'}$, where $s$
and $s'$ are states of $\fk{B}_Q$. We define a new alphabet
$\Lambda \coloneqq 2^{\mbb{K}_{\sche{S},l} \cup \Pi}$. Notice that
$\Lambda$ is of double-exponential size in the size of $Q$. For
$\rho \in \Lambda$, we denote by
$\rho \upharpoonright \Gamma_{\sche{S},l}$ the \emph{restriction of
  $\rho$ to $\Gamma_{\sche{S},l}$}, that is,
$\rho \cap \mbb{K}_{\sche{S},l}$.  The \emph{restriction} of a
$\Lambda$-labeled tree $t$ to $\Gamma_{\sche{S},l}$, denoted
$t \upharpoonright \Gamma_{\sche{S},l}$, is the tree that arises from
$t$ when we restrict the label of each node of $t$ to
$\Gamma_{\sche{S},l}$. We say that a $\Lambda$-labeled tree is
\emph{consistent} if
\begin{enumerate*}[label={(\roman*)}]
 \item its restriction to $\Gamma_{\sche{S},l}$ is
consistent and
\item symbols $\rho \in \Lambda$ such that
  $\rho \cap \Pi \neq \emptyset$ appear only in leaf nodes of $t$.
\end{enumerate*}
Likewise, we say that a consistent $t$ is \emph{lean} if
$t \upharpoonright \Gamma_{\sche{S},l}$ is. The decoding $\dec{t}$ of
$t$ is naturally extended to consistent $\Lambda$-labeled trees by
setting $\dec{t} \coloneqq \dec{t \upharpoonright \Gamma_{\sche{S},l}}$.
The following lemma is a straightforward extension of
Lemmas~\ref{lem:automaton-1} and~\ref{lem:automatonlean}.

\begin{lemma}
  \label{lem:conslean}
  There are 2WAPAs $\fk{C}_\Lambda$ and $\fk{L}_\Lambda$ that
  respectively accept a $\Lambda$-labeled tree iff it is consistent and
  lean. Both have logarithmically many states in the size of $\Lambda$
  and can be constructed in polynomial time in the size of $\Lambda$.
\end{lemma}

Let $t$ be a lean and consistent $\Lambda$-labeled tree. We say that
$t'$ is an \emph{extension} of $t$ if $t'$ is a
$\Gamma_{\sche{S},l}$-labeled tree that arises from $t$ by attaching
$\Gamma_{\sche{S},l}$-labeled trees to those leaves of $t$ that contain
elements from $\Pi$. Furthermore, for such nodes, the labels of the
corresponding nodes in $t'$ are those of $t$ restricted to
$\Gamma_{\sche{S},l}$.

\begin{definition}
  Let $\ca{L}_Q$ be the set of all lean and consistent $\Lambda$-labeled
  trees $t$ such that $\dec{t} \not\models Q$, yet there is an
  extension $t'$ of $t$ such that $\dec{t'} \models Q$ and
  $\dec{t'}_{>0} \not\models Q$.\hfill\markfull
\end{definition}

\begin{lemma}
  \label{lem:infinitelqiffucqrew}
  $\ca{L}_Q$ is infinite iff $Q$ is not UCQ rewritable.
\end{lemma}
\begin{proof}
  Suppose $\ca{L}_Q$ is infinite. Since the trees at hand have bounded
  branching degree, for every $k \geq 0$, there is a $t \in \ca{L}_Q$
  such that $\dec{t}$ is a $C$-tree (for some $C \subseteq \dec{t}$)
  that contains individuals whose distance from any $a \in \adom{C}$ is
  greater than or equal to $k$ and $\dec{t} \not\models Q$, yet for some
  extension $t'$ of $t$, we have $\dec{t'} \models Q$ but
  $\dec{t'}_{>0} \not\models Q$. Suppose now that $Q$ is UCQ
  rewritable. Let $\ell$ be such that for all $C'$-trees $D$ (of the
  appropriate dimensions), $D \models Q$ implies
  $D_{\leq \ell} \models Q$ or $D_{>0} \models Q$. Choose $k > \ell$ and
  $t,t'$ such that
  \begin{enumerate*}[label={(\roman*)}]
    \item $t'$ is an extension of $t$,
    \item $t$ has depth greater than $k$, and
    \item $\dec{t} \not\models Q$ but $\dec{t'} \models Q$ and
      $\dec{t'}_{>0} \not\models Q$.
  \end{enumerate*}
  Since $\dec{t'} \models Q$, we know that
  $\dec{t'}_{\leq \ell} \models Q$ or $\dec{t'}_{> 0} \models Q$. The
  latter is impossible by assumption, the former contradicts the fact
  $\dec{t} \not\models Q$, since $k > \ell$. This proves the direction
  from left to right. The other direction is immediate.
\end{proof}

We are now ready to establish Proposition~\ref{pro:ucqrew-to-infinity}:

\medskip
\begin{proofcustom}{of Proposition~\ref{pro:ucqrew-to-infinity}}
  We are now going to describe the construction of a 2WAPA $\fk{A}$ such
  that $\ca{L}(\fk{A}) = \ca{L}_Q$, which will prove the claim by virtue
  of Lemma~\ref{lem:infinitelqiffucqrew}. This automaton is the
  intersection of several ones. First of all, we ensure that all the
  accepted $\Lambda$-trees are lean and consistent
  (cf.~Lemma~\ref{lem:conslean}). We additionally intersect the
  automaton with the complement of $\fk{A}_{Q,l}$ from
  Lemma~\ref{lem:automaton-1} (more precisely, the version of it running
  on $\Lambda$-labeled trees) and another automaton
  $\fk{D}_Q = \tup{S,\Lambda, \delta, s_0, \Omega}$ whose construction we
  shall describe in more detail here.
On a high level, $\fk{D}_Q$ will be constructed so as to accept a lean
and consistent $\Lambda$-labeled tree if and only if there is an
extension $t'$ of $t$ such that $\fk{B}_Q$ accepts $t'$. Let $\hat{S}$
be the set of states of $\fk{B}_Q$, $\hat{\delta}$ its transition
function, and $\hat{\Omega}$ its parity function. For
$\sigma \in \Gamma_{\sche{S},l}$, let $\hat{B}(\sigma)$ be the set of
tuples
$\tup{s,s'} \in \hat{S} \times (\hat{S} \cup \set{\ptrue})$ such that the following holds:
\begin{itemize}
\item There is a $\Gamma_{\sche{S},l}$-labeled tree $t = \tup{T,\eta}$
  such that $\eta(\varepsilon) = \sigma$ and a run
  $t_r = \tup{T_r,\eta_r}$ of $\fk{B}_Q$ on $t$ such that
\begin{enumerate}
\item $\eta_r(\varepsilon) = \tup{\varepsilon,s}$, i.e., $t_r$
  \emph{starts from $s$};\footnote{Strictly speaking, $t_r$ is, of
    course, not a run since it does not start in the initial state.}
  \item $s' = \ptrue$ and $t_r$ is accepting on $\fk{B}_Q$, or there is
    a node $v \in T_r$ such that $\eta_r(v) = \tup{\varepsilon,s'}$.
\end{enumerate}
\end{itemize}
Now the set of states of $\fk{D}_Q$ is the same as of $\fk{B}_Q$, i.e.,
$S \coloneqq \hat{S}$. Accordingly, the initial state of $\fk{D}_Q$ is
that of $\fk{B}_Q$. Furthermore ${\Omega}(s) \coloneqq \hat{\Omega}(s)$,
for every $s \in S$. Given $s \in S$ and $\rho \in \Lambda$, we let
\begin{align*}
  \delta(s,\rho) \coloneqq \bigvee \set{s' \mid \tup{s, s'} \in \rho \cap \hat{B}(\Gamma_{\sche{S},l} \upharpoonright \rho)} \lor \hat{\delta}(s,\Gamma_{\sche{S},l} \upharpoonright \rho).
\end{align*}
We are going to give an intuitive explanation of this construction in
the following. Roughly, a pair $\tup{s,s'} \in \hat{B}(\sigma)$
indicates that there is a $\Gamma_{\sche{S},l}$-labeled tree $t$ and run
of $\fk{B}_Q$ on $t$ such that the root of $t$ is labeled with $\sigma$,
the run starts in state $s$, and either $\fk{B}_Q$ accepts $t$, or it
traverses the root again at some point, then being in state $s'$. The
set $\hat{B}(\sigma)$ can be computed a priori in \twoexptime;
considering that $\Gamma_{\sche{S},l}$ is of double-exponential size in
the size of $Q$, it follows that the collection
$\set{\hat{B}(\sigma)}_{\sigma \in \Gamma_{\sche{S},l}}$ can be computed
in \twoexptime. Now the input tree for $\fk{D}_Q$ comes with labels from
$\Pi$ of the form $\tup{s,s'}$ in its leaves. These ``types'' amount to
guesses of possible extensions of the input tree. Utilizing the sets
$\hat{B}(\sigma)$, $\fk{D}_Q$ thus explores the possible ways how the
given input tree can be extended to a $\Gamma_{\sche{S},l}$-labeled tree
$t'$ that is accepted by $\fk{B}_Q$.
\end{proofcustom}

%%% Local Variables:
%%% fill-column: 72
%%% TeX-PDF-mode: t
%%% TeX-debug-bad-boxes: t
%%% TeX-master: "main.tex"
%%% TeX-parse-self: t
%%% TeX-auto-save: t
%%% reftex-plug-into-AUCTeX: t
%%% End: 

\balance

\bibliographystyle{abbrv}
%\bibliography{references}

\end{document}